\documentclass{amsart}
\usepackage{amsmath}%
\usepackage{amssymb}
\usepackage{newunicodechar}
\usepackage{amsthm}
\usepackage{mathtools}
\usepackage{physics}
\usepackage{xfrac}
\usepackage{amsfonts}%
\usepackage{dsfont}
\usepackage{graphicx}
\usepackage{tikz-cd}
\usepackage[utf8]{inputenc}
\usepackage{indentfirst}
\usepackage{enumerate}
\usepackage{setspace}
\usepackage{verbatim}
\usepackage{yfonts}
\usepackage{cite}
\usepackage{url}
\usepackage{hyperref}
\usepackage{ragged2e}

\DeclareFontFamily{U}{mathx}{\hyphenchar\font45}
\DeclareFontShape{U}{mathx}{m}{n}{
      <5> <6> <7> <8> <9> <10>
      <10.95> <12> <14.4> <17.28> <20.74> <24.88>
      mathx10
      }{}
\DeclareSymbolFont{mathx}{U}{mathx}{m}{n}
\DeclareFontSubstitution{U}{mathx}{m}{n}
\DeclareMathAccent{\widecheck}{0}{mathx}{"71}
\DeclareMathAccent{\wideparen}{0}{mathx}{"75}

\theoremstyle{plain}
\newtheorem{theorem}{Theorem}[section]
\newtheorem{definition}[theorem]{Definition}

\newtheorem{corollary}[theorem]{Corollary}

\newtheorem{example}[theorem]{Example}

\newtheorem{lemma}[theorem]{Lemma}
\newtheorem{notation}[theorem]{Notation}

\newtheorem{proposition}[theorem]{Proposition}
\newtheorem{remark}[theorem]{Remark}

\numberwithin{equation}{section}
\newenvironment{skproof}{\begin{proof}}{\end{proof}}

\newcommand{\cg}[6]{C{\mathstrut}^{#1,}_{#4,}{\mathstrut}^{#2,}_{#5,}{\mathstrut}^{#3}_{#6}}

\begin{document}

\title[On symbol correspondences for quark systems II]{On symbol correspondences for\\ quark systems II: Asymptotics}
\author{P. A. S. Alcântara}
\author{P. de M. Rios}
\address{Instituto de Ci\^encias Matem\'aticas e de Computa\c{c}\~ao, Universidade de S\~ao Paulo. \newline
S\~ao Carlos, SP, Brazil.}
\email{pedro.antonio.alcantara@usp.br}
\email{prios@icmc.usp.br}
\thanks{This work was supported in part by Coordenação de Aperfeiçoamento de Pessoal de Nível Superior (CAPES), Brasil - Finance Code 001.}



\begin{abstract}
We study the semiclassical asymptotics of twisted algebras induced by symbol correspondences for quark systems ($SU(3)$-symmetric mechanical systems) as defined in our previous paper \cite{ARpreprint}. The linear span of harmonic functions on (co)adjoint orbits is identified with the space of polynomials on $\mathfrak{su}(3)$ restricted to these orbits, and we find two equivalent criteria for the asymptotic emergence of Poisson algebras from sequences of twisted algebras of harmonic functions on (co)adjoint orbits which are induced from sequences of symbol correspondences (the fuzzy orbits). Then, we proceed by ``gluing'' the fuzzy orbits along the unit sphere $\mathcal S^7\subset \mathfrak{su}(3)$, defining Magoo spheres, and studying their asymptotic limits. We end by highlighting the possible generalizations from $SU(3)$ to other compact symmetry groups, specially compact simply connected semisimple Lie groups, commenting on some peculiarities from our treatment for $SU(3)$ deserving further investigations.
\end{abstract}
\maketitle

\tableofcontents

\allowdisplaybreaks

\section{Introduction}\label{sec:intro}

This is Paper II of two serial works on correspondences for quark systems, i.e. mechanical systems with $SU(3)$-symmetry. Here we present the asymptotic analysis of twisted products induced by symbol correspondences over symplectic (co)adjoint orbits, as defined in \cite{ARpreprint} (henceforth referred to as Paper I), 
and address the question of how such twisted products can be extended to the unit sphere $\mathcal S^7\subset \mathfrak{su}(3)$.

Throughout this paper, we shall often recall and refer to the results of Paper I. Thus, excerpts of Paper I are cited by adding an ``I'' to the number; for instance, (I.2.71) means equation (2.71) of Paper I, likewise for Proposition I.3.5, etc.

While we worked with abstract orbits $\mathbb CP^2$ and $\mathcal E$ in Paper I, in this text we adopt a special family of actual (co)adjoint orbits in $\mathcal S^7\subset\mathfrak{su}(3)$, the orbits that are equivalent by rescaling to an orbit in $\mathfrak{su}(3)$ of highest weight $(p,q)$, $p,q\in\mathbb N$, which shall be called \emph{rational orbits}, so that this family defines a \emph{rational coarsening} of the orbit foliation of $\mathcal S^7$, cf.~Definitions \ref{ratorbdef} and \ref{def:rat_coars}. For each of these rational orbits, the sequences of correspondences suitable for semiclassical asymptotic analysis are \emph{rays of correspondences} defined on sequences of quantum systems determined by rays in the lattice of dominant weights given by the orbits themselves. 

That is, for each rational orbit $\mathcal O_\xi\subset\mathcal S^7$ if $\omega_{(p_0,q_0)}$ is the first highest weight whose orbit is equivalent to $\mathcal O_\xi$, 
we consider the sequence of highest weights $(\omega_{(sp_0,sq_0)})_{s\in\mathbb N}$, with its sequence of symbol correspondences to functions on $\mathcal O_\xi$, cf. Definition \ref{def:pois_ray}.
Then, for each irrep $(sp_0,sq_0)$, a symbol correspondence defines a  twisted algebra on a $d^2$-dimensional subspace of $C^\infty_{\mathbb C}(\mathcal O_\xi)$, where $d=\dim(sp_0,sq_0)$, which is isomorphic to the matrix algebra $M_{\mathbb C}(d)$. Hence, each $\xi$-ray of correspondences defines a sequence of twisted algebras of functions on $\mathcal O_\xi$, also called a \emph{fuzzy orbit}.

The necessity of working with sequences of (increasing) finite dimensional twisted algebras of functions on $\mathcal O_\xi$ and investigating if/when/how their asymptotic limits coincide with the classical Poisson algebra of smooth functions on  $\mathcal O_\xi$, stems from some results for $SU(3)$-invariant unital $C^\star$-algebra structures on $C^\infty(\mathcal O_\xi)$, which we state and prove, cf. Theorem \ref{theo:const_quant_cp2}, Proposition \ref{CstarE} and Corollary \ref{nogo-cor}. 

However, while the method used in \cite{RS} for studying such asymptotic limits can be generalized from spin systems to pure-quark systems, albeit with greater difficulty, its generalization to mixed-quark systems seems hopeless, so in this paper we develop a new method using the universal enveloping algebra. Also, in \cite{RS} the criterion for recovering the Poisson algebra of harmonic functions on $\mathcal S^2$ as an asymptotic limit of spin twisted algebras is more clearly seen by comparison to a suitable sequence of Stratonovich-Weyl correspondences for spin systems. However, for quark systems, specifically mixed quark systems, the characterization of Stratonovich-Weyl correspondences is quite cumbersome, cf. Remark I.5.26, so here we adopt as paradigm the sequences of (highest weight) Berezin correspondences. 

Karabegov \cite{karab} has shown in a quite general setting  that such Berezin correspondences satisfy a version of the so called \emph{correspondence principle}, which we enunciate in the context of quark system as an asymptotic ($s\to\infty$) \emph{Poisson type} property, cf. Definition \ref{PoissonOxi}. Then, we apply the results for Berezin correspondences to derive a classification of $\xi$-rays of correspondences for quark systems which are of Poisson type, stated in two different ways, cf. Theorems \ref{cor:pois_crit} and \ref{theo:pois_char_m_ber}.

Thereafter, given $\xi$-rays of correspondences with their induced  sequences of twisted algebras,  defined for each and every rational orbit $\mathcal O_\xi\subset\mathcal S^7$, we proceed by ``gluing'' all these fuzzy orbits together along the rational coarsening of $\mathcal S^7$, thus defining a \emph{Magoo sphere}, cf. Definition \ref{def:penc_magoo}. We do so by first defining a chain of nested subsets of rational orbits as a sequence indexed by $n\in\mathbb N$ which converges to the full set of rational orbits in $\mathcal S^7$. This leads to the definition of a Magoo sphere as a bi-sequence of twisted algebras and, by first taking the asymptotic limit $s\to\infty$ and then the chain limit $n\to\infty$, we arrive at the definition of Magoo spheres of Poisson type, cf. Definition  \ref{MagooAsymp-defn}, and Theorem \ref{theo:Magoo_pois} shows that this property is satisfied for a Magoo sphere if and only if every fuzzy orbit is of Poisson type. 

Inverting the order of the limits for a Magoo sphere, taking $n\to\infty$ first and then $s\to\infty$, leads to the definition of Magoo spheres of \emph{uniform Poisson type}, cf. Definition \ref{unifPoiss} and Proposition \ref{charUnifPoisson}. Thus, we end by studying if this property is satisfied for the Berezin Magoo sphere, and Theorem \ref{theo:unif_conv_compact_ber} states that this is so if we restrict to any compact ``cylinder'' $\mathcal S^7|_{\mathcal K}\subset\mathcal S^7$ which does not contain neighborhoods of the nongeneric orbits. On the other hand, in Proposition \ref{Berspecial-prop} we present an example of Magoo sphere of Poisson type for which the uniform Poisson property does not hold even in any such a ``cylinder'', showing that the Berezin Magoo sphere is special, in this sense. However, we have not yet been able to prove or disprove the uniform Poisson property for the whole Berezin Magoo sphere. 

This paper is organized as follows.

In section \ref{sec:basic_frame} we stablish some basic tools and results used throughout the paper. We describe the symplectic foliation of $\mathfrak{su}(3)$ and its unit sphere $\mathcal S^7$ by (co)adjoint orbits, and introduce the coarse Poisson sphere as the countable collection of rational orbits in $\mathcal S^7$. Then we state and prove some results on $C^\star$-algebras and discuss how they imply the necessity to work with sequences of finite-dimensional twisted algebras to study the semiclassical asymptotic limit. We 
also describe harmonic functions on orbits and on $\mathcal S^7$ as polynomial functions, resorting to an isomorphism from the universal enveloping algebra $U(\mathfrak{sl}(3))$ to $Poly(\mathfrak{su}(3))$ in order to describe the Poisson algebra of polynomials.
Then we use the pullback of symbol correspondences to  $U(\mathfrak{sl}(3))$ so that we can deal with correspondences defined on a fixed domain, which makes it easier to take asymptotic limits.

In section \ref{sec:asymp_anal}, we develop the semiclassical analysis of twisted algebras of functions on orbits (fuzzy orbits). First, we reproduce some general results of \cite{karab} in the specific setting of quark systems, and we use them to obtain two equivalent conditions for a ray of correspondences to be of Poisson type. The first criterion  is a comparison between limits of symbols and polynomials, and the second one is by means of the characteristic matrices defined in Paper I.

 Section \ref{sec:corresp_sphere} is devoted to ``gluing'' the fuzzy orbits along the coarse Poisson sphere, defining the Magoo spheres, and studying their asymptotic limits.   

 Then, in the last section \ref{sec:conc} we discuss how most of the results of both papers I and II can be generalized to other compact symmetry groups, specially to general  compact  simply connected semisimple Lie groups, and finish with last comments on peculiarities from our treatment of $SU(3)$ that deserve further investigations. 

 Finally, in Appendix \ref{proofThmBerCharNumbers} we present a proof of Proposition \ref{ThmBCN},
 and in Appendix \ref{app:pq-asymp} we summarize the Clebsch-Gordan approach to the asymptotics of twisted algebras for pure-quark systems, which is presented in full in \cite{Thesis}.

\section{Basic framework and preliminary results}\label{sec:basic_frame}
We will work with symbol correspondences for functions on concrete adjoint orbits $\mathcal O\subset\mathfrak{su}(3)$ rather than the abstract ones, $\mathbb CP^2$ or $\mathcal E$, as indicated in Remark I.3.3. Our approach shall be based on the quite general method that Karabegov applied to Berezin correspondences in \cite{karab}. We begin by establishing some definitions and notations. We refer to Paper I, Section 2, for details.

Recall that $\{E_j = i\lambda_j/\sqrt{2}:j=1,...,8\}$, cf. (I.2.1), is an orthonormal basis of $\mathfrak{su}(3)$ w.r.t.~the standard inner product (I.2.8), and the fundamental weights are
\begin{equation}\label{fundweights}
	\varpi_1 = \dfrac{1}{\sqrt{2}}E_3+\dfrac{1}{\sqrt{6}}E_8 \ , \ \ \varpi_2 = \sqrt{\dfrac{2}{3}}E_8\, ,
\end{equation}
cf. (I.2.65), so that dominant weights are of the form
\begin{equation}\label{high-weight}
	\omega_{\vb*p} = p\,\varpi_1 + q\,\varpi_2 \ , \ \ \ \vb*p=(p,q)\in \mathbb N_0\times \mathbb N_0\, ,
\end{equation}
and we identify an irreducible representation with highest weight $\omega_{\vb*p}$  by the pair $\vb*p = (p,q)$, the case $\vb*p = (0,0)$ being the trivial representation which is often discarded. 

Now, by the Stone-Weierstrass Theorem,
\begin{equation}\label{Poly(O)}
	Poly(\mathcal O) := \{f|_{\mathcal O}: f\in Poly(\mathfrak{su}(3))\}
\end{equation}
is uniformly dense in $C^\infty_{\mathbb C}(\mathcal O)$ for every orbit $\mathcal O \subset \mathfrak{su}(3)$. Since the space $Poly_d(\mathfrak{su}(3))$ of complex homogeneous polynomials on $\mathfrak{su}(3)$ of degree $d\in\mathbb N$ is an invariant subspace for the $SU(3)$-action, the linear span of harmonic functions on $\mathcal O$ is precisely $Poly(\mathcal O)$. However, although $Poly_d(\mathfrak{su}(3))$ provide a grading for the algebra of polynomials $Poly(\mathfrak{su}(3))$, its restriction to an orbit $\mathcal O \subset \mathfrak{su}(3)$, 
\begin{equation}\label{polydO}
	Poly_d(\mathcal O) := \{f|_{\mathcal O}: f\in Poly_d(\mathfrak{su}(3))\} \, , 
\end{equation}
does not provide a grading of $Poly(\mathcal O)$ because the restriction of polynomials of different degrees from $\mathfrak{su}(3)$ to $\mathcal O$ may coincide. For instance, if $(x_1,..., x_8)$ are coordinates on $\mathfrak{su}(3)$ w.r.t.~the orthonormal basis $\{E_j\}_{1\leq j\leq 8}$, then 
\begin{equation}
	\sum_{j=1}^8x_j^2\big|_{\mathcal O} \equiv 1 \ \  \ \ \forall\, \mathcal O \subset \mathcal S^7 \ , 
\end{equation}
where $\mathcal S^7\subset\mathfrak{su}(3)$ is the unitary sphere. In the same vein, there is a homogeneous cubic polynomial, associated to the cubic Casimir of $SU(3)$, that is constant along each orbit $\mathcal O \subset \mathfrak{su}(3)$, cf. Proposition \ref{prop:c3_pol}, further below. Even so, we will still make use of $Poly_d(\mathcal O)$, as well as
\begin{equation}\label{polyleqdO}
	Poly_{\le d}(\mathcal O) = \bigoplus_{m=0}^d Poly_m(\mathcal O)\, .
\end{equation}
Likewise, for the unitary sphere $\mathcal S^7\subset\mathfrak{su}(3)$,
\begin{equation}
    Poly(\mathcal S^7) := \{f|_{\mathcal S^7}:f \in Poly(\mathfrak{su}(3))\}
\end{equation}
is uniformly dense in $C^\infty_{\mathbb C}(\mathcal S^7)$, and we will also make use of the spaces 
\begin{equation}\label{polyS7} 
\begin{aligned}
    & Poly_d(\mathcal S^7) := \{f|_{\mathcal S^7}: f\in Poly_d(\mathfrak{su}(3))\} \ , \\
    & \hspace{2 em} Poly_{\le d}(\mathcal S^7) = \bigoplus_{m=0}^d Poly_m(\mathcal S^7)\, .
\end{aligned}
\end{equation}

\subsection{The smooth and the coarse Poisson spheres} 

We shall be interested in algebras of functions on $\mathcal S^7\subset \mathfrak{su}(3)$ or on orbits  $\mathcal O\subset\mathcal S^7$.

\begin{notation}\label{notationF}
Let $\overline{\mathcal F}$ be the arc of circumference given by the intersection of the unitary sphere $\mathcal S^7\subset \mathfrak{su}(3)$ with the closed principal Weyl chamber, so that $\mathcal F$ is the subset obtained by removing the endpoints. We can write the points of $\overline{\mathcal F}$ as
\begin{eqnarray}
	&\xi_{(x,y)} := \sqrt{\dfrac{3}{2}}\left(x\,\varpi_1+y\,\varpi_2\right) \equiv \dfrac{i}{\sqrt{6}}\small{\begin{pmatrix}
		2x+y & 0 & 0\\
		0 & -x+y & 0\\
		0 & 0 & -x-2y
	\end{pmatrix}}\in \overline{\mathcal F}\, ,& \label{pointsF} \\
 &\mbox{where}\quad\quad\begin{cases}
		\norm{\xi_{(x,y)}}^2=x^2+xy+y^2 = 1\\ 
        x,y\ge 0 
	\end{cases}\, ,& \label{par_xy}
\end{eqnarray}
with strict inequality in (\ref{par_xy}) for $\xi_{(x,y)}\in\mathcal F$. Given $\xi_{(x,y)}=\xi \in \overline{\mathcal F}$, we write $\mathcal O_{(x,y)}=\mathcal O_\xi\subset \mathcal S^7$ for its orbit, identifying $\overline{\mathcal F}$ with the set of unitary orbits,
\begin{equation}
    \overline{\mathcal F}\ni\xi\leftrightarrow \mathcal O_\xi\subset \mathcal S^7\subset \mathfrak{su}(3) \ .
\end{equation}
For functions on $\mathcal O_\xi$, we denote the supremum norm  by $\norm{\,}_\xi$ whereas on $\mathcal S^7$ we denote the supremum norm by $\norm{\,}_\infty$. In addition, we use the left-invariant integral on $\mathcal O_\xi$ induced by the Haar measure of $SU(3)$ to define the inner product $\ip{\cdot}{\cdot}_\xi$ as
\begin{equation}\label{Haarip}
    \ip{f_1}{f_2}_\xi = \int_{\mathcal O_\xi}\overline{f_1}(\vb*\varsigma)f_2(\vb*\varsigma) d\vb*\varsigma
\end{equation}
for $f_1,f_2 \in L^2(\mathcal O_\xi)$ w.r.t.~the inner-product norm $\norm{f}_{\xi,2}=\sqrt{\ip{f}{f}_\xi}$.
\end{notation}

\subsubsection{The symplectic foliation of the smooth Poisson sphere}
We recall that the collection of all unitary adjoint orbits $\mathcal O_\xi\subset\mathcal S^7$ defines a symplectic foliation of  the smooth Poisson manifold $(\mathcal S^7, \widehat{\Pi}_{\mathfrak g})$, where  $\widehat{\Pi}_{\mathfrak g}={\Pi}_{\mathfrak g}|_{\mathcal S^7}$ for ${\Pi}_{\mathfrak g}$ the KAKS Poisson bi-vector on $\mathfrak{g}=\mathfrak{su}(3)$ given by
\begin{equation}\label{KAKS-bivec}
    {\Pi}_{\mathfrak g} = \sum_{j,k,l}c^l_{kj}\,x_l\,\partial_j\otimes\partial_k\, ,
\end{equation}
where $c^l_{k,j}$ are the constant structures of $\mathfrak{su}(3)$ in the basis $\{E_1,...,E_8\}$ and likewise for $(x_1,...,x_8)$ being coordinates in this basis, see \cite{kir}. We denote this foliation by
\begin{equation}\label{foliS7}
\begin{aligned}
    \bigcup_{\xi\in\overline{\mathcal F}}(\mathcal O_\xi,{\Pi}_{\mathfrak g}|_{\mathcal O_\xi})&=(\mathcal S^7,\widehat{\Pi}_{\mathfrak g}) \ , \  \ \widehat{\Pi}_{\mathfrak g}={\Pi}_{\mathfrak g}|_{\mathcal S^7} \ , \\
   (\mathcal O_\xi,{\Pi}_{\mathfrak g}|_{\mathcal O_\xi})\equiv(\mathcal O_\xi&,\Omega_\xi) \ , \ \ \Omega_\xi=\Pi_{\mathfrak g}|_{\mathcal O_\xi} 
   \ \ \mbox{symplectic}  \  . 
   \end{aligned}     
\end{equation}

The orbits for $\xi_{(x,y)}\in\mathcal F$ are the leaves $\mathcal O_{\xi_{(x,y)}}\simeq\mathcal E$ of the regular part of this foliation, with the two closing orbits $\mathcal O_{\xi_{(1,0)}}\simeq\mathcal O_{\xi_{(0,1)}}\simeq\mathbb C P^2$ comprising the singular leaves. We now describe this singular foliation in more detail. 

Recall parametrization \eqref{par_xy} of $\overline{\mathcal F}$. For $x\ge 1/\sqrt{3}$, we have $y\le 1/\sqrt{3}$ and we consider the orbit $\mathcal O_{(x,y)}$ as a $\mathcal S^2$ bundle over the base $SU(3)/H$, where each fiber $\mathcal S^2$ is generated by the action of $H\simeq U(2)$. In this manner, as $\xi_{(x,y)}$ approaches $\xi_{(1,0)}$, whose isotropic subgroup is $H$, the $2$-spheres given by the action of $H$ on $\xi_{(x,y)}$ must collapse. More explicitly, via the parametrization of $H$ by Euler angles,
\begin{equation}
    R_U(\alpha,\beta, \gamma) = \exp(-i\alpha U_3)\exp(-\dfrac{\beta}{2}(U_+-U_-))\exp(-i\gamma U_3)\, ,
\end{equation}
we get the following parametrization of the fiber that contains $\xi_{(x,y)}$:
\begin{equation}
\begin{aligned}
    & \dfrac{\sqrt{3}}{4}\left\{2x+y\left(1-\cos(\beta)\right)\right\}E_3+\dfrac{\sqrt{3}}{2}y\sin(\beta)\cos(\alpha)E_6\\
    &\hspace{1 em} +\dfrac{\sqrt{3}}{2}y\sin(\beta)\sin(\alpha)E_7+\dfrac{1}{4}\left\{2x+y(1+3\cos(\beta))\right\}E_8 \\
    = \ &\dfrac{i}{2\sqrt{6}}{\small\begin{pmatrix}
        4x+2y & 0 & 0\\
        0 & -2x-y & 0\\
        0 & 0 & -2x-y
    \end{pmatrix}} +\, \dfrac{iy}{2}\sqrt{\frac{3}{2}} {\small\begin{pmatrix}
        0 & 0 & 0\\
        0 & \cos(\beta) & e^{-i\alpha}\sin(\beta)\\
        0 & e^{i\alpha}\sin(\beta) & -\cos(\beta)
    \end{pmatrix}}\, .
\end{aligned}
\end{equation}
This is a $2$-sphere centered at the diagonal matrix 
\begin{equation}
    \dfrac{\sqrt{3}}{4}(2x+y)E_3+\dfrac{1}{4}(2x+y)E_8 = \dfrac{i}{2\sqrt{6}}(2x+y)\big(2, -1,-1\big)
\end{equation}
in the affine $3$-dimensional space given by translations by $E_6,E_7,E_3-\sqrt{3}E_8$, and the radius of the sphere is 
\begin{equation}\label{rady}
   \varrho(y)= \frac{\sqrt{3}}{2}y\to 0 \ , \ \ \mbox{as} \ \ \  y\to 0 \ . 
\end{equation}

Parameterizing the solutions of \eqref{par_xy} by $y \in [0,1]$, we get
\begin{equation}\label{x(y)}
    x(y) = \dfrac{-y+\sqrt{4-3y^2}}{2} \implies \xi_{(x,y)}=\xi_{(x(y),y)} =:\zeta_y \ ,
\end{equation}
so that we have $\overline{\mathcal F} = \{\zeta_y:y\in [0,1]\}$ and $\mathcal F = \{\zeta_y: y\in (0,1)\}$, and we set
\begin{equation}
    \overline{\mathcal F}_\le := \{\zeta_y: y\in [0,1/\sqrt{3}]\} \ , \ \ \mathcal F_\le := \{\zeta_y: y\in (0,1/\sqrt{3}]\}\, .
\end{equation}

Thus each leaf $\mathcal O_{(x,y)}=\mathcal O_{\zeta_y}$ of the symplectic foliation in a neighborhood of $\mathcal O_{(1,0)}=\mathcal O_{\zeta_0}$ in $\mathcal S^7$ is parametrized by $y\in[0,1/\sqrt{3}]$ and 
\begin{equation}\label{Morsef}
    f: \overline{\mathcal F}_\le\to\mathbb R^+\, , \ \ y\mapsto 
    f(y) = \frac{\sqrt{3}}{4}y^2 \ ,
\end{equation} 
is a Morse function for the Morse-Bott singularity at $y=0$. 

Analogously, for $x\le 1/\sqrt{3}$, $y\ge 1/\sqrt{3}$, we consider the orbit $\mathcal O_{(x,y)}$ as an $\mathcal S^2$ bundle over $SU(3)/\widecheck H$ and obtain equations  \eqref{rady}-\eqref{Morsef} with $x\leftrightarrow y$ interchanged, describing the foliation in a neighborhood of the Bott-Morse singular orbit $\mathcal O_{\xi_{(0,1)}}$. Furthermore, the two closed neighborhoods $\{0\le y\le 1/\sqrt{3}\}$ and $\{1\ge y\ge 1/\sqrt{3}\}$ are glued together at the mesonic orbit $\mathcal O_{(x,y)}$ with $x=y=1/\sqrt{3}$. 

Thus, the singular foliation of $(\mathcal S^7,\widehat\Pi_{\mathfrak{g}})$ by (co)adjoint orbits, with singularities of Morse-Bott type, is analogous to the singular foliation of  $\mathcal S^2$ by circles of constant latitude, with singularities of Morse type, except that now we have isolated singular orbits (isomorphic to $\mathbb C P^2$), instead of isolated singular points.

But for our purposes, it will also be useful to construct the foliation via the special polynomial function below. Again, let $(x_1,...,x_8)$ be coordinates on $\mathfrak{sl}(3)$ in the basis $\{E_1,...,E_8\}$ and recall the parametrization $\xi_{(x,y)}\in \overline{\mathcal F}$, cf. \eqref{pointsF}-\eqref{par_xy}. 

\begin{proposition}\label{prop:c3_pol}
The polynomial $\tau:\mathfrak{sl}(3)\to\mathbb C$ given by
 \begin{equation}\label{taupol}
 \begin{aligned}
     \tau & = 6(x_1^2+x_2^2+x_3^2)x_8-2x_8^3+6\sqrt{3}(x_1(x_4x_6+x_5x_7)-x_2(x_4x_7-x_5x_6))\\
    & \hspace{1 em}-3(x_4^2+x_5^2+x_6^2+x_7^2)x_8+3\sqrt{3}x_3(x_4^2+x_5^2-x_6^2-x_7^2)
\end{aligned}
 \end{equation}
is $SU(3)$-invariant and separates the points of $\overline{\mathcal F}$.
\end{proposition}

The proof of this proposition is deferred to after Proposition \ref{prop:S_equiv}, since the latter  will be used for this proof. 

\begin{remark}\label{checktau}
    Thus, for each $\xi\in\overline{\mathcal F}$, the orbit $\mathcal O_{\xi}\subset \mathcal S^7$ is exactly the preimage by $\tau|_{\mathcal S^7}$ of the real number
\begin{equation}
    \chi_\xi := \tau(\xi) \ .
\end{equation}
In addition, the polynomial which is the complement of the restriction $\tau|_{\mathcal O_\xi}$, 
\begin{equation}\label{delta_func_xy}
    \check\tau_\xi := \tau|_{\mathcal S^7}-\chi_\xi \in Poly(\mathcal S^7) \ , 
\end{equation}
is an $SU(3)$-invariant polynomial vanishing on $\mathcal O_\xi$ and only on this orbit of $\mathcal S^7$.
\end{remark}

\subsubsection{The coarse Poisson sphere}
However, we shall not concern ourselves with functions on all unitary orbits, but only on a countable family identified as follows. Consider the equivalence relation $\sim$ on orbits of $\mathfrak{su}(3)$, which is given by rescaling: 
\begin{equation}\label{eqOrb}
\mathcal O\sim\mathcal O'    \iff \exists\, \alpha> 0 \ \ \mbox{s.t.} \ \ v\mapsto \alpha v \ \ \mbox{is a bijection} \ \mathcal O\to\mathcal O'\, .
\end{equation}

\begin{definition}\label{ratorbdef}
	An \emph{integral orbit} is the orbit in $\mathfrak{su}(3)$ of a dominant weight. A \emph{rational orbit} is an orbit in $\mathcal S^7\subset\mathfrak{su}(3)$ equivalent to some integral orbit.
\end{definition}
\begin{notation}\label{ratorbnot}
   We shall denote by $\overline{\mathcal Q}\subset \overline{\mathcal F}$ the subset of rational orbits, and by $\mathcal Q$ the respective subset of $\mathcal F$. 
\end{notation}

\begin{definition}\label{intrad-defn}
 For each $\xi \in \overline{\mathcal Q}$, its \emph{integral radius} is
\begin{equation}\label{intrad}
    	r(\xi) := \min\{R>0: R\,\xi \  \mbox{is a dominant weight}\, \} 
\end{equation}
and its \emph{first dominant weight} is
	\begin{equation}\label{1_dom_weight}
		\omega_\xi := r(\xi)\,\xi\, .
	\end{equation}
\end{definition}
In other words, for  each $\xi \in \overline{\mathcal Q}$,
\begin{equation}\label{defr} 
(r(\xi))^2=\norm{\omega_\xi}^2=\dfrac{2}{3}(p_1^2+p_1q_1+q_1^2) \ , \ \ \ 
        \omega_\xi = r(\xi)\,\xi=\omega_{(p_1,q_1)} \, ,
\end{equation}
where $\omega_\xi=\omega_{(p_1,q_1)}$ is the first nonzero dominant weight proportional to $\xi\in \overline{\mathcal Q}$, that is,  the dominant weight $\omega_{(p_1,q_1)}\propto\xi$ with the smallest\footnote{Clearly, if $\omega_{(p_1,q_1)}\propto\xi$, then $\omega_{(sp_1,sq_1)}\propto\xi$ $\forall s\in\mathbb N$, with $\norm{\omega_{(sp_1,sq_1)}}=s\norm{\omega_{(p_1,q_1)}}$.} nonzero norm in $\mathfrak{su}(3)$, which is by definition  the integral radius $r(\xi)$ of $\xi$.

Note that for $\xi\in(\overline{\mathcal Q}\setminus{\mathcal Q})=(\overline{\mathcal F}\setminus{\mathcal F})$, we have $r(\xi)=\sqrt{2/3}$ and the first dominant weight is either $\omega_{(1,0)}=\varpi_1$, for the defining representation of $SU(3)$, or  $\omega_{(0,1)}=\varpi_2$ for its dual, cf. (\ref{fundweights}). On the other hand, for any $\xi_{(x,y)}\in \mathcal F$, we have that $\xi_{(x,y)}\in {\mathcal Q}$ if and only if $x/y\in\mathbb Q$ (hence Definition \ref{ratorbdef} and Notation \ref{ratorbnot}), thus the set of rational orbits is dense in the set of all adjoint unitary orbits. 

Therefore, the collection of all rational orbits provides a countably dense symplectic foliation of the Poisson manifold $(\mathcal S^7, \widehat{\Pi}_{\mathfrak g})$ which includes the singular leaves $\mathcal O_{(1,0)}\simeq\mathcal O_{(0,1)}\simeq\mathbb C P^2$ of foliation \eqref{foliS7}. We denote this by 
\begin{equation}\label{coarseS7}
    \bigcup_{\xi\in\overline{\mathcal Q}}(\mathcal O_\xi,{\Pi}_{\mathfrak g}|_{\mathcal O_\xi})=:\{\mathcal S^7,\widehat{\Pi}_{\mathfrak g}\}\subset(S^7,\widehat{\Pi}_{\mathfrak g}) \, . 
\end{equation}

\begin{definition}\label{def:rat_coars}
    We shall refer to $\{\mathcal S^7,\widehat{\Pi}_{\mathfrak g}\}$ as the \emph{rational coarsening} of $(S^7,\widehat{\Pi}_{\mathfrak g})$, or simply refer to  $\{\mathcal S^7,\widehat{\Pi}_{\mathfrak g}\}$ as the \emph{coarse Poisson sphere}.\footnote{In implicit contrast to the \emph{smooth Poisson sphere} $(S^7,\widehat{\Pi}_{\mathfrak g})$.}
\end{definition}

\begin{remark}\label{minradunb}
    As emphasized, $\{\mathcal S^7,\widehat{\Pi}_{\mathfrak g}\}$ is the dense subset of $(S^7,\widehat{\Pi}_{\mathfrak g})$ where we have a well defined function 
    \begin{equation}\label{welldefr}
        r:\{\mathcal S^7,\widehat{\Pi}_{\mathfrak g}\}\to\mathbb R^+ \ , \ \ \mathcal O_\xi\mapsto r(\xi) \, , 
    \end{equation}
    for $r(\xi)$ the integral radius\footnote{We emphasize, for clarity, that the integral radius $r$ of a rational orbit in $\{\mathcal S^7,\widehat\Pi_{\mathfrak{g}}\}$ \emph{is not} the radius $\varrho$ of the two-sphere that fibers over $\mathbb C P^2$ for a generic orbit, cf. \eqref{rady}.} of $\mathcal O_\xi$, cf. \eqref{intrad}-(\ref{defr}). This function $r$, as defined by \eqref{intrad}-(\ref{welldefr}),  has minimum equal to $\sqrt{2/3}$, which is the integral radius of the two singular orbits in $\mathcal S^7$, but $r$ has no upper bound because we can have $\xi_{(x,y)}\propto\omega_{(p_1,q_1)}$ for $p_1$ and $q_1$ without common divisors and as large as we want. 
\end{remark}

In fact, the argument in Remark \ref{minradunb} actually implies: 

\begin{proposition}\label{unboundedr}
    The integral radius function $r:\overline{\mathcal Q}\to\mathbb R^+$, cf. \eqref{intrad}-\eqref{welldefr}, is unbounded on any neighborhood of any $\xi\in\overline{\mathcal Q}$. 
\end{proposition}

\begin{remark}
    The equivalence relation (\ref{eqOrb}) compensates, up to a point, for the fact that we will be working with actual adjoint orbits embedded in $\mathfrak{su}(3)$, rather than the abstract orbits  $\mathbb C P^2$ or $\mathcal E$. A bonus for this setting is that  we shall later be able to investigate how the twisted algebras defined for each $\mathcal O_\xi\in\{\mathcal S^7,\widehat{\Pi}_{\mathfrak g}\}$ can or cannot be ``glued'' along the rational coarsening of Poisson manifold $(\mathcal S^7, \widehat{\Pi}_{\mathfrak g})$, for appropriate families of symbol correspondence sequences, in an asymptotic limit.
\end{remark} 

\subsection{Main results for $C^\ast$-algebras on (co)adjoint orbits}

We now state and prove the main results for $C^\ast$-algebras on (co)adjoint orbits of $SU(3)$ that will be relevant for our considerations on asymptotics of quark systems.\footnote{It is not yet known to us whether (some of) the results presented below have been stated or proved before, therefore we do so here.} 

First, for the particular cases of pure-quark systems, we have the analogous of the no-go theorem for spin systems, that is, we have the theorem below which is just the translation for the pair $(SU(3),\mathbb C P^2)$ of the theorem proved by Rieffel in \cite{rieffel} for the pair $(SU(2), \mathbb C P^1)$.\footnote{In \cite{rieffel}, Rieffel actually stated his theorem with respect to $SO(3)$, but since the action of $SU(2)$ on $S^2 \simeq\mathbb CP^1$ is effectively an action of $SO(3)$, the two statements are equivalent.} 

\begin{theorem}\label{theo:const_quant_cp2}
	Any $SU(3)$-equivariant unital $C^\ast$-algebra structure on $C^\infty_{\mathbb C}(\mathbb CP^2)$ is commutative.
\end{theorem}
\begin{proof}
We shall follow closely to Rieffel's proof for $(SU(2), \mathbb C P^1)$, making the necessary adaptations for $(SU(3),\mathbb C P^2)$. The main idea is to show that the product of linear polynomials is commutative and generates the entire algebra for $C^\infty_{\mathbb C}(\mathbb CP^2)$.

Let $A_\star=(A,\star, ^\ast, \norm{\,})$ denote a $SU(3)$-equivariant unital $C^\ast$-algebra structure on $A = C^\infty_{\mathbb C}(\mathbb CP^2)$, where $\star$, $^\ast$ and $\norm{\,}$ are the product, involution and $C^\ast$-norm, respectively. We know that $A$ decomposes as a sum of irreps $(n,n)$, for every non negative integer $n$, and each such irrep appears just once, cf. Proposition I.4.2 and Definition I.4.3. Let $A_n \subset A$ be the invariant subspace where $SU(3)$ acts via the irrep $(n,n)$, so that 
\begin{equation}\label{Adecomp}
    A=A_0\oplus A_1\oplus A_2\oplus\cdots A_n\oplus A_{n+1}\oplus\cdots 
\end{equation}

\begin{lemma}\label{lemma:A_0-cent}
$A_0$ is the linear span of the identity in $A_\star=(A,\star, ^\ast, \norm{\,})$.
\end{lemma}
\begin{proof}
    Let $e \in A$ be the identity in $A_\star=(A,\star, ^\ast, \norm{\,})$. Then
    \begin{equation}
        e^ga = (ea^{g^{-1}})^g = a = (a^{g^{-1}}e)^g = ae^g\, ,
    \end{equation}
    hence $e \in A_0$.
\end{proof}

The next lemma is the crucial part of the proof of the theorem.  

\begin{lemma}\label{lemma:A_1-comm}
The product on $A_1\subset A$ is commutative, that is, 
\begin{equation}\label{trivialhomomA1}
    a\star b = b\star a \ , \ \ \forall a,b \in A_1 \ .
\end{equation}
\end{lemma}
\begin{proof}
Consider the commutator on $A_\star$: 
\begin{equation}\label{star-comm}
    [a,b]_\star:=a\star b-b\star a \ , \ \ \forall a,b\in A \, . 
\end{equation}
Since $A_\star$ is $SU(3)$-equivariant, the map $A_1\times A_1 \ni (a,b)\mapsto[a,b]_\star \in A$ factors through an equivariant map $A_1\wedge A_1 \ni a\wedge b \mapsto [a,b]_\star\in A$. First, one can easily verify that
\begin{equation}\label{sumAA}
    (1,1)\wedge (1,1) = (1,1)\oplus (3,0) \oplus (0,3)\, .
\end{equation}
Then, by straightforward computations, we obtain that the highest weight vectors for each respective summand in \eqref{sumAA} are as follows (cf. Definition I.2.1):   
\begin{equation}
\begin{aligned}
    e_>^{(1,1)} & = \, \left(\sqrt{\dfrac{3}{2}}\vb*e((1,1);\vb 0_1, 0)-\dfrac{1}{\sqrt{2}}\vb*e((1,1);\vb 0_1, 1)\right)\wedge \vb*e((1,1);(210),1/2)\\
    & \ \ \ \ \ \ \ + \vb*e((1,1);(120),1)\wedge\vb*e((1,1);(201), 1/2), \\
    e_>^{(3,0)} & =  \vb*e((1,1);(201),1/2)\wedge \vb*e((1,1);(210), 1/2)\, , \\
    e_>^{(0,3)} & = \vb*e((1,1);(120),1)\wedge \vb*e((1,1);(210), 1/2)\, .
\end{aligned}
\end{equation}

By Schur's Lemma and the decomposition of $A$ into irreps $(n,n)$, cf. \eqref{Adecomp}, we conclude that the invariant subspace of $A_1\wedge A_1$ corresponding to $(3,0)\oplus (0,3)$ is in the kernel of the induced commutator map, whereas the restriction of such map to the invariant subspace corresponding to $(1,1)$ is either an isomorphism or the null map, hence $[A_1,A_1]_\star$ is either $A_1$ or $0$. 

Suppose that $[A_1,A_1]_\star = A_1$. Then
\begin{equation}
    e^{(1,1)}_>+T_-(e^{(3,0)}_>) = \sqrt{6}\, \vb*e(1;\vec 1,0)\wedge \vb*e(1;(210),1/2)
\end{equation}
is mapped into a highest weight vector of $A_1$ by the induced commutator map on $A_1\wedge A_1$. So we can choose $a_0,a_> \in A_1$, where $a_0$ is self adjoint and $a_>$ is a highest weight vector, such that $[a_0,a_>]_\star = a_>$. Let $B_k$ be the space spanned by the $\star$-product of at most $k$ elements of $A_1\subset A_\star$, and $B_\star$ the algebra generated by $A_1$. Since the product map $A_1\times A_n\to A$ that sends $(a_1,a)\in A_1\times A_n$ to $a_1\star a$ is a bilinear map, it factors through $A_1\otimes A_n$. Then, by the equivariance of $A_\star$ and the Clebsch-Gordan series of $(1,1)\otimes (n,n)$, we get that $B_k\subset A_1\oplus...\oplus A_k$. 

By the Leibniz rule and induction on $k$, we have
\begin{equation}\label{in_act_alg}
    [a_0, (a_>)^k]_\star = a_>\star[a_0,(a_>)^{k-1}]_\star+[a_0,a_>]\star (a_>)^{k-1} = k(a_>)^k
\end{equation}
for every $k\in \mathbb N$. Hence $B_1\subset ...\subset B_k\subset ...$ eventually stabilizes, otherwise $[a_0,\cdot\, ]_\star$ would be an unbounded operator, from \eqref{in_act_alg}, contradicting the fact the we have a $C^\ast$-algebra. So $B_\star$ is finite dimensional and there is some $k$ such that $C_\star = A_0\oplus B_\star$ is a finite dimensional unital subalgebra whose underlying space decomposes as 
\begin{equation}\label{Cdecomp}
    C = A_0\oplus A_1\oplus...\oplus A_k \ .
\end{equation}

On the other hand, because $C_\star$ is a finite dimensional $C^\ast$-algebra, in principle it would be a direct sum of full matrix algebras, say
\begin{equation}
    C_\star \simeq \bigoplus_{j=1}^n M_{\mathbb C}(d_j)\, .
\end{equation}
But by the assumption $[A_1,A_1]_\star = A_1$ and Schur's Lemma, we have a $SU(3)$-equivariant homomorphism $\phi:\mathfrak{su}(3)\to A_1$ inducing an inner action $\alpha_\phi$ of $\mathfrak{su}(3)$ on $C_\star$ that coincides with the infinitesimal action induced by the natural $SU(3)$-action. Then each identity $\mathds 1_j \in M_{\mathbb C}(d_j)$ is in the center of $C_\star$, which means each $\mathds 1_j$ is fixed by $SU(3)$. However, $C$ has only one copy of the trivial irrep of $SU(3)$, namely on $A_0$, cf. \eqref{Cdecomp}, hence $[A_1,A_1]_\star = A_1$ implies that $C_\star$ is a  $SU(3)$-equivariant unital subalgebra of $A_\star$ isomorphic to a full matrix algebra, 
\begin{equation}\label{C=Md}
    C_\star \simeq M_{\mathbb C}(d)\, , 
\end{equation}
and furthermore implies that we have a $SU(3)$-equivariant isomorphism 
\begin{equation}
    A_\star\simeq C_\star\otimes C_\star' \, , 
\end{equation}
where $C_\star'$ is the commutant  of $C_\star$ in $A_\star$, whose underlying space $C'$ is also invariant by the action of $SU(3)$. 

Therefore, from \eqref{Adecomp} and \eqref{Cdecomp}, either $C'=A_0$, in which case $A_\star\simeq C_\star$ is isomorphic to a matrix algebra, cf. \eqref{C=Md}, in contradiction to $A=C^\infty_{\mathbb C}(\mathbb C P^2)$, or for each nontrivial $A_n \subset C'$, $A_1\otimes A_n$ has two copies of $A_n$, another contradiction, cf. \eqref{Adecomp}. Thus $[A_1,A_1]$ cannot be $A_1$ and hence $[A_1,A_1] = 0$.
\end{proof}

To finish the proof of the theorem, let again $C_\star$ be the $C^\ast$-subalgebra generated by $A_0\oplus A_1$, and let $C_k$ be the linear span of the product of at most $k$ elements in $A_0\oplus A_1$, for $k\in \mathbb N$. As already argued in the proof of the previous lemma, $C_k \subset A_0\oplus A_1\oplus ... \oplus A_k$. Suppose that the chain $C_1 \subset C_2 \subset ... \subset  C_k \subset ...$ eventually stabilizes, which means $A_{n_0}\star A_1 \subset A_{n_0-1}\oplus A_{n_0}$ for some $n_0\in\mathbb N$, and $C_\star$ is a finite-dimensional $C^\ast$-algebra. 

Again, $C_\star$ would in principle be a direct sum of full matrix algebras, cf. \eqref{C=Md}. But since $C_\star$ is commutative due to Lemmas \ref{lemma:A_0-cent} and \ref{lemma:A_1-comm}, we could at most have
\begin{equation}\label{C=sum_complex}
    C_\star \simeq \bigoplus_{j=1}^{\dim C}\mathbb C_j \ , \ \ \mathbb C_j \simeq \mathbb C\hspace{1 em} \forall\, j\in \{1,...,\dim C \}\, .
\end{equation}
Let $1_j \in \mathbb C_j$ be its identity, so that the primitive spectrum of $C_\star$ is a finite discrete space $\mathrm{Prim}(C_\star) =\{\ker(\pi_1),...,\ker(\pi_{\dim C})\}$, where each $\pi_j$ is multiplication by $1_j$, which works as a projection onto $\mathbb C_j$. By $SU(3)$-equivariance of $C_\star$, we have an induced continuous action of $SU(3)$ on $\mathrm{Prim}(C_\star)$. Since $SU(3)$ is connected, this action is trivial, implying that each $1_j$ is fixed by $SU(3)$. But since $C$ carries only one copy of $(0,0)$, namely the subspace $A_0$, cf. \eqref{Cdecomp}, $C_\star$ must be isomorphic to $\mathbb C$, which contradicts the fact that $A_1 \subset C$. Therefore, for every $k\in \mathbb N$, $C_k$ is a proper subspace of $C_{k+1}$ and $C=A\implies C_\star=A_\star$ is commutative.
\end{proof}

Then, the following corollary is immediate from Theorem \ref{theo:const_quant_cp2} and its proof. 

\begin{corollary}\label{nogo-cor1}
    Let $\mathcal O\simeq \mathbb C P^2$ be a nongeneric (co)adjoint orbit  of $SU(3)$. Then,  
    there is no $SU(3)$-equi\-va\-ri\-ant unital $C^\ast$-algebra structure $A_\star=\big(A,\star, ^\ast, \norm{\,}\big)$ for $A=C^\infty_{\mathbb C}(\mathcal O)$ with a nontrivial $SU(3)$-equivariant homomorphism $\phi:\mathfrak{su}(3)\to A_\star$ as  
    \begin{equation}\label{a_X}
    \phi:\mathfrak{su}(3)\ni X\mapsto a_X\in A \ , \ \ a_{[X,Y]} = [a_X,a_Y]_{\star} \ , 
\end{equation}
where $[\cdot,\cdot]_\star$ is the commutator in $A_\star$. Furthermore, if $A_\star=\big(A,\star, ^\ast, \norm{\,}\big)$ is a $SU(3)$-equi\-va\-ri\-ant unital $C^\ast$-algebra structure for an invariant subspace $A\subset C^\infty_{\mathbb C}(\mathcal O)$, then a nontrivial $SU(3)$-equivariant homomorphism $\phi:\mathfrak{su}(3)\to A_\star$ as  in \eqref{a_X} exists only if $A$ is finite dimensional, in which case $A_\star$ is isomorphic to a full matrix algebra with a $\phi$-induced inner action
\begin{equation}\label{innX}
    \alpha_\phi:\mathfrak{su}(3)\times A_\star\to A_\star \ , \ \ (X,a)\mapsto [a_X,a]_\star \ , 
\end{equation}
which coincides with the natural action\footnote{$SU(3)$ acts on the space $A$ of $A_\star$ and this induces the natural infinitesimal action of $\mathfrak{su}(3)$.} of $\mathfrak{su}(3)$ on $A$.
\end{corollary}

Now, for a generic (co)adjoint orbit $\mathcal O\simeq\mathcal E$ of $SU(3)$, an analogous of Theorem \ref{theo:const_quant_cp2} is not known to us. But we can state a weakened version of Corollary \ref{nogo-cor1}. 

\begin{definition}
    A $SU(3)$-equi\-va\-ri\-ant unital $C^\ast$-algebra $A_\star=\big(A,\star, ^\ast, \norm{\,}\big)$ is a \emph{bona-fide} $SU(3)$-$C^\ast$-\emph{algebra} if there is a nontrivial $SU(3)$-equivariant homomorphism $\phi:\mathfrak{su}(3)\to A_\star$ as in \eqref{a_X} inducing a nontrivial inner action $\alpha_\phi$ of $\mathfrak{su}(3)$ on $A_\star$ as in \eqref{innX}.
    In this case, we denote the algebra by $A_\star^\phi=\big(A,\star, ^\ast, \norm{\,}, \phi\big)$.
\end{definition}

\begin{proposition}\label{CstarE}
Let $\mathcal O\simeq \mathcal E$ be a generic (co)adjoint orbit of $SU(3)$ and assume that $A_\star^\phi=\big(A,\star, ^\ast, \norm{\,},\phi\big)$ is a bona-fide $SU(3)$-$C^\ast$-algebra for $A=C^\infty_{\mathbb C}(\mathcal O)$. Then, the $C^\ast$-algebra generated by $\phi(\mathfrak{su}(3))$ is a finite-dimensional bona-fide $SU(3)$-$C^\ast$-subalgebra $C_\star^\phi\subset A_\star^\phi$ which  is isomorphic to the algebra of operators on an irrep of $SU(3)$ and we have the $SU(3)$-equivariant isomorphism
\begin{equation}\label{A-C}
    A_\star^\phi \simeq C_\star^\phi\otimes C_\star'\, , 
\end{equation}
where $C_\star'$ is the commutant of $C_\star^\phi$ in $A_\star^\phi$. Furthermore, the $\phi$-induced inner action $\alpha_\phi$ of $\mathfrak{su}(3)$ on $C_\star^\phi$ coincides with the natural $\mathfrak{su}(3)$-action on the underlying space $C\subset A$, but $\alpha_\phi$ vanishes on $C'_\star$. 
\end{proposition}
\begin{proof}
The proof follows closely to most of the proof of Lemma \ref{lemma:A_1-comm}. 

Denote by $A_1\subset A$ the complex linear span of the image of $\phi$.\footnote{The Lie algebra $\mathfrak{su}(3)$ is a \emph{real} vector space and the homomorphism \eqref{a_X} is a real linear map.} Then, similar to what we did in the proof of Lemma \ref{lemma:A_1-comm}, for each $k \in \mathbb N$, let $B_k$ be the linear span of products of at most $k$ elements of $A_1$. Each $B_k$ is an $SU(3)$-invariant subspace of $A$ for which the natural $\mathfrak{su}(3)$-action coincides with the induced inner action as in \eqref{innX}, that is, for every $X \in \mathfrak{su}(3)$ there is $a_X\in A_1$ such that the natural action of $X$ on $B_k$ is of the form $B_k\ni b\mapsto [a_X,b]_\star\in B_k$. 

Again, we claim that the chain $B_1 \subset ... \subset B_k \subset ...$ stabilizes. Suppose it doesn't. Then, there is a sequence $(D_k)_{k\geq 2}$ such that each $D_k\subset B_k\setminus B_{k-1}$ is a $SU(3)$-invariant subspace of $A$ carrying a representation $\vb*a_k$ with
\begin{equation}
    \lim_{k\to\infty}|\vb*a_k| =\infty\, ,
\end{equation}
cf. Notation I.2.3. Thus, we can take $a_{X_0}=\phi(X_0)\in A_1$ for $X_0=2i(T_3+U_3)$ and normalized highest weight vectors $e_>^k \in D_k$, so that
\begin{equation}
    \norm{[a_{X_0},e_>^k]_\star} = |\vb*a_k| \to \infty\, ,
\end{equation}
which is absurd, since $[a_{X_0},\cdot\,]_\star$ must be a bounded operator. 

Therefore, the $C^\ast$-algebra generated by $A_1=\mathrm{Span}_{\mathbb C}(\phi(\mathfrak{su}(3)))$ is a finite dimensional $C^\ast$-subalgebra  $C_\star^\phi\subset A_\star^\phi$ with a closed nontrivial inner action  of $\mathfrak{su}(3)$, 
\begin{equation}\label{innXc}
\alpha_\phi:\mathfrak{su}(3)\times C_\star^\phi\to C_\star^\phi \ , \ \ (X,c)\mapsto [a_X,c]_\star \ , 
\end{equation}
which coincides with the natural $\mathfrak{su}(3)$-action on the underlying space $C\subset A$.

In complete analogy to Lemma \ref{lemma:A_0-cent}, the subspace $A_0 \subset A$ of invariant elements is unidimensional and is generated by the identity of $A_\star^\phi$. Using a suitable Casimir operator, cf. (I.B.3), the morphism of $\mathfrak{su}(3)$ into $C_\star^\phi$ creates a non trivial invariant element in $C_\star^\phi$, thus $A_0 \subset C_\star^\phi$ so that $C_\star^\phi$ is also unital, hence it is a bona-fide $SU(3)$-$C^\ast$-subalgebra of $A_\star^\phi$, and in the same vein as was shown in Lemma \ref{lemma:A_1-comm}, $C_\star^\phi$ must be isomorphic to a full matrix algebra,
\begin{equation}
    C_\star^\phi \simeq M_{\mathbb C}(d)\, .
\end{equation}

In particular, the composition of $\phi$ with the above isomorphism gives a representation of $\mathfrak{su}(3)$ on $\mathbb C^d$, which is the infinitesimal action induced by a representation of $SU(3)$ since the group is simply connected. Such $SU(3)$-representation on $\mathbb C^d$ is irreducible because a projection on any invariant subspace of $\mathbb C^d$ spans a trivial irrep of $SU(3)$ within $C_\star^\phi$, but $C_\star^\phi$ carries only one copy of the trivial irrep, namely $A_0$. Thus, $C_\star^\phi$ is isomorphic to the algebra of operators on an irrep of $SU(3)$, and we have the global $SU(3)$-equivariant isomorphism \eqref{A-C} with $\alpha_\phi$ vanishing on $C_\star'$, the commutant of the $C^\ast$-algebra generated by $\phi(\mathfrak{su}(3))$.
\end{proof}

In view of the above, we introduce: 

\begin{definition}
    A bona-fide $SU(3)$-$C^\ast$-algebra $A_\star^\phi=\big(A,\star, ^\ast, \norm{\,},\phi\big)$ is a \emph{faithful} $SU(3)$-$C^\ast$-algebra if the inner $\mathfrak{su}(3)$-action $\alpha_\phi$ coincides with
    the natural $\mathfrak{su}(3)$-action on the underlying space $A$. 
\end{definition}

\begin{definition}
    Let $A_\star^\phi=\big(A,\star, ^\ast, \norm{\,},\phi\big)$  be a bona-fide $SU(3)$-$C^\ast$-algebra. If $A_\star^\phi$ decomposes as in \eqref{A-C}, where $C_\star^\phi$ is a faithful $SU(3)$-$C^\ast$-subalgebra and $\alpha_\phi$ vanishes on $C_\star'$, then $C_\star^\phi$ is the $SU(3)$-\emph{core} of $A_\star^\phi$.
\end{definition}

Thus, a bona-fide $SU(3)$-$C^\ast$-algebra $A_\star^\phi=\big(A,\star, ^\ast, \norm{\,},\phi\big)$ is faithful if and only if $A_\star^\phi=C_\star^\phi$ ($C'=A_0$ in \eqref{A-C}), and we can restate the previous  results as: 

\begin{corollary}\label{nogo-cor}
    Let $\mathcal O$ be any (co)adjoint orbit of $SU(3)$ and $A\subseteq C^\infty_{\mathbb C}(\mathcal O)$ an invariant subspace. If $A_\star^\phi=\big(A,\star, ^\ast, \norm{\,},\phi\big)$ is a faithful 
    $SU(3)$-$C^\ast$-algebra, then $A$ is finite dimensional. More generally, if $A_\star^\phi$
    is a bona-fide $SU(3)$-$C^\ast$-algebra, then $A_\star^\phi$ has a finite-dimensional $SU(3)$-core $C_\star^\phi$ isomorphic to
    the algebra of operators on an irrep of $SU(3)$ defined by $\phi$. In particular, if $\mathcal O\simeq \mathbb C P^2$, then $A_\star^\phi=C_\star^\phi$.  
\end{corollary}

\subsection{Preliminary considerations for semiclassical asymptotics}\label{main-arguments}

We now reflect on the semiclassical asymptotics for quark systems, in light of the results of the previous subsection. First, we look at the program of deformation quantization. 

Since every (co)adjoint orbit $\mathcal O$ of $SU(3)$ is a Hamiltonian $SU(3)$-space \cite{kir}, the $SU(3)$-invariant symplectic form on $\mathcal O$, cf. \eqref{foliS7}, defines the classical algebra of observables, which is the Poisson algebra $A_{\mathcal P}=\big(A,\cdot\, , \{\cdot,\cdot \}\big)$, where $\cdot$ is the pointwise product on $A=C^\infty_{\mathbb C}(\mathcal O)$, with respect to which the Poisson bracket $\{\cdot ,\cdot\}$ is a derivation in both entries. 
Furthermore, we have a nontrivial equivariant homomorphism $\widehat\phi$ from $\mathfrak{su}(3)$ to $A_{\mathcal P}$,  
\begin{equation}\label{a_X-pb}
   \widehat\phi: \mathfrak{su}(3)\to A\, , \ X\mapsto \widehat a_X\, , \ \ \mbox{s.t.}\, \ \  \widehat a_{[X,Y]} = \{\widehat a_X,\widehat a_Y\}\, , 
\end{equation}
which induces a nontrivial action $\widehat\alpha$ of $\mathfrak{su}(3)$ on $A_{\mathcal P}$, given by
\begin{equation}\label{su3actPb}
    \widehat\alpha:\mathfrak{su}(3)\times A_{\mathcal P}\to A_{\mathcal P} \ , \ \ (X,f)\mapsto \{\widehat{a}_X,f\} \ . 
\end{equation}

In this setting, the program of deformation quantization amounts to deforming the pointwise product $\cdot$ on $A_{\mathcal P}$ to a noncommutative product $\star_\hbar$ on $A[[\hbar]]$, the ring of formal power series in the deformation parameter $\hbar$ with coefficients in $A=C^\infty_{\mathbb C}(\mathcal O)$, such that, for any $f=\sum_{k=0}^\infty f_k\hbar^k\in A[[\hbar]]$ and  $g=\sum_{k=0}^\infty g_k\hbar^k\in A[[\hbar]]$,\footnote{For explicity constructions of deformation quantizations of coadjoint orbits of compact semisimple Lie groups, we refer to \cite{fioresi, lledo}.}  
\begin{equation}\label{hto0}
    \lim_{\hbar\to 0} f\star_\hbar g=f_0g_0 \ , \ \ \lim_{\hbar\to 0}\left(\hbar^{-1}(f\star_\hbar g-g\star_\hbar f)\right)=i\{f_0,g_0\} \ .
\end{equation}
Such a formal algebra $A_\hbar=(A[[\hbar]],\star_\hbar)$ would be thought of as a `quantum'' algebra deforming the classical algebra $A_{\mathcal P}$ and this is often called a \emph{quantization} of $\mathcal O$. 

But since $SU(3)$ is the symmetry group of $\mathcal O$, any true quantum algebra which respects the $SU(3)$-equivariance of $A$ is a $SU(3)$-equivariant  unital $C^\ast$-algebra. However, from Theorem \ref{theo:const_quant_cp2}, for $A=C^\infty_{\mathbb C}(\mathbb C P^2)$ it is impossible for such a formally deformed algebra $A_\hbar$ to converge\footnote{For instance, by treating $\hbar$ as a constant, as it is in Physics (and which for an appropriate choice of units can be set $\hbar=1$), reinterpreting the limits in \eqref{hto0} accordingly (semiclassical limit of high energies, high momenta, high quantum numbers, high expectation values etc.).} to a $SU(3)$-equi\-va\-ri\-ant unital $C^\ast$-algebra $A_\star$ such that its commutator tends to the Poisson bracket in some limit of elements in $A$. 

Furthermore, any $SU(3)$-symmetric quantum algebra worthy of its name must have an inner action of $\mathfrak{su}(3)$, that is, quantum operators generating the symmetry group. However, from Proposition \ref{CstarE}, for $A=C^\infty_{\mathbb C}(\mathcal E)$ it is impossible for such a formally deformed algebra $A_\hbar$ to converge to a bona-fide $SU(3)$-$C^\ast$-algebra $A_\star^\phi$ such that its commutator approaches the Poisson bracket in some limit, since the $\mathfrak{su}(3)$-action \eqref{su3actPb} is trivial only on the subspace of constant functions, but the inner $\mathfrak{su}(3)$-action on $A_\star^\phi$ is only nontrivial on a finite dimensional subspace of $A$.

On the other hand, from Corollary \ref{nogo-cor}, for any (co)adjoint orbit $\mathcal O$ of $SU(3)$, if we ask for a nontrivial $SU(3)$-equivariant homomorphism from $\mathfrak{su}(3)$ to a $SU(3)$-equivariant unital $C^\ast$-algebra structure $A_\star$ on some invariant subspace $A \subset C^\infty_{\mathbb C}(\mathcal O)$, we end up with a $SU(3)$-core that is isomorphic to the algebra of operators on some irrep of $SU(3)$, that is, a quantum quark system in the sense of Definition I.5.6.

It follows from these previous results that, just as for spin systems and functions on $\mathbb C P^1$, in order to properly approach the asymptotic limit of noncommutative products of  functions on a $SU(3)$-(co)adjoint orbit $\mathcal O$, we must work with sequences of symbol correspondences from quantum quark systems, in other words,  sequences of twisted algebras defined on increasing finite-dimensional subspaces of $C^\infty_{\mathbb C}(\mathcal O)$ which are induced from  symbol correspondence sequences, and then study the asymptotic limit of such sequences as the dimension tends to infinity.

Thus, the first problem we must deal with is the identification of sequences of quantum quark systems that are suitable for semiclassical asymptotic analysis.

For pure-quark systems, the classical phase space is the orbit $\mathcal O \simeq \mathbb CP^2$, and each symbol correspondence is an isomorphism between the algebra of $\delta(p)\times \delta(p)$ complex matrices $M_{\mathbb C}(\delta(p))$ and the corresponding twisted algebra (cf.~Definition I.3.21) on a $\delta(p)^2$-dimensional subspace of $C^\infty_{\mathbb C}(\mathbb CP^2)$, where $\delta(p)$ is the dimension of an $SU(3)$-irrep $\vb*p=(p,0)$, or $\widecheck{\vb*p}=(0,p)$, which is given by 
\begin{equation}\label{d(p)}
\delta(p)=\dim(\vb*p)=\dim(\widecheck{\vb*p})=\frac{(p+1)(p+2)}{2} \ ,
\end{equation}
so that, $p\mapsto p+1\implies \delta(p)\mapsto \delta(p+1)= \delta(p)+p+2$. In this scenario, we must consider sequences of quantum pure-quark systems $((p,0))_{p\in\mathbb N}$ or $((0,p))_{p\in\mathbb N}$. 

The asymptotic analysis of such sequences of symbol correspondences and twisted algebras for pure-quark systems can be worked out in a way that, although quite more cumbersome, is somewhat analogous to the treatment developed in \cite{RS} for spin systems. In Appendix \ref{app:pq-asymp} we summarize the steps and results of this approach. 

There we show the conditions (on the characteristic numbers) for the sequence of symbol correspondences and their twisted algebras to be of Poisson type, that is, for the sequence of twisted products $(\star^p)_{p\in\mathbb N}$ to be such that, in some sense,\footnote{The precise sense for these limits is presented in Section \ref{sec:asymp_anal} and Appendix \ref{app:pq-asymp}.}  
\begin{equation}\label{limp}
    \lim_{p\to\infty}f\star^p g=fg \ , \ \ \lim_{p\to\infty}p[f,g]_{\star^p}=i\{f,g\} \ , \ \ \forall f,g\in Poly(\mathbb C P^2)\, .
\end{equation}

However, the choice of sequences for pure-quark systems needs to be better justified with a principle that can be extended to mixed quark systems, where the classical phase space is a generic orbit $\mathcal O \simeq \mathcal E$ and the matrix algebras of quantum quark systems are indexed not by single integers, but by pairs $(p,q)$ of integers.

Such a generalized principle shall lead to the definition of ``rays'' of correspondences for each (co)adjoint orbit $\mathcal O_\xi$ in the coarse Poisson sphere $\{\mathcal S^7,\widehat\Pi_{\mathfrak{g}}\}$, cf. \eqref{coarseS7}, as presented in the next section (cf. Definition \ref{def:pois_ray}). With  this definition, we shall be able to make sense of limits similar to the ones in \eqref{limp} and thus study the conditions for such rays of correspondences  to be of Poisson type.

But even with such identification of the sequences of general quantum quark systems suitable for  semiclassical asymptotic analysis, the approach presented in Appendix \ref{app:pq-asymp} is not easily generalized to the asymptotic analysis of mixed quark systems. So we shall develop a new framework using the PBW Theorem for the universal enveloping algebra of $\mathfrak{sl}(3)$, as presented in the next subsections. 

\subsection{PBW Theorem and Poisson algebras of harmonic functions}

We consider general orbits $\mathcal O\simeq \mathcal O_\xi\subset\mathcal S^7\subset\mathfrak{su}(3)$ and, 
in what follows,  invoke the Poincaré-Birkhoff-Witt (PBW) Theorem to describe the Poisson algebra on $Poly(\mathcal O)$.

First, we take $\mathfrak{sl}(3)$ as the complexification of $\mathfrak{su}(3)$,
\begin{equation}
\mathrm{Span}_{\mathbb R}\{E_j: j=1,...,8\}=\mathfrak{su}(3)\subset \mathfrak{sl}(3) = \mathrm{Span}_{\mathbb C}\{E_j: j=1,...,8\}\, .
\end{equation}
Note that the restriction of complex polynomials provides identification
\begin{equation*}
    Poly(\mathfrak{sl}(3))\equiv Poly(\mathfrak{su}(3))\ .
\end{equation*}
Furthermore, on $\mathfrak{sl}(3)$ we have the bilinear form 
\begin{equation}
    (X,Y) = \tr(X Y) \ \ \ \ \forall X,Y \in \mathfrak{sl}(3) \ ,
\end{equation}
which is just a renormalization of the Killing form (and naturally restricts to $\mathfrak{su}(3)$), and which defines the standard inner product (cf. (I.2.8))
\begin{equation}\label{standard-ip}
    \langle X|Y\rangle=(X^\dagger,Y)=\tr(X^\dagger Y) \ \ \ \ \forall X,Y \in \mathfrak{sl}(3) \ .
\end{equation}

We consider the GT basis of $\mathfrak{sl}(3)$, with adjoint representation $(1,1)$ of $SU(3)$, as depicted on Figure \ref{gt_b-sl(3)}.
\begin{figure}[h]\small{
	\centering
	\label{gt_b-sl(3)}
	\begin{tikzpicture}[scale = 1]
		
		\foreach \ang in {120,180,...,360}{
			\filldraw[black] (\ang:2cm) circle (2pt);
		}
		
		\filldraw[black] (60:2cm) +(-2pt,-2pt) rectangle +(2pt,2pt);
		
		\filldraw[black] (0,0) circle (2pt);
		\draw (0,0) circle (4pt);
		
		\node[anchor=south west,scale=0.8] at (0:2cm) {$e_{(2,0,1),1/2}=-T_+$};
		
		\node[anchor=south west,scale=0.8] at (0:-4.7cm) {$e_{(0,2,1),1/2}=T_-$};
		
		\node[anchor=south west,scale=0.8] at (60:2cm) {$e_{(2,1,0),1/2}=V_+$};
		
		\node[anchor=south west,scale=0.8] at (155:4cm) {$e_{(1,2,0),1}=U_+$};
		
		\node[anchor=south west,scale=0.8] at (170:1.5cm) {$e_{\vb*0_1,1}= e_{\vec 1 1}=-\sqrt{2}\,U_3$};
		
		\node[anchor=south west,scale=0.8] at (215:2cm) {$e_{\vb*0_1,0}= e_{\vec 1 0}=\sqrt{\dfrac{2}{3}}\,(T_3+V_3)$};
		
		\node[anchor=south west,scale=0.8] at (210:4.5cm) {$e_{(0,1,2),1/2}=V_-$};
		
		\node[anchor=south west,scale=0.8] at (-65:2.5cm) {$e_{(1,0,2),1}=-U_-$};
	\end{tikzpicture}
	\caption{GT basis for $\mathfrak{sl}(3)$, cf. Definition I.2.1.} }
\end{figure}
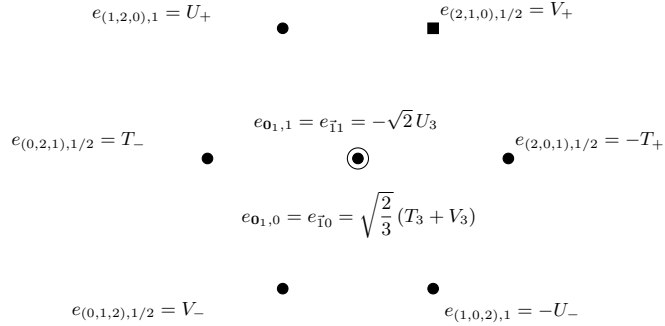
We also impose an ordering on this orthonormal basis s.t. $\{e_1, e_2,e_3\}$ are annihilation operators, $\{e_4,e_5,e_6\}$ are creation operators and $\{e_7,e_8\}$ are Cartan operators. Specifically, we shall choose the ordered basis vectors
\begin{equation}\label{orderingB1}
	\begin{aligned}
		& e_1 = e_{(0,2,1)1/2}=T_- \, , \ e_2  = e_{(1,0,2)1}=-U_- \, , \ e_3 = e_{(0,1,2)1/2}=V_- \, ,\\
	&	e_4 = e_{(2,0,1)1/2}=-T_+\, , \ e_5  = e_{(1,2,0)1}=U_+ \, , \ e_6 = e_{(2,1,0)1/2}=V_+ \, ,\\
		& e_7 = e_{\vec 1 1}=-\sqrt{2}U_3 \  ,  \  e_8 = e_{\vec 1 0}=\sqrt{2/3}(T_3+V_3)=\sqrt{2/3}(2T_3+U_3)\, ,
	\end{aligned}
\end{equation}
cf. Definition I.2.1, 
and we denote this choice of ordered basis for $\mathfrak{sl}(3)$  by 
\begin{equation}\label{ordB1}
    \mathcal B_1=\{e_1,...,e_8\} \, . 
\end{equation}

By PBW Theorem \cite{hump}, the universal enveloping algebra $U(\mathfrak{sl}(3))$ has a basis
\begin{equation}\label{univ_basis}
	\mathcal B_\infty=\bigcup_{d\in \mathbb N_0}\mathcal B_d \ , \ \ \
  \mathcal B_d = \{e_{j_1}...e_{j_d}:1\le j_1\le...\le j_d\le 8\}\, ,    
\end{equation}
where the empty product ($d=0$) is the unity $1$ and where $e_1,\cdots, e_8$ satisfy the commutation relations of $\mathfrak{su}(3)$ (but not any specific nilpotency relation apriori, that is, not represented by matrices of a specific dimension apriori). Thus, for each arbitrary $d\in\mathbb N$, each basis vector in the ordered basis $\mathcal B_d$ is an ordered product of $d$ elements of $\mathcal B_1$, this ordered product induced by the order in $\mathcal B_1$. For instance, 
\begin{equation}
   \mathcal B_2 = \{e_1^2, e_1e_2, e_1e_3, \cdots e_1e_8, e_2^2, e_2e_3, \cdots e_2e_8, e_3^2, e_3e_4,  \cdots e_8^2\} \ .  
\end{equation}

\begin{definition}[PBW]
The \emph{universal enveloping algebra} of $\mathfrak{sl}(3)$,
\begin{equation}
    U(\mathfrak{sl}(3)) := \mathrm{Span}_{\mathbb C}(\mathcal B_\infty) \, ,
\end{equation}
cf. \eqref{ordB1}-\eqref{univ_basis}, is defined by the $\mathfrak{su}(3)$-commutation relations for $\mathcal B_1$ and the fact that commutation is a derivation. It is a graded algebra where
  each subspace  
\begin{equation}
	U_d(\mathfrak{sl}(3)) := \mathrm{Span}_{\mathbb C}(\mathcal B_d) \, 
\end{equation}
is the space of \emph{elements of homogeneous degree $d$}. On the other hand, the \emph{degree}  of a general $u\in U(\mathfrak{sl}(3))$ is given by 
\begin{equation}
	\deg(u):=\min\left\{d\in \mathbb N_0:u \in U_{\le d}(\mathfrak{sl}(3))\right\} \ , 
\end{equation}
where 
\begin{equation}
	U_{\le d}(\mathfrak{sl}(3)) := \bigoplus_{m=0}^d U_m(\mathfrak{sl}(3))\, .
\end{equation}
\end{definition}
Thus, for instance, $e_2e_1$ is not homogeneous of degree $2$, that is  $e_2e_1\notin U_2(\mathfrak{sl}(3))$, but  $e_2e_1$ has degree $2$ since $e_2e_1=e_1e_2-[e_1,e_2]=e_1e_2-e_3\in  U_{\leq 2}(\mathfrak{sl}(3))$. 

Therefore, once chosen the ordering \eqref{orderingB1}-\eqref{univ_basis} defining $\mathcal B_\infty$, the linear map 
\begin{equation}\label{betamap}
    \beta:U(\mathfrak{sl}(3))\to Poly(\mathfrak{sl}(3))
\end{equation}
defined in the basis $\mathcal B_\infty$ by
\begin{equation}\label{beta_basis_B_infty}
 \beta[e_{j_1}...e_{j_d}](X) = (e_{j_1},X)...(e_{j_d},X) \ \ \ \ \forall X\in \mathfrak{sl}(3) \, ,
\end{equation}
is an isomorphism of vector spaces\footnote{This is not a canonical isomorphism $U(\mathfrak{sl}(3))\to Poly(\mathfrak{sl}(3))$ since it depends on a given but not canonical choice of basis for $U(\mathfrak{sl}(3))$, and is obviously not an algebra homomorphism.}, which breaks down into isomorphisms\footnote{Looking at the inverses of \eqref{beta_basis_B_infty} and \eqref{isod}, these are given by a choice of \emph{ordered} basis for each $\mathcal B_d$ whose elements are \emph{ordered} products of elements in $\mathcal B_1$.  In the context of affine systems, where $\mathcal B_1=\{x_i,\partial /\partial x_i\}$, this is also referred to as the ordering problem in quantization.} 
\begin{equation}\label{isod}
    \beta|_{U_d(\mathfrak{sl}(3))}:U_d(\mathfrak{sl}(3))\to Poly_d(\mathfrak{sl}(3)) \ , \ \ \forall d\in\mathbb N_0 \, . 
\end{equation}
In particular, for $\omega_{(p,q)}$ as in \eqref{high-weight}, we have
\begin{equation}\label{charac-omega_p}
	\begin{cases}
		\beta[2T_3](-i\omega_{(p,q)}) = p\\
		\beta[2U_3](-i\omega_{(p,q)}) = q
	\end{cases}\, .
\end{equation}

Now, looking at the (extension of the) adjoint action
\begin{equation}\label{adjactU}
	SU(3)\times U(\mathfrak{sl}(3))\ni (g,u)\mapsto Ad_g(u) \equiv gug^{-1}\in U(\mathfrak{sl}(3)) \, ,
\end{equation}
we have that
$U_d(\mathfrak{sl}(3))$ is not $SU(3)$-invariant because in general the action of $g\in SU(3)$ on $u\in U_d(\mathfrak{sl}(3))$ adds monomials of lower degrees. However, this action never adds monomials of higher degrees, hence $U_{\le d}(\mathfrak{sl}(3))$ is $SU(3)$-invariant. 
\begin{proposition}
    The linear map $\beta$, cf. \eqref{betamap}-\eqref{beta_basis_B_infty}, is not  $SU(3)$-equivariant. 
\end{proposition}
\begin{proof}
   $\beta$ gives an isomorphism between $U_d(\mathfrak{sl}(3))$ and $Poly_d(\mathfrak{sl}(3))$, cf. \eqref{isod}, and $Poly_d(\mathfrak{sl}(3))$ is $SU(3)$-invariant but $U_d(\mathfrak{sl}(3))$ is not. 
\end{proof}

On the other hand, defining the natural projection 
\begin{equation}\label{pid}
\pi_d:U(\mathfrak{sl}(3))\to U_d(\mathfrak{sl}(3)) \ , \ \ u\mapsto \pi_d(u) \, , 
\end{equation}
we have the following proposition. 

\begin{proposition}\label{equivarbeta}
For each $d \in \mathbb N_0$, the map
\begin{equation}
\begin{aligned}
	\beta_d:U_{\le d}(\mathfrak{sl}(3))&\to Poly_d(\mathfrak{sl}(3)) \ , \\
    u&\mapsto \beta_d[u] := \beta[\pi_d(u)] \ ,
\end{aligned}     
\end{equation}
is a linear $SU(3)$-equivariant surjection.
\end{proposition}
\begin{proof}
The statement is trivial for $d = 0$. For $d>0$, linearity and surjectivity are immediate, so we prove only equivariance. Note that, for any $u \in \mathfrak{sl}(3)\equiv U_1(\mathfrak{sl}(3))$,
\begin{equation}
    \beta_1[Ad_g(u)] = (Ad_g(u),\, \cdot\, ) = (u,Ad_{g^{-1}}(\,\cdot\,)) = (\beta_1[u])^g\, , \ \forall g\in SU(3)  .
\end{equation}
For $d>1$, we have already argued above that $U_{\le d-1}(\mathfrak{sl}(3))$ is invariant, thus $\beta_d[Ad_g(u)] = 0$, $\forall u \in U_{\le d-1}(\mathfrak{sl}(3))$, $\forall g\in SU(3)$. On the other hand, if $u = e_{j_1}...e_{j_d}$ is an element of $\mathcal B_d$, we have
\begin{equation}
    Ad_g(u) = \sum_{k_1,...,k_d}D^{(1,1)}_{k_1,j_1}(g)...D^{(1,1)}_{k_d,j_d}(g) e_{k_1}...e_{k_d}\, ,
\end{equation}
where $D^{(1,1)}_{k,j}$ are Wigner $D$-functions in the basis $\mathcal B_1$, cf. Definition I.2.6. In general, the indices $k_1,...,k_d$ are not necessarily in increasing order, so the rewriting of $e_{k_1}...e_{k_d}$ in the basis $\mathcal B_\infty$, by applying commutation relations, splits in two parts:
\begin{equation}
    e_{k_1}...e_{k_d} = e_{k_{f_{k_1,...,k_d}(1)}}...e_{k_{f_{k_1,...,k_d}(d)}}+ v_{k_1,...,k_d}\, ,
\end{equation}
where $f_{k_1,...,k_d}\in S_d$ is some permutation that places the indices $k_1,...,k_d$ in increasing order, and $v_{k_1,...,k_d}\in U_{\le d-1}(\mathfrak{sl}(3))$. Therefore
\begin{equation}
    \pi_d(Ad_g(u)) = \sum_{k_1,...,k_d}D^{(1,1)}_{k_1,j_1}(g)...D^{(1,1)}_{k_d,j_d}(g) e_{k_{f_{k_1,...,k_d}(1)}}...e_{k_{f_{k_1,...,k_d}(d)}}\, .
\end{equation}
Using the fact that the product of polynomials is commutative, the application of $\beta$ on $\pi_d(Ad_g(u))$ allows us to leave out the permutations $f_{k_1,...,k_d}$, giving
\begin{equation}
\begin{aligned}
    \beta_d[Ad_g(u)] & = \sum_{k_1,...,k_d}D^{(1,1)}_{k_1,j_1}(g)...D^{(1,1)}_{k_d,j_d}(g) \beta_1[e_{k_1}]...\beta_1[e_{k_d}]\\
    & = \left(\sum_{k_1}D^{(1,1)}_{k_1,j_1}(g)\beta_1[e_{k_1}]\right)...\left(\sum_{k_d}D^{(1,1)}_{k_d,j_d}(g)\beta_1[e_{k_d}]\right)\, .
\end{aligned}
\end{equation}
We have already proved that $\beta_1$ is equivariant, so
\begin{equation}
    \beta_d[Ad_g(u)] = \beta_1[e_{j_1}]^g...\beta_1[e_{j_d}]^g = (\beta_1[e_{j_1}]...\beta_1[e_{j_d}])^g = \beta_d[u]^g\, ,
\end{equation}
which proves the equivariance of $\beta_d$.
\end{proof}

Using the commutation relations, one can easily verify the next proposition.

\begin{proposition}\label{prop:point_p_beta}
    The pointwise product of elements of $Poly(\mathfrak{su}(3))$ satisfies
    \begin{equation}\label{point_prod_pi}
	\beta_{\deg(u)+\deg(v)}[uv] = \beta_{\deg(u)}[u]\beta_{\deg(v)}[v]
    \end{equation}
    for every $u,v\in U(\mathfrak{sl}(3))$.
\end{proposition}

For the Poisson bracket, we have the following characterization.

\begin{proposition}\label{prop:pois_b_beta}
    The Poisson bivector ${\Pi}_{\mathfrak g}$ defines a Poisson bracket $\{\cdot,\cdot\}$ on $Poly(\mathfrak{su}(3))$ satisfying
    \begin{equation}\label{pois_b_pi}
	\Big\{\beta_{\deg(u)}[u],\beta_{\deg(v)}[v]\Big\} = \beta_{\deg(u)+\deg(v)-1}[uv - vu]
\end{equation}
    for every $u,v \in U(\mathfrak{sl}(3))$.
\end{proposition}
\begin{proof}
    It is immediate that \eqref{pois_b_pi} is skew-symmetric. We will show now that it is a biderivation. For any $u,\tilde u,v \in U(\mathfrak{sl}(3))$, let $d = \deg(u)+\deg(\tilde u)+\deg(v)$. By \eqref{point_prod_pi}, we have
    \begin{equation}
    \begin{aligned}
        \beta_{\deg(u)}[u]\beta_{\deg(\tilde u)}[\tilde u] &= \beta_{\deg(u)+\deg(\tilde u)}[u\tilde u] \ \ \implies \\
         \left\{\beta_{\deg(u)}[u]\beta_{\deg(\tilde u)}[\tilde u],\beta_{\deg(v)}[v]\right\} &= \left\{\beta_{\deg(u)+\deg(\tilde u)}[u\tilde u],\beta_{\deg(v)}[v]\right\}\\ =\beta_{d-1}[u\tilde uv-vu\tilde u] &= \beta_{d-1}[u(\tilde uv-v\tilde u)]+\beta_{d-1}[(uv-vu)\tilde u]\, ,
    \end{aligned}
    \end{equation}
    and again using \eqref{point_prod_pi},
    \begin{equation}
    \begin{aligned}
        \beta_{d-1}[u(\tilde uv-v\tilde u)] & = \beta_{\deg(u)}[u]\beta_{\deg(\tilde u)+\deg(v)-1}[\tilde uv-v\tilde u]\\
        & = \beta_{\deg(u)}[u]\left\{\beta_{\deg(\tilde u)}[\tilde u],\beta_{\deg(v)}[v]\right\}\, ,
    \end{aligned}
    \end{equation}
    \begin{equation}
    \begin{aligned}
        \beta_{d-1}[(uv-vu)\tilde u] & = \beta_{\deg(u)+\deg(v)-1}[uv-vu]\beta_{\deg(\tilde u)}[\tilde u]\\
        & = \left\{\beta_{\deg(u)}[u],\beta_{\deg(v)}[v]\right\}\beta_{\deg(\tilde u)}[\tilde u]\, ,
    \end{aligned}
    \end{equation}
    thus \eqref{pois_b_pi} is a derivation in the first coordinate. Since it is skew-symmetric, it is a biderivation.

    To finish, we will verify that \eqref{pois_b_pi} matches the Poisson bracket of ${\Pi}_{\mathfrak g}$ for linear polynomials, and the biderivation property will imply equality for polynomials of any degree. For the linear coordinates $(x_1,...,x_8)$ in the basis $\{E_1,...,E_8\}$, we have
    \begin{equation*}
    \begin{aligned}
        x_j = \tr(E_j^\dagger\, \cdot\,) &= -\tr(E_j \,\cdot \,) = -\beta_1[E_j] \ \ \implies \\
        \left\{x_j,x_k\right\} =\beta_1[E_jE_k-E_kE_j] &= \sum_{l=1}^8c^l_{jk}\beta_1[E_l] = \sum_{l=1}^8c^l_{kj}x_l = \Pi_{\mathfrak g}(dx_j,dx_k) \ . 
    \end{aligned}
    \end{equation*}
    Therefore, 
    \begin{equation}
     \{f,h\} = {\Pi}_{\mathfrak g}(df,dh) \ ,
     \end{equation}
    for every $f,h \in Poly(\mathfrak{su}(3))$.
\end{proof}

Finally, we shall also make use of the \emph{symmetrization} linear map
\begin{equation}\label{symm_def}
\begin{aligned}
    S:Poly(\mathfrak{sl}(3))&\to U(\mathfrak{sl}(3)) \, , \\
S(\beta_1[e_{j_1}]...\beta_1[e_{j_d}]) & = \dfrac{1}{d!}\sum_{f\in S_d}e_{j_{f(1)}}...e_{j_{f(d)}}\, ,
\end{aligned}    
\end{equation} 
where $S_d$ is the symmetric group.

\begin{proposition}\label{prop:S_equiv}
    The symmetrization map $S$ is equivariant. Also, for every polynomial $f \in Poly_d(\mathfrak{sl}(3))$, we have $S(f) \in U_{\le d}(\mathfrak{sl}(3))$ and $\beta_{d}[S(f)] = f$.
\end{proposition}
\begin{proof}
From Proposition \ref{equivarbeta},
\begin{equation}
    (\beta_1[e_{j_1}]...\beta_1[e_{j_d}])^g = \sum_{k_1,...,k_d}D_{k_1,j_1}^{(1,1)}(g)...D_{k_d,j_d}^{(1,1)}(g)\beta_1[e_{k_1}]...\beta_1[e_{k_d}]\, .
\end{equation}
Applying $S$, we obtain
\begin{equation}
    S((\beta_1[e_{j_1}]...\beta_1[e_{j_d}])^g) = \dfrac{1}{d!}\sum_{k_1,...,k_d}\sum_{f \in S_d}D_{k_1,j_1}^{(1,1)}(g)...D_{k_d,j_d}^{(1,1)}(g) e_{k_{f(1)}}...e_{k_{f(d)}}\, .
\end{equation}
The product of Wigner $D$-functions is obviously commutative, so
\begin{equation}
    \begin{aligned}
        & \sum_{k_1,...,k_d}D_{k_1,j_1}^{(1,1)}(g)...D_{k_d,j_d}^{(1,1)}(g) e_{k_{f(1)}}...e_{k_{f(d)}}\\
        & \hspace{3 em} = \sum_{k_1,...,k_d}\left(D_{k_{f(1)},j_{f(1)}}^{(1,1)}(g)e_{k_{f(1)}}\right)...\left(D_{k_{f(d)},j_{f(d)}}^{(1,1)}(g)e_{k_{f(d)}}\right)\\
        & \hspace{7 em} = Ad_g(e_{j_{f(1)}})...Ad_g(e_{j_{f(d)}}) = Ad_g(e_{j_{f(1)}}...e_{j_{f(d)}})\, .
    \end{aligned}
\end{equation}
Therefore, 
\begin{equation}
\begin{aligned}
S((\beta_1[e_{j_1}]...\beta_1[e_d])^g) & = \dfrac{1}{d!}\sum_{f\in S_d}Ad_g(e_{j_{f(1)}}...e_{j_{f(d)}})\\
& = Ad_g\left(S(\beta_1[e_{j_1}]...\beta_1[e_d])\right)\, .  
\end{aligned}
\end{equation}
This proves the equivariance of $S$. The remaining of the statement follows straightforwardly from the definition.
\end{proof}

We can now prove Proposition \ref{prop:c3_pol}.

\begin{proof}[Proof of Proposition \ref{prop:c3_pol}]

Since $\tau$ is homogeneous of degree $3$, we have $\tau = \beta_3[S(\tau)]$ from Proposition \ref{prop:S_equiv}. From equivariance of $S$ and $\beta_3$, cf. Propositions \ref{equivarbeta} and \ref{prop:S_equiv}, $\tau$ is fixed by $SU(3)$ if and only if $S(\tau)$ is fixed by $SU(3)$. But $S(\tau)$ is (proportional to) the cubic Casimir operator of $\mathfrak{sl}(3)$, cf. (I.B.3) or e.g. \cite[eq. (7.31)]{grein}, so this completes the proof of $SU(3)$-invariance for $\tau$.

For the separation of orbits, note that
\begin{equation}
    \tau(\xi_{(x,y)}) = 2x^3+3x^2y-3xy^2-2y^3
\end{equation}
for every $\xi_{(x,y)}\in \overline{\mathcal F}$. Taking
\begin{equation}\label{tau_F_elipse}
    f(x,y) = 2x^3+3x^2y-3xy^2-2y^3\ , \ \ h(x,y) = x^2+xy+y^2\, ,
\end{equation}
the critical points of $\tau|_{\overline{\mathcal F}}$ are the critical points of the restriction of $f$ to the ellipse $h=1$ in the first quadrant, cf. \eqref{par_xy}. By the method of Lagrange multipliers, we want to solve for $\lambda \in \mathbb R$ and $x,y\ge 0$, the system:
\begin{equation}
    \begin{cases}
        2x^2+2xy-y^2 = \lambda(2x+y)\\
        x^2-2xy-2y^2 = \lambda(x+2y)
    \end{cases}\implies \hspace{1 em} x^2-y^2 = \lambda(x+y)\, .
\end{equation}
There are two kind of solutions: $x+y = 0$, which lies outside the first quadrant, and $x+y \ne 0$, which implies $xy = 0$, meaning the critical point must be an endpoint of $\overline{\mathcal F}$. Therefore, $\tau|_{\overline{\mathcal F}}$ is injective.
\end{proof}

\begin{remark}
    As a homogeneous cubic polynomial, $\tau$ is an odd function, so
    \begin{equation}
        \tau(-\xi_{(x,y)}) = -\tau(\xi_{(x,y)})\, .
    \end{equation}
    In particular, $\tau$ vanishes on the mesonic orbit, cf. (I.2.76). This is aligned with the fact that the cubic Casimir operator $S(\tau)$ assumes the form $C(p,q)\mathds 1$ in the representation $(p,q)$ with $C(p,q) = -C(q,p)$.
\end{remark}

\subsection{Universal correspondences for general quark systems}

For a dominant weight $\omega$ of $\mathfrak{su}(3)$, let $\mathcal H_\omega$ be a quark system with highest weight $\omega$.\footnote{For simplicity, we shall often denote a quantum quark system with highest weight $\omega=\omega_{(p,q)}$, cf. Definitions I.4.6 and I.5.6, simply by its Hilbert space $\mathcal H_\omega$.} If $\omega$ is proportional to $\xi\in \overline{\mathcal Q}$, that is, if $\omega = \norm{\omega}\,\xi$, then the quantum quark system $\mathcal H_\omega$ admits symbol correspondences to $\mathcal O_\xi$, cf. Theorems I.4.8 and I.5.9. As suggested by the previous subsection, it will be useful to work on the universal algebra, so we pullback symbol correspondences to $U(\mathfrak{sl}(3))$ via the irreducible representation $\rho_{\omega}:U(\mathfrak{sl}(3))\to \mathcal B(\mathcal H_{\omega})$ of the universal enveloping algebra on $\mathcal H_\omega$ which is induced by the irreducible representation of $SU(3)$ on $\mathcal H_{\omega}$ in the natural way.

\begin{definition}\label{def:univ_corresp}
Given a dominant weight $\omega = \norm{\omega}\,\xi$, with $\xi \in \overline{\mathcal Q}$, a \emph{universal correspondence for $\omega$}, or simply a \emph{universal correspondence}, is a map
\begin{equation}
	w: U(\mathfrak{sl}(3))\to Poly(\mathcal O_\xi): u\mapsto w[u]
\end{equation}
that factors through a symbol correspondence $W^\omega:\mathcal B(\mathcal H_\omega)\to Poly(\mathcal O_\xi)$ and the irrep $\rho_\omega$ of $U(\mathfrak{sl}(3))$ on $\mathcal H_\omega$, as shown in the diagram below:
\begin{equation}\label{univ_corresp_w}
	\begin{tikzcd}
		U(\mathfrak{sl}(3)) \arrow[r, "\rho_\omega"] \arrow[rr, bend right, "w"'] & \mathcal B(\mathcal H_\omega) \arrow[r, "W^\omega"] & Poly(\mathcal O_\xi)
	\end{tikzcd}
\end{equation}
\end{definition}

\begin{remark}\label{infdimker}
Since $U(\mathfrak{sl}(3))$ is infinite dimensional and $\mathcal B(\mathcal H_\omega)$ is finite dimensional, each $\rho_{\omega}$ (and hence also $w$) has an infinite dimensional kernel, which is a primitive ideal of $U(\mathfrak{sl}(3))$ by definition. Thus, universal correspondences are particular instances of equivariant linear maps $U(\mathfrak{sl}(3))\to Poly(\mathcal O_\xi)$ whose kernels are elements of $\mathrm{Prim}\,U(\mathfrak{sl}(3))$.
\end{remark}

\begin{remark}\label{rmk:mq-only}
    In this way, according to Definition \ref{def:univ_corresp}, for a classical mixed-quark system $\mathcal O_\xi \simeq \mathcal E$, $\xi \in \mathcal Q$, we only consider correspondences from mixed-quark systems $\mathcal H_\omega$, with $\omega=\omega_{(p,q)}$ satisfying $pq\ne 0$, cf. Definition I.5.6.
\end{remark}

A family of correspondences of particular interest to us is the Berezin family. The projection $\Pi_>\in \mathcal B(\mathcal H_\omega)$ onto the highest weight subspace is an operator kernel that gives a Berezin correspondence $B^\omega:\mathcal B(\mathcal H_\omega)\to Poly(\mathcal O_\xi)$, $A\mapsto B^\omega_A$ ,  s.t. 
\begin{equation}\label{ber_corresp}
	B_A^\omega(Ad_g(\xi)) = \tr(A\Pi_>^g)\, ,
\end{equation}
for $\omega = \norm{\omega}\,\xi$, cf. Proposition I.4.19, Remark I.4.20 and Theorem I.5.24.

\begin{definition}\label{defn:univBerezin}
	Given a dominant weight $\omega= \norm{\omega}\,\xi$, $\xi \in \overline{\mathcal Q}$, the \emph{universal Berezin correspondence  for $\omega$} is the universal correspondence $b:U(\mathfrak{sl}(3))\to Poly(\mathcal O_\xi)$ obtained from $B^\omega$ given by \eqref{ber_corresp} according to Definition \ref{def:univ_corresp}.
\end{definition}

The map $\beta$ obtained from PBW theorem is very pertinent to describe universal Berezin correspondences.

\begin{proposition}\label{prop:B=b}
	The universal Berezin correspondence $b:U(\mathfrak{sl}(3))\to Poly(\mathcal O_\xi)$ for $\omega = \norm{\omega}\,\xi$, $\xi \in \overline{\mathcal Q}$, is given by
	\begin{equation*}
		b[u](Ad_g(\xi)) = \beta\!\left[Ad_{g^{-1}}(u)\right]\!(-i\omega)
	\end{equation*}
	for every $u\in U(\mathfrak{sl}(3))$ and $g \in SU(3)$.
\end{proposition}
\begin{proof}
	Let $\vb*e_>\in \mathcal H_\omega$ be a highest weight unit vector. By definition,
	\begin{equation}
		b[u](Ad_g(\xi)) = \ip{\vb*e_>}{\rho_\omega\left(Ad_{g^{-1}}(u)\right)\vb*e_>} = \beta\!\left[Ad_{g^{-1}}(u)\right]\!(-i\omega)\, ,
	\end{equation}
	where the last equation comes from decomposing $Ad_{g^{-1}}(u)$ in the basis $\mathcal B_\infty$, cf. \eqref{univ_basis}, and using \eqref{charac-omega_p}.
\end{proof}

By construction, a universal correspondence $w: U(\mathfrak{sl}(3))\to Poly(\mathcal O_\xi)$ for the weight $\omega = \norm{\omega}\,\xi$, $\xi \in\overline{\mathcal Q}$, induces a twisted product $\star$ on the image of $w$ by
\begin{equation}
	w[u]\star w[v] = w[uv]
\end{equation}
for every $u,v \in U(\mathfrak{sl}(3))$, so that, recalling Remark \ref{infdimker}, this is the same product induced by the symbol correspondence $W:\mathcal B(\mathcal H_\omega)\to Poly(\mathcal O_\xi)$ that generates $w$. With that in mind, we also import the notion of Stratonovich-Weyl correspondences for the universal ones.

\begin{definition}\label{def:univ_sw}
A universal correspondence $w: U(\mathfrak{sl}(3))\to Poly(\mathcal O_\xi)$ is of type \emph{Stratonovich-Weyl} if, for every $u,v \in U(\mathfrak{sl}(3))$, 
\begin{equation}
	\int_{\mathcal O_\xi}w[uv](\vb*\varsigma)d\vb*\varsigma =  \int_{\mathcal O_\xi}w[u](\vb*\varsigma)w[v](\vb*\varsigma)d\vb*\varsigma \ . 
\end{equation}
\end{definition}

Thus, Proposition I.3.13 translates for universal correspondences as:

\begin{theorem}
	No universal Berezin correspondence is Stratonovich-Weyl.
\end{theorem}

\section{Asymptotic analysis for general quark systems}\label{sec:asymp_anal}

The first problem we face in order to work out semiclassical analysis for quark systems is the identification of pertinent sequences of quantum quark systems\footnote{Recall that each quantum quark system is identified by a pair of natural numbers, $(p,q)$, so in principle we could have bi-sequences of such systems, and this will be explored in the next section.}.  We need to find some principle that recovers the case of spin systems, where this problem does not exist at all, cf. \cite{RS}. Such a principle should align with the fact that the orbits $\mathcal O_{(1,0)}$ and $\mathcal O_{(0,1)}$, being isomorphic to $\mathbb CP^2$, correspond to classical pure-quark systems and only admit correspondences from quantum pure-quark systems $(p,0)$ and $(0,p)$. This restriction, together with Definition \ref{def:univ_corresp}, points to a reasonable principle: given $\xi \in\overline{\mathcal Q}$ so that $\mathcal O_\xi\subset\{\mathcal S^7,\widehat\Pi_{\mathfrak g}\}$, we shall look at the ray from the origin in the direction of $\xi$, in the lattice of dominant weights; or in other words, we shall look at the sequence of dominant weights $(s\omega_\xi)_{s\in\mathbb N}$.

\subsection{Rays of universal correspondences: fuzzy orbits}

\begin{definition}\label{def:pois_ray}
	Given $\xi \in\overline{\mathcal Q}$ so that $\mathcal O_\xi\subset\{\mathcal S^7,\widehat\Pi_{\mathfrak g}\}$, a \emph{ray of symbol correspondences attached to $\xi$}, or in short, a  \emph{$\xi$-ray of symbol correspondences} is a sequence of symbol correspondences $$\left(W^{s\omega_\xi}:\mathcal B(\mathcal H_{s\omega_\xi})\to Poly(\mathcal O_\xi)\right)_{s\in \mathbb N} \ ,$$ cf. Definition \ref{intrad-defn}. Accordingly, a \emph{$\xi$-ray of universal correspondences}   is a sequence $$\left(w^s_\xi:U(\mathfrak{sl}(3))\to Poly(\mathcal O_\xi)\right)_{s\in \mathbb N} \ ,$$ where each $w^s_\xi$ is an universal correspondence for $s\,\omega_\xi$ according to Definition \ref{def:univ_corresp}. If $(\mathcal W^s_\xi)_{s\in\mathbb N}$ denotes the sequence of images of $(W^{s\omega_\xi})_{s\in\mathbb N}$ or $(w^s_\xi)_{s\in\mathbb N}$, we have the induced \emph{$\xi$-ray of twisted products} $(\star_\xi^s)_{s\in\mathbb N}$, where each $\star_\xi^s: \mathcal W^s_\xi\times\mathcal W^s_\xi\to\mathcal W^s_\xi$ is given by 
\begin{equation}
w^s_\xi[u]\star^s_\xi w^s_\xi[v]=w^s_\xi[uv] \ , \ \ \forall u,v\in U(\mathfrak{sl}(3)) \ . 
\end{equation}
Then, the pair sequence $\big(\mathcal W^s_\xi,\star_\xi^s\big)_{\!s\in\mathbb N}$ shall be called a    \emph{$\xi$-ray of twisted algebras}, or in short a 
\emph{fuzzy $\xi$-orbit}, denoted 
\begin{equation}\label{fuzzyOxi}
\mathfrak{W}(\mathcal O_\xi)= \big(\mathcal W^s_\xi,\star_\xi^s\big)_{\!s\in\mathbb N} \ .
\end{equation}
\end{definition}

The restriction stated in Remark \ref{rmk:mq-only}, that we do not consider correspondences from quantum pure-quark systems to classical mixed-quark systems, now reverberates in the fact that the sequence of images of a $\xi$-ray of correspondences, $(\mathcal W^s_\xi)_{s\in\mathbb N}$,  asymptotically covers $Poly(\mathcal O_\xi)$, as shown in Lemma \ref{lemma:f_image_orbit} below.

\begin{notation}\label{not:sub-irrep}
	For a linear space $V$ carrying a representation of $SU(3)$, we denote by $V^{\vb*a}$ the maximal invariant subspace of $V$ where $SU(3)$ acts via (copies of) the irrep $\vb*a$.
\end{notation}

\begin{lemma}\label{lemma:f_image_orbit}
    Given $\xi \in \overline{\mathcal Q}$ and $f \in Poly(\mathcal O_\xi)$, there exists $s_0\in\mathbb N$ such that, for every fuzzy $\xi$-orbit, we have $f\in\mathcal W^s_\xi$, $\forall s\ge s_0$.
\end{lemma}
\begin{proof}
    Without loss of generality, we assume $f \in Poly(\mathcal O_\xi)^{(a,b)}$. If $\xi \in \mathcal Q$, recall \eqref{defr}, hence $(sp_1,sq_1) \equiv s\, \omega_\xi$. Then, we can conclude $f$ is in the image of $w^s_\xi$ if
    \begin{equation}
        \mathfrak m(sp_1,sq_1;a,b) = m(a,b)\, ,
    \end{equation}
    cf. Notation I.5.8. From (I.5.17)-(I.5.18), it is sufficient to have
    \begin{equation}
        s\min\{p_1,q_1\}\ge \max\{a,b\}\, .
    \end{equation}

    Therefore, it is sufficient to take 
    \begin{equation}\label{suff-s0}
       s_0=\max\{a,b\} \ ,  
    \end{equation}
    to have $f$ in the image of $w^s_\xi$  for every $s\ge s_0$. Now, if $\xi \in \overline{\mathcal Q}\setminus \mathcal Q$, then $f$ is nonzero only if $a = b$. Since $(a,a)$ is multiplicity free in $Poly(\mathcal O_\xi)$, and
    \begin{equation}
        (s,0)\otimes (0,s) = \bigoplus_{n=0}^s(n,n)\, ,
    \end{equation}
    $f$ is in the image of $w^s_\xi$ for every $s\ge \max\{a,b\}$, cf. \eqref{suff-s0}.
\end{proof}

In view of the previous lemma, we introduce the following. 
    
\begin{definition}\label{PoissonOxi}
    A $\xi$-ray of (symbol or universal) correspondences is of \emph{Poisson type}, or equivalently a fuzzy $\xi$-orbit 
    is of \emph{Poisson type}, if the $\xi$-ray of twisted products $(\star^s_\xi)_{s\in \mathbb N}$ satisfies
	\begin{eqnarray}\label{Poisson-eq}
		\lim_{s\to\infty}\norm{f_1\star^s_\xi f_2-f_1f_2}_\xi &=& 0 \ , \\ \lim_{s\to\infty}\norm{sr(\xi)[f_1,f_2]_{\star^s_\xi}-i\{f_1,f_2\}}_\xi &=& 0 \ , \label{Poisson-eq-2}
	\end{eqnarray}
    for every $f_1,f_2 \in Poly(\mathcal O_\xi)$, where $[\cdot,\cdot]_{\star^s_\xi}$ is the commutator of $\star^s_\xi$ and $\{\cdot,\cdot\}$ is the Poisson bracket defined by $\Omega_\xi=\widehat\Pi_{\mathfrak g}|_{\mathcal O_\xi}$, cf. \eqref{foliS7}. In this case, we denote
    \begin{equation}\label{WtoP}
        \mathfrak{W}(\mathcal O_\xi) \xrightarrow{ \ \sim \ } Poly(\mathcal O_\xi,\Omega_\xi) \ ,  
    \end{equation}
    cf. \eqref{fuzzyOxi}, 
    where the r.h.s.~denotes the Poisson algebra of polynomials on $\mathcal O_\xi$. 
\end{definition}

\begin{remark}\label{rmk:sp_resc}
    The effective asymptotic parameter for the commutator $[\cdot,\cdot]_{\star_\xi^s}$ , for each rational orbit $\mathcal O_\xi\subset \{\mathcal S^7,\widehat\Pi_{\mathfrak g}\}$, is actually $sr(\xi)$, cf. \eqref{Poisson-eq-2}, where $s\in \mathbb N$ and $r(\xi)$ is the integral radius of $\xi$, that is, $\omega_\xi=r(\xi)\xi$, cf. Definition \ref{intrad-defn}. However, since $r(\xi)$ is fixed, for each $\xi$-ray of correspondences $(w_\xi^s)_{s\in\mathbb N}$, we can consider $s\in \mathbb N$ as the single asymptotic parameter in \eqref{Poisson-eq-2}, $\forall\xi\in\overline{\mathcal Q}$, since $s\in \mathbb N$ is also the single asymptotic parameter for the product, $\forall\xi\in\overline{\mathcal Q}$, cf. \eqref{Poisson-eq}. This helps to read all asymptotic limits studied as limits of sequences just indexed by $s\in\mathbb N$, $\forall\xi\in\overline{\mathcal Q}$. But in this way, when considering every $\xi\in\overline{\mathcal Q}$ together, in Section \ref{sec:corresp_sphere}, we shall have to resort to a $\xi$-dependent rescaling of the commutator of the twisted product, as in  Definition \ref{def:penc_magoo-p}, further below.
\end{remark}

With the aim of verifying conditions for Poisson, we first explore rays of Berezin correspondences, reproducing some results of \cite{karab} for the particular setting of $SU(3)$. This will make clear a sufficient condition for a $\xi$-ray of universal correspondences to be of Poisson type, and then we will prove this condition is also necessary.

\subsection{On the asymptotics of Berezin fuzzy orbits}

\begin{definition}\label{Berezin-fuzzy}
    For any orbit $\mathcal O_\xi\subset\mathcal S^7$, its $\xi$-ray of twisted algebras defined by the $\xi$-ray of universal Berezin correspondences $(b^s_\xi)_{s\in \mathbb N}$, cf.~Definition \ref{defn:univBerezin} and Proposition \ref{prop:B=b}, and Definition \ref{def:pois_ray}, is called the \emph{Berezin fuzzy $\xi$-orbit}, denoted $\mathfrak{B}(\mathcal O_\xi)$.
\end{definition}

Then, for $\mathfrak{B}(\mathcal O_\xi)$, $\xi \in \overline{\mathcal Q}$, and every $s\in\mathbb N$, consider the \emph{error map} 
\begin{equation}\label{error_map}
\begin{aligned}
    \varepsilon^s_\xi&:U(\mathfrak{sl}(3))\to Poly(\mathcal O_\xi), \ u\mapsto \varepsilon^s_\xi[u] \ , \\
	\varepsilon^s_\xi[u]&:=(sr(\xi))^{-\deg(u)}b^s_\xi[u]-(-i)^{\deg(u)}\beta_{\deg(u)}[u]\big|_{\mathcal O_\xi}\, .
\end{aligned}    
\end{equation}

\begin{proposition}\label{prop:error_conv}
	Every error map $\varepsilon^s_\xi$ is $SU(3)$-equivariant. Also, given $u\in U(\mathfrak{sl}(3))$, there exists $M(u)\ge 0$ such that
    \begin{equation}\label{unif-error}
        \norm{\varepsilon^s_\xi[u]}_\xi\le \dfrac{M(u)}{sr(\xi)} \leq\dfrac{1}{s}\sqrt{\dfrac{3}{2}}M(u) \ ,
    \end{equation}
    for every $\xi \in \overline{\mathcal Q}$ and $1< s\in \mathbb N$.
\end{proposition}
\begin{proof}
	The equivariance is immediate from the equivariance of the maps in the r.h.s.~of \eqref{error_map}. For the upper bound, we use Proposition \ref{prop:B=b} to get
	\begin{equation}
		\varepsilon^s_\xi[u](Ad_g(\xi)) = \sum_{j=1}^{\deg(u)-1}(-i)^j(sr(\xi))^{j-\deg(u)}\beta[\pi_j(Ad_{g^{-1}}(u))](\xi)
	\end{equation}
	for every $g\in SU(3)$. Let $\psi_j:U_{\le \deg(u)}(\mathfrak{sl}(3))\to Poly(\mathcal S^7)$ be given by
    \begin{equation}
        \psi_j[v] = \beta[\pi_j(v)]\big|_{\mathcal S^7}\, .
    \end{equation}
    Take an equivariant inner product on $U_{\le \deg(u)}(\mathfrak{sl}(3))$ so that $u$ has norm $\norm{u}$ and $\psi_j$ has operator norm $\norm{\psi_j}$. Hence $\norm{\psi_j[Ad_{g^{-1}}(u)]}_\infty\le \norm{\psi_j}\norm{u}$ for every $g\in SU(3)$. Using triangular inequality and setting
	\begin{equation}
		M(u)=(\deg(u)-1)\max\{\norm{\psi_j}\norm{u}:j=1,...,\deg(u)-1\}\, ,
	\end{equation}
	we get what we want, since $sr(\xi)>1$, $\forall \xi\in\overline{\mathcal Q}$, $\forall s\geq 2$.
\end{proof}

We can rewrite \eqref{error_map} as 
\begin{equation}
	(sr(\xi))^{-\deg(u)}b^s_\xi[u] = (-i)^{\deg(u)}\,\beta_{\deg(u)}[u]\big|_{\mathcal O_\xi}+\varepsilon^s_\xi[u]\, ,
\end{equation}
so that we get immediately from Proposition \ref{prop:error_conv}, the following:

\begin{corollary}\label{lemma:lim-t-1}
For every $u \in U(\mathfrak{sl}(3))$, we have 
\begin{equation}
    \lim_{s\to\infty}(sr(\xi))^{-\deg(u)}b^s_\xi[u] = (-i)^{\deg(u)}\,\beta_{\deg(u)}[u]\big|_{\mathcal O_\xi}\, .
\end{equation}
\end{corollary}

We shall now use the \emph{symmetrization map} $S$ given by \eqref{symm_def}. Restrictions of elements in $Poly(\mathfrak{su}(3))$ generate $Poly(\mathcal S^7)$ and $Poly(\mathcal O_\xi)$, and likewise, restrictions of $Poly_d(\mathfrak{su}(3))^{\vb*a}$ generate $Poly_d(\mathcal S^7)^{\vb*a}$ and $Poly_d(\mathcal O_\xi)^{\vb*a}$, cf. Notation \ref{not:sub-irrep}. On the other hand, for $f$ either in $Poly_d(\mathcal O_\xi)^{\vb*a}$ or $Poly_d(\mathcal S^7)^{\vb*a}$, we know there exists $\widetilde f \in Poly_d(\mathfrak{su}(3))$ that restricts to $f$. If $\widetilde f$ has any component in $Poly_d(\mathfrak{su}(3))^{\vb*b}$ for $\vb*b\ne \vb*a$, we can subtract it and the restriction still matches $f$, so such $f$ is always the restriction of some element of $Poly_d(\mathfrak{su}(3))^{\vb*a}$. 

Below, we will apply the symmetrization map $S$ on $f\in Poly_d(\mathcal O_\xi)^{\vb*a}$ (eventually also on $f\in Poly_d(\mathcal S^7)^{\vb*a}$), which will be a little abuse of notation for the application of $S$ on the respective $\widetilde f\in Poly_d(\mathfrak{sl}(3))^{\vb*a} \equiv Poly_d(\mathfrak{su}(3))^{\vb*a}$ that restricts to $f$.

Now, by adjoining basis of $Poly_d(\mathfrak{su}(3))^{\vb*a}$ for each $d \in \mathbb N$, we produce a basis of $Poly(\mathfrak{su}(3))^{\vb*a}$ comprised only by homogeneous polynomials. Restricting the elements of such basis to $\mathcal O_\xi$, we obtain a generating set for $Poly(\mathcal O_\xi)^{\vb*a}$, from which we can extract a basis $\{h_1,...,h_m\}$ with $h_j \in Poly_{d(j)}(\mathcal O_\xi)^{\vb*a}$.

From the above considerations and Proposition \ref{prop:S_equiv}, we can define
\begin{equation}
\begin{aligned}
u_j &= (-i)^{-d(j)}S(h_j)\in U_{\le d(j)}(\mathfrak{sl}(3))^{\vb*a} \ , \\
	h_j &= (-i)^{d(j)}\beta_{d(j)}[u_j]\big|_{\mathcal O_\xi}\in Poly_{d(j)}(\mathcal O_\xi)^{\vb*a}\, .
\end{aligned}    
\end{equation}
Let $u_1,...,u_m\in U(\mathfrak{sl}(3))^{\vb*a}$ be chosen in this way.

\begin{lemma}\label{lemma:b^s_basis}
There exists $s_0 \in \mathbb N$ such that
\begin{equation}\label{basis_b_u_j}
	\left\{(sr(\xi))^{-d(1)}b^s_\xi[u_1],...,(sr(\xi))^{-d(m)}b^s_\xi[u_m]\right\}
\end{equation}
is a basis of $Poly(\mathcal O_\xi)^{\vb*a}$ for every $s\geq s_0$.
\end{lemma}
\begin{proof}
Since $\{h_1,...,h_m\}$ is l.i., there are $\vb*\varsigma_1,...,\vb*\varsigma_m\in \mathcal O_\xi$ such that the matrix $H$ with entries $(H)_{j,k} = h_k(\vb*\varsigma_j)$ is non singular. Let $B(s)$ be the matrix with entries $(B(s))_{j,k} = (sr(\xi))^{-d(k)}b^s_\xi[u_k](\vb*\varsigma_j)$. From Corollary \ref{lemma:lim-t-1}, $B(s)$ converges to $H$, so there exists $s_0\in \mathbb N$ such that $B(s)$ is non singular for every $s\geq s_0$, in other words, the set \eqref{basis_b_u_j} is l.i.~for $s\geq s_0$.
\end{proof}

The above Lemma guarantees that we can decompose any element of $Poly(\mathcal O_\xi)^{\vb*a}$ as a linear combination of symbols of fixed elements of the universal algebra, and this simplifies the writing of  twisted products.

\begin{proposition}\label{prop:decomp_b^s_basis}
	If $f\in Poly(\mathcal O_\xi)^{\vb*a}$, then there are $s_0\in \mathbb N$ and $\alpha_j(s)\in \mathbb C$ for $j\in \{1,...,m\}$ and $s\geq s_0$ such that
    \begin{equation}\label{decomp_b^s_basis}
		f = \sum_{j=1}^m\alpha_j(s)(sr(\xi))^{-d(j)}b^s_\xi[u_j] = \sum_{j=1}^m\alpha_j^\infty h_j \ , \ \ \ \alpha_j^\infty=\lim_{s\to\infty}\alpha_j(s)\in\mathbb C\, .
	\end{equation}
\end{proposition}
\begin{proof}
	Of course, $f$ is a linear combination of any basis of $Poly(\mathcal O_\xi)^{\vb*a}$, thus we can write \eqref{decomp_b^s_basis} by Lemma \ref{lemma:b^s_basis} and the construction of $\{h_1,...,h_m\}$. Recalling the proof of Lemma \ref{lemma:b^s_basis},  let $A(s) = B(s)^{-1}$ for $s\geq s_0$, and $F = H^{-1}$, so that
	\begin{equation}
		\alpha_j(s) = \sum_{k=1}^m(A(s))_{j,k}f(\vb*\varsigma_k) \ , \ \ \alpha_j^\infty = \sum_{k=1}^m(F)_{j,k}f(\vb*\varsigma_k)\, .
	\end{equation}
	By continuity of the inversion map, we have $A(s)\to F$, implying $\alpha_j(s)\to\alpha_j^\infty$.
\end{proof}

\begin{theorem}\label{theo:prod_ber_conv} 
For $\xi\in\overline{\mathcal Q}$, let $(\star^s_\xi)_s$ be the sequence of twisted products of the Berezin fuzzy $\xi$-orbit $\mathfrak B(\mathcal O_\xi)$, cf. Definition \ref{Berezin-fuzzy}. 
As $s\to\infty$, the uniform convergences 
\begin{eqnarray}
    f_1\star^s_\xi f_2&\to& f_1f_2 \ , \label{f1f2} \\ 
    sr(\xi)[f_1,f_2]_{\star^s_\xi}&\to& i\big\{f_1,f_2\big\} \label{f1,f2} \ , 
\end{eqnarray}
cf. \eqref{Poisson-eq}-\eqref{Poisson-eq-2}, 
hold for every $f_1,f_2 \in Poly(\mathcal O_\xi)$. 
\end{theorem}
\begin{proof}
We start by proving (\ref{f1f2}). By bilinearity of the products, it is sufficient to show the convergence for $f_j \in Poly(\mathcal O_\xi)^{\vb*a_j}$. Now, we use Proposition \ref{prop:decomp_b^s_basis} to write
\begin{equation}\label{Zk-decomp1}
	f_j = \sum_{k=1}^{m_j} \alpha^j_k(s)(sr(\xi))^{-d_j(k)}b^s_\xi[u_k^j] = \sum_{k=1}^{m_j}\alpha^j_k h_k^j \ , 
\end{equation}
for $s\geq s_0\in \mathbb N$, where 
\begin{eqnarray}
&\displaystyle{\lim_{s\to \infty}}\alpha^j_k(s) = (\alpha_k^j)^\infty\equiv \alpha^j_k\, , \\
&	\displaystyle{\lim_{s\to\infty}}(sr(\xi))^{-d_j(k)}b^s_\xi[u_k^j] = (-i)^{d_j(k)}\beta_{d_j(k)}[u_k^j] = h_k^j \, ,
\end{eqnarray}
cf. Corollary \ref{lemma:lim-t-1}. Therefore, 
\begin{equation}\label{convprod}
\begin{aligned}
     \lim_{s\to\infty}f_1\star^s_\xi f_2 &= \lim_{s\to\infty}\sum_{j=1}^{m_1}\sum_{k=1}^{m_2} \alpha^1_j(s) \alpha^2_k(s)(sr(\xi))^{-(d_1(j)+d_2(k))}b^s_\xi\left[u_j^1u_k^2\right] \\
     &=\sum_{j=1}^{m_1}\sum_{k=1}^{m_2}\alpha^1_j\alpha^2_kh_j^1 h_k^2  = \Bigg(\sum_{j=1}^{m_1}\alpha^1_j h_j^1\Bigg)\Bigg(\sum_{k=1}^{m_2}\alpha^2_k h_k^2\Bigg)
      = f_1f_2\, ,
\end{aligned}
\end{equation}
where we have used Proposition \ref{prop:point_p_beta}. Similarly for proving \eqref{f1,f2}, we have that
\begin{equation}\label{convpb}
	\begin{aligned}
		\lim_{s\to\infty}sr(\xi)\left[f_1,f_2\right]_{\star^s_\xi} & = \lim_{s\to\infty}\sum_{j=1}^{m_1}\sum_{k=1}^{m_2} \alpha^1_j(s)\alpha^2_k (s)\\
		& \hspace{4 em} \times (sr(\xi))^{-(d_1(j)+d_2(k)-1)}b^s_\xi\left[u^1_j u^2_k-u^2_k u^1_j\right]\\
		& = i \sum_{j=1}^{m_1}\sum_{k=1}^{m_2} \alpha_j^1 \alpha_k^2\left\{h_j^1,h_k^2\right\} = i\Big\{f_1,f_2\Big\}\, .
	\end{aligned}
\end{equation}
where we have used Proposition \ref{prop:pois_b_beta}.
\end{proof}

Therefore, according to Definition \ref{PoissonOxi}, we have:

\begin{corollary}\label{cor:ber_ray_pois}
For any $\xi\in \overline{\mathcal Q}$, the Berezin fuzzy $\xi$-orbit is of Poisson type,
\begin{equation}
    \mathfrak{B}(\mathcal O_\xi) \xrightarrow{ \ \sim \ } Poly(\mathcal O_\xi,\Omega_\xi) \ ,  
\end{equation}
or in other words, the $\xi$-ray of Berezin correspondences is of Poisson type.
\end{corollary}

\subsection{First criterion for Poisson: convergence of symbols}

The thread from Corollary \ref{lemma:lim-t-1} to Proposition \ref{prop:decomp_b^s_basis} makes it clear that the proof of Theorem \ref{theo:prod_ber_conv} depends solely on Corollary \ref{lemma:lim-t-1} of Proposition \ref{prop:error_conv}, therefore any $\xi$-ray of universal correspondence satisfying Corollary \ref{lemma:lim-t-1} is of Poisson type. That is, we already have:

\begin{proposition}\label{thmprop1}
    For $\xi \in \overline{\mathcal Q}$, a $\xi$-ray of universal correspondences $(w^s_\xi)_{s\in\mathbb N}$ is of Poisson type if
    \begin{equation}\label{condition1}
        \lim_{s\to\infty}(sr(\xi))^{-\deg(u)}w^s_\xi[u] = (-i)^{\deg(u)}\beta_{\deg(u)}[u]\big|_{\mathcal O_\xi} \ , \ \  \forall u\in U(\mathfrak{sl}(3))\, .
    \end{equation}
\end{proposition}

We now show that \eqref{condition1} is also a  necessary condition for Poisson type.

\begin{lemma}\label{lemma:pois_bound}
If a $\xi$-ray $(w^s_\xi)_{s\in\mathbb N}$ of universal correspondences is of Poisson type, then $((sr(\xi))^{-\deg(u)}w^s_\xi[u])_{s\in\mathbb N}$ is a bounded sequence in $Poly(\mathcal O_\xi)$, $\forall u\in U(\mathfrak{sl}(3))$.
\end{lemma}
\begin{proof}
We prove it by induction on the degree of $u$. First, take $u\in\mathfrak{sl}(3)\equiv U_1(\mathfrak{sl}(3))$ non null and let
\begin{equation}
    N_s = \norm{w^s_\xi[u]}_\xi^{-1}\, ,
\end{equation}
so the sequence $(h_s)$, with $h_s = N_sw^s_\xi[u]$, is in the unit sphere of $Poly(\mathcal O_\xi)^{(1,1)}$. For any $v\in U_1(\mathfrak{sl}(3))$ with
\begin{equation}
    \widetilde v = uv-vu\ne 0\, ,
\end{equation}
we have that $\widetilde v\in U_1(\mathfrak{sl}(3))$ as well and, by Schur's Lemma, the sequences $(f_s)$ and $(\widetilde f_s)$ given by $f_s = N_sw^s_\xi[v]$ and $\widetilde f_s = N_sw^s_\xi[\widetilde v]$ are never zero and bounded.

Now, let $(\star^s_\xi)$ be the sequence of twisted products induced by $(w^s_\xi)$. The Poisson condition implies that the sequence $(C^s_\xi)_s$ of operators
\begin{equation}
	C^s_\xi:Poly(\mathcal O_\xi)^{(1,1)}\wedge Poly(\mathcal O_\xi)^{(1,1)}\to Poly(\mathcal O_\xi):f\wedge h \mapsto sr(\xi)[f,h]_{\star^s_\xi}
\end{equation}
converges pointwisely to $f\wedge h\mapsto i\{f,h\}$. By the Uniform Boundedness Principle, $(C^s_\xi)_s$ is uniformly bounded. Thus
\begin{equation}
	\norm{C^s_\xi(f_s\wedge h_s)}_\xi = \norm{sr(\xi)N_s^2w^s_\xi[\widetilde v]}_\xi = sr(\xi)N_s\norm{\widetilde{f_s}}_\xi
\end{equation}
is bounded on $s$. Hence $N_s \in O(1/s)$. This shows the claim for $u\in U_1(\mathfrak{sl}(3))$.

To complete the induction, suppose the claim holds for every element of the universal enveloping algebra with degree $\le d$. If $v_j \in U_{\le d_j}(\mathfrak{sl}(3))^{\vb*a_j}$ for $j\in\{1,2\}$, with $1\le d_1,d_2\le d$, then $v_1v_2 \in U_{\le d_1+d_2}(\mathfrak{sl}(3))$. Again, the Poisson condition implies that the sequence $(T^s_\xi)_{s}$ of operators
\begin{equation}\label{tp_oper}
	T^s_\xi:Poly(\mathcal O_\xi)^{\vb*a_1}\otimes Poly(\mathcal O_\xi)^{\vb*a_2}\to Poly(\mathcal O_\xi): f\otimes h\mapsto f\star^s_\xi h
\end{equation}
converges pointwisely to $f\otimes h\to fh$, so $(T^s_\xi)_{s}$ is uniformly bounded. Therefore
\begin{equation}
	\norm{(sr(\xi))^{-(d_1+d_2)}T^s_\xi(w^s_\xi[v_1]\otimes w^s_\xi[v_2])}=(sr(\xi))^{-(d_1+d_2)}\norm{w^s_\xi[v_1v_2]}_\xi
\end{equation}
is bounded on $s$. By writing $u \in U_{\le d+1}(\mathfrak{sl}(3))$ as a linear combination of products of elements of degrees $\le d$, we conclude that the claim also holds for every element of $U_{\le d+1}(\mathfrak{sl}(3))$.
\end{proof}

\begin{lemma}
	If a $\xi$-ray $(w^s_\xi)_{s\in\mathbb N}$ of universal correspondences is of Poisson type, then
	\begin{equation}
		\lim_{s\to \infty}(sr(\xi))^{-1}w^s_\xi[u] = -i\beta_1[u]\big|_{\mathcal O_\xi} \ , \ \  \forall u\in U(\mathfrak{sl}(3))\, .
	\end{equation}
\end{lemma}
\begin{proof}
	Let $h_s = (sr(\xi))^{-1}w^s_\xi[u]$. By the previous lemma, $(h_s)$ is a bounded sequence in $Poly(\mathcal O_\xi)^{(1,1)}$. Let $(\widetilde h_n)$, $\widetilde h_n = h_{s_n}$, be any convergent subsequence, $\widetilde h_n \to h$. We want to prove that $h = -i\beta_1[u]|_{\mathcal O_\xi}= -i\beta[u]|_{\mathcal O_\xi}$. To do so, we will prove that
	\begin{equation}
		i\{h,f\} = i\{-i\beta[u]|_{\mathcal O_\xi},f\}
	\end{equation}
    for every $f\in Poly_1(\mathcal O_\xi)$, which allows us to conclude that $h$ and $-i\beta[u]|_{\mathcal O_\xi}$ have the same Hamiltonian vector field, so they differ by a constant\footnote{Recall that $\mathcal O_\xi$ is connected.}; since both functions lies in $Poly(\mathcal O_\xi)^{(1,1)}$, whose only constant function is identically $0$, the functions must coincide. Thus, let $X_u$ be the vector field that represents the action of $u$ on $C^\infty(\mathcal O_\xi)$, naturally induced by the $SU(3)$-action, and take
    \begin{equation}
        \widetilde f = X_u(f)\, .
    \end{equation}
    Then, let $v_n \in \mathfrak{sl}(3) \equiv U_1(\mathfrak{sl}(3))$ be such that $f = w^{s_n}_\xi[v_n]$. By equivariance of $w^{s_n}_{\xi}$, we have that
    \begin{equation}
        \widetilde v_n = uv_n-v_nu
    \end{equation}
    satisfies $\widetilde f = w^{s_n}_\xi[\widetilde v_n]$. Then 
	\begin{equation}\label{[h_n,f]}
		s_nr(\xi)[h_n,f]_{\star^{s_n}_\xi} = w^{s_n}_\xi[uv_n-v_nu] = w^{s_n}_\xi[\widetilde v] = \widetilde f,
	\end{equation}
	where $\star^s_\xi$ is the twisted product induced by $w^s_\xi$ as usual.
    We can rewrite \eqref{[h_n,f]} as
    \begin{equation}
        \widetilde f = s_nr(\xi)[h_n-h,f]_{\star^{s_n}_\xi}+sr(\xi)[h,f]_{\star^{s_n}_\xi}\, .
    \end{equation}
    As we argued in the proof of the previous lemma, the Poisson hypothesis implies that the sequence of operators $(F_n)$,
    \begin{equation}
        F_n:Poly_1(\mathcal O_\xi) \to Poly_1(\mathcal O_\xi):\widetilde h\mapsto s_nr(\xi)[\widetilde h,f]_{\star^{s_n}_\xi}\, ,
    \end{equation}
    is uniformly bounded, so
	\begin{equation}
		\widetilde f = \lim_{n\to\infty}s_nr(\xi)[h_n-h,f]_{\star^{s_n}_\xi}+\lim_{n\to\infty}s_nr(\xi)[h,f]_{\star^{s_n}_\xi} = i\{h,f\}\, .
	\end{equation}
	Now, let $v = S(f)$ and $\widetilde v = S(\widetilde f)$, so $f = \beta_1[v]|_{\mathcal O_\xi}$, $\widetilde v = uv-vu$ and
    \begin{equation}
        \widetilde f = \beta_1[\widetilde v]|_{\mathcal O_\xi} = \beta_1[uv-vu]|_{\mathcal O_\xi}\, ,
    \end{equation}
    cf. Proposition \ref{prop:S_equiv}. From Proposition \ref{prop:pois_b_beta}, we have
	\begin{equation}
		i\{h,f\} = \widetilde f = \beta_1[uv-vu]|_{\mathcal O_\xi} = \{\beta_1[u],\beta_1[v]\}|_{\mathcal O_\xi} = i\{-i\beta[u]|_{\mathcal O_\xi},f\}\, .
	\end{equation}
	Therefore, every convergent subsequence of the bounded sequence $(h_s)$ converges to $-i\beta[u]|_{\mathcal O_\xi}$, which means the sequence itself converges to $-i\beta[u]|_{\mathcal O_\xi}$.
\end{proof}

\begin{proposition}\label{thmprop2}
	If a $\xi$-ray $(w^s_\xi)_{s\in\mathbb N}$ of universal correspondences is of Poisson type, then  \eqref{condition1} is satisfied. 
\end{proposition}
\begin{proof}
	We prove by induction on $\deg(u)$, supposing it holds   $\forall u\in U_{\leq d}(\mathfrak{sl}(3))$,  with the previous Lemma showing it holds for $\deg(u)=d=1$. As we did before, if $v_j \in U_{\le d_j}(\mathfrak{sl}(3))^{\vb*a_j}$ for $j\in\{1,2\}$, with $1\le d_1,d_2\le d$, then $v_1v_2 \in U_{\le d_1+d_2}(\mathfrak{sl}(3))$. Let $(f_s)$ and $(\widetilde f_s)$ be given by $f_s = (sr(\xi))^{-d_1}w^s_\xi[v_1]$ and $\widetilde f_s = (sr(\xi))^{-d_2}w^s_\xi[v_2]$. By the hypothesis of induction, $f = (-i)^{d_1}\beta_{d_1}[v_1]$ and $\widetilde f = (-i)^{d_2}\beta_{d_2}[v_2]$ are the limits of $(f_s)$ and $(\widetilde f_s)$, respectively. Since
	\begin{equation}
		\norm{f_s\star^s_\xi \widetilde f_s-f\widetilde f}_\xi \le \norm{f_s\star^s_\xi \widetilde f_s - f\star^s_\xi\widetilde f}_\xi+\norm{f\star^s_\xi \widetilde f - f\widetilde f}_\xi\, ,
	\end{equation}
	we just need to verify that both summands in the r.h.s.~converge to $0$ as $s\to \infty$. The convergence of second summand follows straightforwardly from the Poisson condition. For the first summand, we use again that the sequence of operators \eqref{tp_oper} is bounded, so the convergences $f_s\to f$ and $\widetilde f_s\to \widetilde f$ imply that the limit of the first summand vanishes.
\end{proof}

Therefore, combining Propositions \ref{thmprop1} and \ref{thmprop2}, we have obtained: 

\begin{theorem}\label{cor:pois_crit}
	For $\xi \in\overline{\mathcal Q}$, a $\xi$-ray of universal correspondences $(w^s_\xi)_{s\in\mathbb N}$ is of Poisson type, so that \  $\mathfrak{W}(\mathcal O_\xi)\xrightarrow{ \ \sim \ } Poly(\mathcal O_\xi,\Omega_\xi)\, ,$  if and only if \eqref{condition1} is satisfied. 
\end{theorem}

\subsection{Second criterion for Poisson: characteristic matrices}

As presented in Paper I, every symbol correspondence for a quark system with dominant weight $p\,\varpi_1+q\,\varpi_2$ is uniquely determined by its characteristic matrices, cf. Theorems I.4.8 and I.5.9, and Remark I.5.13. Therefore, a natural question is how to write the Poisson condition for a $\xi$-ray of correspondences in terms of the sequence of their characteristic matrices, or characteristic numbers if $\xi=(1,0)$ or $(0,1)$. To answer this question more clearly, we  translate some notation used in this Paper II to the language of Paper I. 

For each $\xi \in \overline{\mathcal Q}$, let $\vb*p^1_\xi = (p_1,q_1)\in\mathbb N^2_0\setminus\{(0,0)\}$ be the first integral pair for $\xi$, 
\begin{equation}
    \omega_{\vb*p^1_\xi} = p_1\,\varpi_1+q_1\,\varpi_2 = \omega_\xi \ , 
\end{equation}
cf. Definition \ref{intrad-defn}. 
Then, fixed $\xi$,  for each $s \in \mathbb N$ we denote
\begin{equation}\label{norm_omega_ps}
\begin{aligned}
    \vb*p^s_\xi &= (sp_1,sq_1) =  s\vb*p^1_\xi \ , \\  \omega_{\vb*p^s_\xi}&\equiv s\,\omega_\xi \ , \ \ \norm{\omega_{\vb*p^s_\xi}}= sr(\xi)\, ,
\end{aligned} 
\end{equation}
so that  $\mathcal H_{\vb*p^s_\xi}\equiv \mathcal H_{s\omega_\xi}$ is the quantum quark system with irrep $\vb*p^s_\xi$, and $\rho_{\vb*p^s_\xi}\equiv\rho_{s\omega_\xi}$ is the (finite dimensional, cf.~(I.2.17)) representation of $U(\mathfrak{sl}(3))$ on $\mathcal H_{\vb*p^s_\xi}$. Also, consider the $\xi$-ray of symbol correspondences $(W^{\vb*p_\xi^s}\equiv W^{s\omega_\xi})_{s\in\mathbb N}$ generating the $\xi$-ray of universal correspondences  $(w^s_\xi)_{s\in\mathbb N}$ according to Definition \ref{def:univ_corresp}, that is, 
\begin{equation}
    w^s_\xi = W^{\vb*p^s_\xi}\circ \rho_{\vb*p^s_\xi} \ .
\end{equation}
Finally, denote by $\mathbf {C}^s_\xi(\vb*a)$ the characteristic matrices of $W^{\vb*p^s_\xi}\equiv W^{s\omega_\xi}$ (characteristic numbers as $1\times 1$ matrices if $p_1 = 0$ or $q_1 = 0$), cf.  Definition I.5.12. 

Recalling the normalized Hilbert-Schmidt inner product on each $\mathcal B(\mathcal H_{\vb*p^s_\xi})$, 
\begin{equation}\label{norm_hs_ip}
    \ip{A_1}{A_2}_{\vb*p^s_\xi} = \dfrac{1}{\dim(\vb*p^s_\xi)}\tr(A_1A_2) \ , 
\end{equation}
cf. (I.3.10), then based on what is known for spin systems, one should expect that $\xi$-rays of correspondences of Poisson type tend in some sense to an isometry with respect to the inner products $\ip{\cdot}_{\vb*p^s_\xi}$ and $\ip{\cdot}_\xi$ as $s\to\infty$, where the latter is the Haar inner product of functions on $\mathcal O_\xi$. We now show that this is indeed what happens, which will lead to an asymptotic condition for the characteristic matrices.

\begin{lemma}\label{ip_rho_beta}
    For any $u,v\in U(\mathfrak{sl}(3))$, we have
    \begin{equation}
		\begin{aligned}
		& \lim_{s\to \infty}\norm{\omega_{\vb*p_\xi^s}}^{-(\deg(u)+\deg(v))}\ip{\rho_{\vb*p_\xi^s}(u)}{\rho_{\vb*p_\xi^s}(v)}_{\vb*p_\xi^s}\\
		& \hspace{5 em} = (-i)^{\deg(v)-\deg(u)}\ip{\beta_{\deg(u)}[u]}{\beta_{\deg(v)}[v]}_\xi\, .
		\end{aligned}
	\end{equation}
\end{lemma}
\begin{proof}
For the $\xi$-ray $(b^s_\xi)_s$ of universal Berezin correspondences,
\begin{equation}\label{ip_rho_u_rho_v}
\begin{aligned}
    & \norm{\omega_{\vb*p_\xi^s}}^{-(\deg(u)+\deg(v))}\ip{\rho_{\vb*p_\xi^s}(u)}{\rho_{\vb*p_\xi^s}(v)}_{\vb*p_\xi^s} \\
    & \hspace{5 em} = \int_{\mathcal O_\xi}\norm{\omega_{\vb*p_\xi^s}}^{-(\deg(u)+\deg(v))}\overline{b^s_\xi[u]}\star b^s_\xi[v](\vb*\varsigma)d\vb*\varsigma\, .
\end{aligned}
\end{equation}
By decomposing $U_{\le \deg(u)}(\mathfrak{sl}(3))$ and $U_{\le \deg(v)}(\mathfrak{sl}(3))$ into irreps, it is possible to find some $d \in \mathbb N$ such that $Poly_{\le d}(\mathcal O_\xi)$ contains $\overline{b^s_\xi[u]}$ and $b^s_\xi[v]$ for every $s\in \mathbb N$. Since $Poly_{\le d}(\mathcal O_\xi)$ is finite dimensional,  Corollaries \ref{lemma:lim-t-1} and \ref{cor:ber_ray_pois} imply
\begin{equation}
\begin{aligned}
    & \lim_{s\to\infty} \norm{\omega_{\vb*p_\xi^s}}^{-(\deg(u)+\deg(v))}\overline{b^s_\xi[u]}\star b^s_\xi[v]\\
    & \hspace{5 em} = (-i)^{\deg(v)-\deg(u)}\overline{\beta_{\deg(u)}[u]}\beta_{\deg(v)}[v]\, ,
\end{aligned}
\end{equation}
cf.~\eqref{norm_omega_ps}. This convergence is uniform, so taking the integral we get the equation of the statement from \eqref{ip_rho_u_rho_v}.
\end{proof}

\begin{corollary}\label{prop:asymp_sw}
	If $(W^{\vb*p_\xi^s})$ is of Poisson type, then, for every $u,v\in U(\mathfrak{sl}(3))$, 
	\begin{equation}
		\begin{aligned}
		& \lim_{s\to \infty}\norm{\omega_{\vb*p_\xi^s}}^{-(\deg(u)+\deg(v))}\ip{W^{\vb*p_\xi^s}\circ \rho_{\vb*p_\xi^s}(u)}{W^{\vb*p_\xi^s}\circ \rho_{\vb*p_\xi^s}(v)}_\xi\\
        & \hspace{5 em} = (-i)^{\deg(v)-\deg(u)}\ip{\beta_{\deg(u)}[u]}{\beta_{\deg(v)}[v]}_\xi \\
		& \hspace{8 em} = \lim_{s\to \infty}\norm{\omega_{\vb*p_\xi^s}}^{-(\deg(u)+\deg(v))}\ip{\rho_{\vb*p_\xi^s}(u)}{\rho_{\vb*p_\xi^s}(v)}_{\vb*p_\xi^s} \ . 
		\end{aligned}
	\end{equation}
\end{corollary}
\begin{proof}
It is immediate from Theorem \ref{cor:pois_crit} and Lemma \ref{ip_rho_beta}.
\end{proof}

\begin{theorem}\label{theo:pois_char_m_iso}
    If $(W^{\vb*p_\xi^s})$ is of Poisson type, then the characteristic matrices satisfy
    \begin{equation}
        \lim_{s\to\infty}(\mathbf {C}^s_\xi(\vb*a))^\dagger\mathbf {C}^s_\xi(\vb*a)= \mathds 1\, .
    \end{equation}
\end{theorem}
\begin{proof}
    Let $s$ be large enough so that the dimension of the highest weight space of
    \begin{equation}
        \mathcal B(\vb*p_\xi^s;\vb*a) =\mathcal B(\mathcal H_{\vb*p_\xi^s})^{\vb*a}
    \end{equation}
    is constant $m\equiv m(\vb*a)$, cf. Notation I.5.8. Take $u_1,...,u_m \in U(\mathfrak{sl}(3))^{\vb*a}$ as highest weight vectors of degrees $\deg(u_\gamma) = d(\gamma)$ such that
    \begin{equation}
        (-i)^{d(\gamma_1)-d(\gamma_2)}\ip{\beta_{d(\gamma_1)}[u_{\gamma_1}]}{\beta_{d(\gamma_2)}[u_{\gamma_2}]} = \delta_{\gamma_1,\gamma_2}\, .
    \end{equation}
    By the previous corollary, for $s$ large enough, the set
    \begin{equation}
        \left\{\norm{\omega_{\vb*p_\xi^s}}^{-d(1)}\rho_{\vb*p_\xi^s}(u_1),...,\norm{\omega_{\vb*p_\xi^s}}^{-d(m)}\rho_{\vb*p_\xi^s}(u_m)\right\}
    \end{equation}
    is a basis of the highest weight space of $\mathcal B(\vb*p_\xi^s;\vb*a)$. 
    
    Now, for $\sigma \in \{1,...,m\}$, take
    \begin{equation}\label{A_sigma^s}
        A_\sigma^s = \sqrt{\dim(\vb*p_\xi^s)}\, \vb*e((\vb*a;\sigma);>_{\vb*a}) \, ,
    \end{equation}
    where $>_{\vb*a}$ stands for the highest weight. Denoting 
    \begin{equation}\label{fssigma}
        f_{\sigma}^s= W^{\vb*p_\xi^s}_{A_\sigma^s} \ , 
    \end{equation}
    the $j\times k$ entry of $(\mathbf {C}^s_\xi(\vb*a))^\dagger \mathbf {C}^s_\xi(\vb*a)$ is $\ip{f_j^s}{f_k^s}_\xi$, cf. Definition I.5.12 and Remark I.5.13, so we want to show
    \begin{equation}\label{lim_jxk_entry}
        \lim_{s\to\infty}\ip{f^s_j}{f^s_k}_\xi = \delta_{j,k}\, .
    \end{equation}
    
    Let $Z(s)$ be the complex square matrix with entries $(Z(s))_{\sigma,\gamma} = z^\sigma_\gamma(s)$ such that
    \begin{equation}\label{rho(u)=zA}
        \norm{\omega_{\vb*p_\xi^s}}^{-d(\gamma)}\rho_{\vb*p_\xi^s}(u_\gamma) = \sum_\sigma z^\sigma_\gamma(s) A_\sigma^s\, .
    \end{equation}
    From Corollary \ref{prop:asymp_sw}, we have that
    \begin{equation}
        \lim_{s\to\infty}Z(s)^\dagger Z(s)=\mathds 1\, .
    \end{equation}
    Therefore, for $C(s) = Z(s)^{-1}$, we have
    \begin{equation}
        \lim_{s\to\infty} C(s)^\dagger C(s) = \mathds 1\, .
    \end{equation}
    Also, the entries $(C(s))_{\gamma,\sigma} = c_\sigma^\gamma(s)$ are bounded on $s$ and satisfy
    \begin{equation}\label{A=crho(u)}
        A_\sigma^s = \sum_\gamma c^\gamma_\sigma(s)\norm{\omega_{\vb*p_\xi^s}}^{-d(\gamma)}\rho_{\vb*p_\xi^s}(u_\gamma)\, .
    \end{equation}
    Hence, $\ip{f_j^s}{f_k^s}_\xi$ is given by
    \begin{equation}\label{ip_simb_A_jA_k}
         \sum_{\gamma_1,\gamma_2}\overline{c^{\gamma_1}_j(s)}c^{\gamma_2}_k(s)\norm{\omega_{\vb*p_\xi^s}}^{-(d(\gamma_1)+d(\gamma_2))}\ip{W^{\vb*p_\xi^s}\circ \rho_{\vb*p_\xi^s}(u_{\gamma_1})}{W^{\vb*p_\xi^s}\circ \rho_{\vb*p_\xi^s}(u_{\gamma_2})}_\xi\, ,
    \end{equation}
    and $\ip{A_j^s}{A_k^s}_{\vb*p_\xi^s}$ is given by
    \begin{equation}\label{ip_A_jA_k}
        \begin{aligned}
            & \sum_{\gamma_1,\gamma_2}\overline{c^{\gamma_1}_j(s)}c^{\gamma_2}_k(s)\norm{\omega_{\vb*p_\xi^s}}^{-(d(\gamma_1)+d(\gamma_2))}\ip{\rho_{\vb*p_\xi^s}(u_{\gamma_1})}{\rho_{\vb*p_\xi^s}(u_{\gamma_2})}_{\vb*p_\xi^s}\, ,
            \end{aligned}
    \end{equation}
    Since $\ip{A_j^s}{A_k^s}_{\vb*p_\xi^s} = \delta_{j,k}$, applying Corollary \ref{prop:asymp_sw} on \eqref{ip_simb_A_jA_k}, we get  \eqref{lim_jxk_entry}.
\end{proof}

The last theorem states that the characteristic matrices of a $\xi$-ray of symbol correspondences of Poisson type are asymptotically unitary, that is to say, the $\xi$-ray of symbol correspondences needs to satisfy a weak \emph{asymptotic Stratonovich-Weyl condition} to be of Poisson type. Nonetheless, in order to get a statement of equivalence between Poisson property and the convergence of characteristic matrices to specific unitary matrices, in the spirit of what happens for spin systems and their characteristic numbers \cite{RS}, we need to fix a method for Clebsch-Gordan decompositions of spaces of operators. We can avoid such choice-dependent classification by comparing with Berezin correspondences instead. 

\begin{theorem}\label{theo:pois_char_m_ber}
	A $\xi$-ray of symbol correspondences $(W^{\vb*p_\xi^s})_{s\in\mathbb N}$ is of Poisson type, so that \eqref{WtoP} is satisfied,  if and only if its sequence of characteristic matrices satisfy
    \begin{equation}\label{lim(C(a)-B(s))}
        \lim_{s\to\infty}\left(\mathbf {C}^s_\xi(\vb*a)-\mathbf{B}^s_\xi(\vb*a)\right)= 0 \ ,
    \end{equation}
    where $\big(\mathbf{B}^s_\xi(\vb*a)\big)_{\!s\in\mathbb N}$ is the sequence of characteristic matrices for the $\xi$-ray $(B^{\vb*p_\xi^s})_{s\in\mathbb N}$ of Berezin symbol correspondences, cf. \eqref{ber_corresp} and Definition I.5.12. 
\end{theorem}
\begin{proof}
    Recalling \eqref{A_sigma^s}, again let $f_\sigma^s$ be the symbol of $A_\sigma^s$ via $W^{\vb*p_\xi^s}$,
    cf. \eqref{fssigma}, and now let $\widetilde f_\sigma^s$ be the symbol of $A_\sigma^s$ via $B^{\vb*p_\xi^s}$. Since $Poly(\mathcal O_\xi)^{\vb*a}$ is finite dimensional, the limit \eqref{lim(C(a)-B(s))} holds if and only if
    \begin{equation}
        \lim_{s\to\infty}\norm{f_\sigma^s-\widetilde f_\sigma^s}_\xi = 0 \ , 
    \end{equation}
    for every $\sigma \in\{1,...,m\}$, which in turn is equivalent to
    \begin{equation}
        \lim_{s\to\infty}\norm{\omega_{\vb*p_\xi^s}}^{-\deg(u)}\norm{W^{\vb*p_\xi^s}\circ \rho_{\vb*p_\xi^s}(u)-B^{\vb*p_\xi^s}\circ \rho_{\vb*p_\xi^s}(u)}_\xi = 0
    \end{equation}
    for every $u\in U(\mathfrak{sl}(3))^{\vb*a}$, cf.~Corollary \ref{prop:asymp_sw} and \eqref{rho(u)=zA}-\eqref{A=crho(u)}. Then the statement is a consequence of Corollary \ref{cor:ber_ray_pois} and Theorem \ref{cor:pois_crit}. 
\end{proof}

For pure-quark systems, with $\xi=(1,0)$ or $\xi=(0,1)$ and $\xi$-rays of representations  $((p,0))_{p\in\mathbb N}$ or $((0,p))_{p\in\mathbb N}$, respectively, the characteristic matrices are $1\times 1$ matrices and are just called characteristic numbers, $c_n^p\in\mathbb R^\times$, cf. Definition I.5.12 and Remark I.5.13. From Theorem \ref{theo:pois_char_m_iso}, a ray of pure-quark correspondences is of Poisson type only if 
\begin{equation}
    \lim_{p\to\infty}|c_n^p|=1 \ , \ \ \forall n\in\mathbb N \ . 
\end{equation}

However, by choosing a decomposition of the space of operators in such a way that the symmetric Berezin correspondences have only positive characteristic numbers, $b_n^p\in\mathbb R^+$, Theorem \ref{theo:pois_char_m_ber} provides a finer criterion by means of the sequence of its characteristic numbers $(c_n^p)_{p\in\mathbb N}$ for a ray of pure-quark correspondences to be of Poisson type, which is analogous to the criterion for spin systems \cite{RS}. 

We illustrate this for $\xi=(1,0)$ with sequence of representations $(\vb*p = (p,0))_{p\in\mathbb N}$, for which we have the following: 

\begin{proposition}\label{ThmBCN}
For every $\vb*p = (p,0)$, $p\in\mathbb N$, taking
\begin{equation}\label{pq_CG_choice}
\begin{aligned}
    & \vb*e(n;(2n,n,0);n/2) = \dfrac{1}{\mu_n(p)}V_+^n \in \mathcal B(\mathcal H_{\vb*p})\, , \\
    & \hspace{2 em} \mu_n(p) = \dfrac{n!}{\sqrt{(2n+2)!}}\sqrt{\dfrac{(p+n+2)!}{(p-n)!}}\, ,
\end{aligned}
\end{equation}
the characteristic numbers $(b_n^p)_{n\leq p}$ of the symmetric Berezin correspondence are all positive and satisfy \ $|b_n^p-1| \in O(1/p)$ as $p\to\infty$, for every fixed $ n\in\mathbb N$.
\end{proposition}

The proof of this Proposition is a straightforward computation using various results and notations from Paper I, so it is placed in Appendix \ref{proofThmBerCharNumbers}. Combining Proposition \ref{ThmBCN} with Theorem \ref{theo:pois_char_m_ber}, we obtain immediately:

\begin{corollary}\label{cor:pois_cp2}
    Let $(W^p)_{p\in \mathbb N}$ be a ray of pure-quark symbol correspondences 
    \begin{equation}\label{seq_symb_corresp_pq}
        \Big(W^p:\mathcal B(\mathcal H_{(p,0)})\to Poly(\mathcal O_{(1,0)})\Big)_{p\in\mathbb N} \ ,
    \end{equation}
    with characteristic numbers $(c_n^p)_{n\leq p}$, for each $p\in\mathbb N$. Assuming \eqref{pq_CG_choice}, $(W^p)_{p\in \mathbb N}$ is of Poisson type if  and only if the characteristic numbers satisfy 
    \begin{equation}
        \lim_{p\to\infty} c_n^p=1 \ , \ \ \forall n\in \mathbb N \ . 
    \end{equation}
\end{corollary}

As mentioned before, in Appendix \ref{app:pq-asymp} we provide an alternative (summary of the) proof of this corollary using a method which, although quite more cumbersome, is analogous to the method used in the case of spin systems, cf. \cite{RS}. 

\section{Universal correspondences on the coarse Poisson sphere}\label{sec:corresp_sphere}

In this section, we develop a method for extending the rays of correspondences defined for each rational orbit $\mathcal O_\xi\subset\mathcal S^7$ to a pencil of correspondence rays defined on the full coarse Poisson sphere $\{\mathcal S^7,\widehat{\Pi}_{\mathfrak g}\}$, such that an extended $SU(3)$-invariant noncommutative algebra constructed by this method, with product induced from the universal enveloping algebra $U(\mathfrak{sl}(3))$, restricts to a fuzzy $\xi$-orbit, for each $\xi\in\overline{\mathcal Q}$, as in Definition \ref{def:pois_ray}. Then we investigate if/how such extended algebras can recover the Poisson algebra of polynomials on $(\mathcal S^7,\widehat \Pi_{\mathfrak g})$ in some asymptotic limit.

\subsection{Pencils of correspondence rays: Magoo spheres}

Before starting the construction of pencils of correspondence rays on the coarse Poisson sphere $\{\mathcal S^7,\widehat\Pi_{\mathfrak g}\}$, we recall that each rational orbit $\mathcal O_\xi$ is a symplectic leaf in a singular foliation of the smooth Poisson sphere $(\mathcal S^7,\widehat \Pi_{\mathfrak g})$ and that this orbit is the preimage  of a fixed number $\chi_\xi\in\mathbb R$ by the restriction of the cubic polynomial $\tau$ to $\mathcal S^7$, cf. \eqref{taupol} in Proposition \ref{prop:c3_pol}, so that, for each $\xi\in\overline{\mathcal Q}$,  we can define a polynomial function $\check\tau_\xi$ vanishing on $\mathcal O_\xi$ and only on this orbit in $\mathcal S^7$, cf. \eqref{delta_func_xy} in Remark \ref{checktau}.

Now, consider the ideal of polynomials vanishing on $\mathcal O_\xi$, 
\begin{equation}
	I_\xi := \{f\in Poly(\mathcal S^7):f|_{\mathcal O_\xi}\equiv 0\}\trianglelefteq Poly(\mathcal S^7)\, ,
\end{equation}
in terms of which we can set the natural isomorphism
\begin{equation}\label{poly_orb_iso}
    Poly(\mathcal O_\xi) \simeq Poly(\mathcal S^7)/I_\xi \ ,
\end{equation}
so that we can write universal correspondences as maps
\begin{equation}\label{wxiIxi}
    w_\xi:
    U(\mathfrak{sl}(3))\to Poly(\mathcal O_\xi) \simeq Poly(\mathcal S^7)/I_\xi\, .
\end{equation}

This can be extended simultaneously for any finite set of rational orbits. Let \begin{equation}\label{mathcalP} 
\mathcal P\subset \overline{\mathcal Q} \ , \ \ \norm{\mathcal P}\in\mathbb N \ ,
\end{equation} 
be a finite subset of rational orbits and, for each $\xi \in \mathcal P$, take a universal correspondence \eqref{wxiIxi}. In order to ``glue together'' such correspondences, we invoke the invariant polynomial $\tau|_{\mathcal S^7}\in Poly(\mathcal S^7)$, given by \eqref{taupol}, and its restriction complements $\check\tau_{\xi}\in I_\xi$, given by  \eqref{delta_func_xy}. Thus, for each $\xi \in \mathcal P$, let
\begin{equation}\label{deltaP}
	\mathfrak{d}^\xi_{\mathcal P} := \dfrac{1}{M_{\mathcal P}^\xi}\prod_{\xi' \in \mathcal P\setminus\{\xi\}}\!\check\tau_{\xi'} \ \ \in Poly(\mathcal S^7)\ \ \ , \ \ \ \ 
	M_{\mathcal P}^\xi = \prod\limits_{\xi'\in \mathcal P\setminus\{\xi\}}\!\check\tau_{\xi'}(\xi) \ ,
\end{equation}
so that $\mathfrak{d}^\xi_{\mathcal P}$ works like a delta-$\xi$ function on $\mathcal P$, that is, for $\xi, \xi'\in \mathcal P$, we have
\begin{equation}\label{deltaP2}
	\mathfrak{d}^\xi_{\mathcal P}\big|_{\mathcal O_{\xi'}} \equiv \begin{cases}
		1 \ , \ \ \mbox{if} \ \ \xi = \xi'\\
		0 \ , \ \ \mbox{if} \ \ \xi\ne \xi'
	\end{cases}\, .
\end{equation}
Then, taking
\begin{equation}\label{ideal_P}
	I_{\mathcal P} := \bigcap_{\xi \in \mathcal P} I_\xi \ ,
\end{equation}
we have\footnote{Since $\mathcal P$ is finite, one could use the more usual notion of direct sum instead of product. However, we chose the product in anticipation of an infinite product that will take place ahead.}
\begin{equation}\label{Poly(P_n)}
    Poly\Big(\bigcup_{\xi\in\mathcal P} \mathcal O_\xi\,\Big) \equiv \prod_{\xi\in\mathcal P}Poly(\mathcal O_\xi)\simeq  Poly(\mathcal S^7)/I_{\mathcal P} \ , 
\end{equation}
and we use $\mathfrak{d}^\xi_{\mathcal P}$ given by \eqref{deltaP} to ``glue'' a finite family of correspondences and also a finite family of correspondence rays, as follows. 

\begin{definition}\label{def:magoo_corresp-rays}
Given a finite subset $\mathcal P\subset \overline{\mathcal Q}$, a \emph{finite pencil of universal correspondence rays for $\mathcal P$}, or simply a \emph{pencil of correspondence rays for $\mathcal P$}, is a sequence of maps $(w_{\mathcal P}^s)_{s\in\mathbb N}$, where each 
\begin{equation}\label{Magoo-eq-s}
	w_{\mathcal P}^s:U(\mathfrak{sl}(3))\to Poly(\mathcal S^7)/I_{\mathcal P}:u\mapsto w_{\mathcal P}^s[u] = \sum_{\xi \in \mathcal P} \mathfrak{d}^\xi_{\mathcal P}w_\xi^s[u] \ ,
\end{equation}
and where, for each $\xi \in \mathcal P$,  $(w_\xi^s)_{s\in\mathbb N}$ is a $\xi$-ray of universal correspondences, such that $w_\xi^1:U(\mathfrak{sl}(3))\to Poly(\mathcal O_\xi)$ is a universal correspondence w.r.t.~its first  dominant weight $\omega_\xi=r(\xi)\xi$, cf. \eqref{intrad}-\eqref{defr}.
\end{definition}

With the above definition, we naturally obtain a sequence of twisted products on the sequence of images of $(w^s_{\mathcal P})_{s\in\mathbb N}$. Since these products are not commutative, their commutators act as derivations on their algebras. But as noted in Section 3, cf. Remark \ref{rmk:sp_resc} w.r.t.~\eqref{Poisson-eq-2} in Definition \ref{def:pois_ray}, a \emph{weighted} derivation is more suitable for the semiclassical limit.

\begin{definition}\label{def:penc_magoo-p}
Given a pencil of correspondence rays $(w_{\mathcal P}^s)_{s\in\mathbb N}$, denote by $(\mathcal W^s_{\mathcal P})_{s\in\mathbb N}$ its sequence of images $\mathcal W_{\mathcal P}^s\subset Poly(\mathcal S^7)/I_{\mathcal P}$. Then, its induced \emph{twisted product sequence} $(\star_{\mathcal P}^s)_{s\in\mathbb N}$ on  $(\mathcal W_{\mathcal P}^s)_{s\in\mathbb N}$ is given by 
    \begin{equation}\label{starPs}
    \begin{aligned}
        w_{\mathcal P}^s[u]\star_{\mathcal P}^sw_{\mathcal P}^s[v]&=w_{\mathcal P}^s[uv]  \\ 
        &= \sum_{\xi \in \mathcal P}\mathfrak{d}^\xi_{\mathcal P} w_\xi^s[u]\star_\xi^s w_\xi^s[v] \ , \ \ \forall u,v\in U(\mathfrak{sl}(3)) \ , 
    \end{aligned}     
    \end{equation}
    and its \emph{$r$-weighted bracket} sequence $\big([\cdot,\cdot]_{\star^s_{\mathcal P}}^r\big)_{s\in\mathbb N}$ on $(\mathcal W_{\mathcal P}^s,\star_{\mathcal P}^s)_{s\in\mathbb N}$ is given by 
  \begin{equation}\label{bracket-s}
        \big[ w_{\mathcal P}^s[u],w_{\mathcal P}^s[v]\big]^r_{\star^s_{\mathcal P}}=\sum_{\xi \in \mathcal P}\mathfrak{d}^\xi_{\mathcal P} r(\xi)\big[w_\xi^s[u], w_\xi^s[v]\big]_{\star_\xi^s} \ ,  \ \ \forall u,v\in U(\mathfrak{sl}(3)) \ . 
    \end{equation}  
\end{definition}

However, any finite family of leaves is far from sufficient to determine the Poisson algebra of polynomials on the sphere, as $s\to\infty$. Thus, we now consider an increasing family of nested finite subsets $\mathcal P\subset \overline{\mathcal Q}$ whose limit is the entire set of rational orbits $\overline{\mathcal Q}$ (so that, in particular, $I_{\mathcal P}\to 0$). Any chain of finite subsets of $\overline{\mathcal Q}$ is countable because $\overline{\mathcal Q}$ is countable, so we can define chain  sequences of the form
\begin{equation}\label{chainCn}
    \mathcal C = (\mathcal P_n)_{n \in \mathbb N} \ , \ \  \mathcal P_n\subset \overline{\mathcal Q} \ \ \mbox{s.t.} \ \ |\mathcal P_n|<\infty \ , \ \  \mathcal P_n\subsetneq \mathcal P_{n+1}\, , \ \lim_{n\to\infty}\mathcal P_n=\overline{\mathcal Q} \, .
\end{equation}
Furthermore, on each $\mathcal P_n$ as above, let's denote, for convenience, 
\begin{equation}
    \mathfrak{d}_n^\xi\equiv\mathfrak{d}_{\mathcal P_n}^\xi \ , \ \ \ \mbox{cf. \eqref{deltaP}}\, ,   
\end{equation}
and also, cf. \eqref{Poly(P_n)},  
\begin{equation}\label{PPn}
   \mathfrak{P}_n\equiv Poly(\mathcal S^7)/I_{\mathcal P_n} \simeq \prod_{\xi\in\mathcal P_n}Poly(\mathcal O_\xi) \  . 
\end{equation}

\begin{definition}\label{def:penc_magoo}
Let $\mathcal C$ be a chain as in \eqref{chainCn} and, for each $\xi \in \overline{\mathcal Q}$, let
\begin{equation}\label{Pxi}
    n_\xi := \min\{n\in\mathbb N:\xi \in \mathcal P_n\}\, .
\end{equation}
A \emph{Magoo pencil of correspondence rays for $\mathcal C$} is a bi-sequence 
\begin{equation}\label{w_C}
    \begin{aligned} 
    w_{\mathcal C}=(w_n^s)_{n,s\in \mathbb N} \ &, \ \ \mbox{with} \\
    w^s_n:U(\mathfrak{sl}(3))\to \mathfrak{P}_n   \, &, \ u\mapsto w^s_n[u] = \sum_{\xi \in \mathcal P_n} \mathfrak{d}_{n}^\xi w^s_\xi[u]\, ,
    \end{aligned}
\end{equation}
where $(w^s_\xi)_{s\in\mathbb N}$ is a $\xi$-ray of universal correspondences, cf. Definition \ref{def:magoo_corresp-rays}, such that, for every $\xi\in \overline{\mathcal Q}$ and every $u\in U(\mathfrak{sl}(3))$,
\begin{equation}\label{notseq-n}
 w^s_n[u]|_{\mathcal O_\xi}=w^s_{n_\xi}[u]|_{\mathcal O_\xi} \ , \ \ \forall    n\geq n_\xi\, , \ \forall s\in\mathbb N\, .
\end{equation}

Then, denoting by $\mathcal W_{\mathcal C}=(\mathcal W_n^s)_{n,s\in \mathbb N}$ the bi-sequence of images of $U(\mathfrak{sl}(3))$ by $w_{\mathcal C}$, its induced \emph{Magoo product} on $\mathcal W_{\mathcal C}$ is the bi-sequence of products 
    \begin{equation}
        \star_{\mathcal C}=(\star^s_n)_{n,s\in\mathbb N} \ , \ \ \star^s_n\equiv \star^s_{\mathcal P_n}:\mathcal W_n^s\times\mathcal W_n^s\to\mathcal W_n^s \ , 
    \end{equation}
    cf. \eqref{starPs}, so that  \ $\forall u,v\in U(\mathfrak{sl}(3))$, 
    \begin{equation}\label{prodC}
      w_\mathcal C[u]\star_{\mathcal C}w_\mathcal C[v]  :=(w^s_n[u]\star^s_n w^s_n[v])_{n,s\in\mathbb N} \ , 
    \end{equation}
    and its \emph{Magoo bracket} on $(\mathcal W_{\mathcal C},\star_{\mathcal C})$ is the bi-sequence of $r$-weighted brackets 
    \begin{equation}
        [\cdot ,\cdot]_{\star_\mathcal C}^r=\big([\cdot,\cdot]_{\star_n^s}^r\big)_{n,s\in\mathbb N} \ , \ \ [\cdot,\cdot]_{\star_n^s}^r\equiv[\cdot,\cdot]_{\star_{\mathcal P_n}^s}^r:\mathcal W_n^s\times\mathcal W_n^s\to\mathcal W_n^s \ , 
    \end{equation}
    cf. \eqref{bracket-s}, so that \ $\forall u,v\in U(\mathfrak{sl}(3))$, 
    \begin{equation}
        \big[w_{\mathcal C}[u],w_{\mathcal C}[v]\big]_{\star_\mathcal C}^r = \Big(\sum_{\xi\in\mathcal P_n} \mathfrak d_n^\xi r(\xi)\big[w^s_n[u],w^s_n[v]\big]_{\star^s_n}\Big)_{\!n,s\in\mathbb N} \ .
    \end{equation}

    In this way, $\big(\mathcal W_{\mathcal C},\star_{\mathcal C},[\cdot ,\cdot]_{\star_\mathcal C}^r\big)$ as above shall be called a \emph{Magoo sphere}, denoted
    \begin{equation}\label{notation-Ms}
        \mathfrak{W}\{\mathcal S^7,\widehat\Pi_{\mathfrak g}\}=\big(\mathcal W_{\mathcal C},\star_{\mathcal C},[\cdot ,\cdot]_{\star_\mathcal C}^r\big) \ . 
    \end{equation}
\end{definition}

Henceforth, let $\mathfrak W\{\mathcal S^7,\widehat\Pi_{\mathfrak g}\}$ be a Magoo sphere as just defined above for the coarse Poisson sphere $\{\mathcal S^7,\widehat\Pi_{\mathfrak g}\}$. If we denote the Poisson algebra of complex polynomials on the smooth Poisson sphere $(\mathcal S^7,\widehat\Pi_{\mathfrak g})$ by  
\begin{equation}\label{smPalg}
Poly(\mathcal S^7,\widehat\Pi_{\mathfrak g}) \ , 
\end{equation}
we want to study if/when/how $\mathfrak{W}\{\mathcal S^7,\widehat\Pi_{\mathfrak g}\}$ converges to $Poly(\mathcal S^7,\widehat\Pi_{\mathfrak g})$ in some asymptotic limit. In particular, we are concerned with asymptotics of Magoo product and Magoo bracket of polynomials, so the first thing to study is if/when the product and the bracket are well defined for general polynomials on $\mathcal S^7$. 

\begin{lemma}\label{lem:magoo_prod_poly}
    Given $f \in Poly(\mathcal S^7)$, there exists $s_0\in\mathbb N$ such that, for every $s\ge s_0$, $f|_{\mathcal O_\xi}$ is in the image of $w^s_\xi$ for every $\xi \in \overline{\mathcal Q}$.
\end{lemma}
\begin{proof}
Without loss of generality, we can assume that $f \in Poly(\mathcal S^7)^{(a,b)}$, then $f|_{\mathcal O_\xi}\in Poly(\mathcal O_\xi)^{(a,b)}$, $\forall \xi\in\mathcal Q$ (and $\forall\xi\in\overline{\mathcal Q}$ when $a=b$). Then, we proceed as in the proof of Lemma \ref{lemma:f_image_orbit}, obtaining $s_0$ given by \eqref{suff-s0}.
\end{proof}

\begin{lemma}\label{prop:magoo_prod_poly}
    Given $f \in Poly(\mathcal S^7)$, let $f_{|n}$ denote its quotient in $\mathfrak{P}_n$, cf.  \eqref{PPn}, and let $s_0\in\mathbb N$ be as in Lemma \ref{lem:magoo_prod_poly}. Then, 
    \begin{equation}
        s\geq s_0\implies f_{|n} \in \mathcal W^s_n \ , \ \ \forall n \in \mathbb N \ . 
    \end{equation}
\end{lemma}
\begin{proof}
By hypothesis, $f|_{\mathcal O_\xi}$ lies in the image of $w^s_\xi$ for every $\xi \in \overline{\mathcal Q}$ and every $s\ge s_0$. Thus, for any fixed $n\in\mathbb N$ and $s\ge s_0$, we need to exhibit $u_n^s \in U(\mathfrak{sl}(3))$ such that $f_{|n} = w^s_n[u_n^s]$. For each $\xi\in\overline{\mathcal Q}$, let $u_\xi^s\in U(\mathfrak{sl}(3))$ be such that $f|_{\mathcal O_\xi} = w^s_\xi[u_\xi^s]$ for $s\ge s_0$. Since the eigenvalues of the Casimir operators separate the representations, cf. (I.B.3), there are $c_\xi^s \in Z(U(\mathfrak{sl}(3))$, $\forall \xi \in \mathcal P_n$, s.t. 
    \begin{equation}\label{newdelta}
        \forall \xi,\xi' \in \mathcal P_n\, , \ \ w^s_{\xi}[c^s_{\xi'}] = \delta_{\xi,\xi'}\implies w^s_{\xi}[c^s_{\xi'}u^s_{\xi'}] = \delta_{\xi,\xi'}w^s_{\xi}[u^s_\xi] \, .
    \end{equation}
    Therefore, from \eqref{w_C} and \eqref{newdelta}, 
    \begin{equation}
        u_n^s = \sum_{\xi\in\mathcal P_n}c_\xi^s u^s_\xi
    \implies 
        w^s_n[u^s_n] = \sum_{\xi\in \mathcal P_n} \mathfrak{d}_n^\xi w^s_\xi[u^s_n] = \sum_{\xi\in \mathcal P_n} \mathfrak{d}_n^\xi w^s_\xi[u^s_{\xi}] = f_{|n}\, ,
    \end{equation}
    and this $u^s_n\in U(\mathfrak{sl}(3))$ is as claimed.
\end{proof}

\begin{remark}\label{rem-s0}
We highlight that $s_0$ as above depends only on $f\in Poly(\mathcal S^7)$, and it works for any Magoo sphere, obtained from any pencil of correspondence rays $(w^s_n)_{n,s\in \mathbb N}$ for any chain $\mathcal C=(\mathcal P_n)_{n\in\mathbb N}$ as in \eqref{chainCn}, once $f$ is fixed. Thus, in light of Lemma \ref{prop:magoo_prod_poly}, given $f_1,f_2\in  Poly(\mathcal S^7)$, for $s_0=\max\{s_0^1,s_0^2\}$, with $s_0^j$ as in Lemma \ref{lem:magoo_prod_poly} w.r.t.~$f_j$, we can make sense of $f_1\star_{\mathcal C}f_2$ as a bi-sequence $({f_1}_{|n}\star^s_n{f_2}_{|n})_{n\in\mathbb N,s\ge s_0}\, ,$ and similarly for the Magoo bracket $[f_1,f_2]^r_{\star_{\mathcal C}}\, .$ 
\end{remark}

Thus, to explore the asymptotics of $\mathfrak W\{\mathcal S^7, \widehat\Pi_{\mathfrak g}\}$, we first need to establish the meaning of limits of sequences $(f^s)_{s\geq s_0}$ and $(h_n)_{n\in\mathbb N}\, ,$ where $s$ is the semiclassical asymptotic parameter and $n$ indexes the nested finite subsets $\mathcal P_n\subset\overline{\mathcal Q}$ in $\mathcal C$, cf. \eqref{chainCn}. That these two limits are of different nature can be inferred by construction and is implied by \eqref{notseq-n}. Hence, to establish these limits, we first invoke
\begin{equation}
    \mathfrak P = \prod_{\xi\in\overline{\mathcal Q}}Poly(\mathcal O_\xi)
\end{equation}
as an ambient space, for which we have an inclusion
\begin{equation}
    Poly(\mathcal S^7)\hookrightarrow \mathfrak P:f\mapsto (f|_{\mathcal O_\xi})_{\xi\in\overline{\mathcal Q}}\equiv (f_\xi)_{\xi\in\overline{\mathcal Q}} \, ,
\end{equation}
and projections
\begin{equation}\label{proj_P_Pn}
    \mathfrak P\to \mathfrak P_n:f = (f_\xi)_{\xi\in\overline{\mathcal Q}}\mapsto f_{|n} = (f_\xi)_{\xi\in\mathcal P_n}
\end{equation}
for every $n\in \mathbb N$, cf. \eqref{PPn}.\footnote{In accordance with notation of Lemma \ref{prop:magoo_prod_poly}.} Then, we consider two distinct types of convergence:

\ 

\noindent {\bf type (i)}: \  for $(f^s)_{s\geq s_0}$, with each $f^s = (f^s_\xi)_{\xi\in\mathcal P_n}\in \mathfrak P_n$ and $f = (f_\xi)_{\xi\in\mathcal P_n} \in \mathfrak P_n$,
\begin{equation}\label{lim_s}
    \lim_{s\to\infty}f^s = f \iff \forall \epsilon>0 , \exists s_\epsilon\in \mathbb N:s\ge s_\epsilon \implies \norm{f^s_\xi-f_\xi}_\xi <\epsilon \ \ \forall \xi \in \mathcal P_n\, ,
\end{equation}
that is, $(f^s)_{s\geq s_0}$ converges to $f$ in $\mathfrak P_n$ iff  $f^s_\xi \xrightarrow{s\to\infty}  f_\xi$ uniformly over $\mathcal P_n$.

\

\noindent {\bf type (ii)}: \  for $(h_n \in \mathfrak P_n)_{n\in\mathbb N}$ and $h \in \mathfrak P$,
\begin{equation}\label{lim_n}
    \lim_{n\to\infty}h_n = h \iff h_n = h_{|n} \ \ \forall n \in \mathbb N\, ,
\end{equation}
cf. \eqref{proj_P_Pn}, that is, $(h_n\in \mathfrak P_n)_{n\in\mathbb N}$ converges to $h\in\mathfrak P$ iff $h$ is a common extension to every $h_n\in \mathfrak P_n$.

\  

Convergence types (i) and (ii) above induce two different kinds of asymptotics for a Magoo sphere, depending on which order of iterated limits we take for  $f_1\star_{\mathcal C}f_2$ and $[f_1,f_2]^r_{\star_\mathcal C}$. We begin by exploring the ordering given by (i) first, then (ii).

\begin{definition}\label{MagooAsymp-defn}
We say that $\mathfrak W\{\mathcal S^7,\widehat\Pi_{\mathfrak g}\}$ is of \emph{Poisson type} if its Magoo product and Magoo bracket satisfy, for any $f_1,f_2 \in Poly(\mathcal S^7)$, 
\begin{equation}\label{magoo_prod_lim}
\begin{aligned} 
    &\lim_{n\to\infty}\lim_{s\to\infty} f_1\star_{\mathcal C}f_2 = f_1f_2\in \mathfrak P \ \iff \\ &({f_1}_{|n}\star^s_n {f_2}_{|n})_{s\ge s_0}\xrightarrow{s\to\infty} (f_1f_2)_{|n}\, , \ \forall n \in \mathbb N\, , 
\end{aligned}     
\end{equation}
\begin{equation}\label{magoo_bracket_lim}
\begin{aligned}
    &\lim_{n\to\infty}\lim_{s\to\infty} s[f_1,f_2]^r_{\star_{\mathcal C}} = i\{f_1,f_2\}\in \mathfrak P \ \iff \\ 
    &(s[{f_1}_{|n},{f_2}_{|n}]^r_{\star^s_n})_{s\ge s_0}\xrightarrow{s\to\infty} i{\{f_1,f_2\}}_{|n} \, , \ \forall n \in \mathbb N\, , 
\end{aligned}    
\end{equation}
 where $n$ indexes $\mathcal P_n\in \mathcal C$, cf.~Remark \ref{rem-s0} and \eqref{lim_s}-\eqref{lim_n}. In this case, we write
\begin{equation}\label{nprecs}
    \mathfrak {W}\{\mathcal S^7,\widehat\Pi_{\mathfrak g}\} \xrightarrow{n\prec s} Poly(\mathcal S^7,\widehat\Pi_{\mathfrak g}) \, , 
\end{equation}
where the superscript $n\prec s$  refers to the order of the limits in \eqref{magoo_prod_lim}-\eqref{magoo_bracket_lim}.
\end{definition}

\begin{theorem}\label{theo:Magoo_pois}
    A Magoo sphere is of Poisson type if and only if all of its $\xi$-rays of universal correspondences are of Poisson type.
\end{theorem}
\begin{proof}
Since the Poisson-type uniform convergences for each $\xi\in\overline{\mathcal Q}$ (cf.~\eqref{Poisson-eq}-\eqref{Poisson-eq-2}) trivially  extend uniformly over finite sets $\mathcal P_n\subset\overline{\mathcal Q}$, the statement follows immediately from the definitions, cf.~\eqref{lim_s}-\eqref{lim_n} and \eqref{magoo_prod_lim}-\eqref{magoo_bracket_lim}. 
\end{proof}

\begin{corollary}\label{cor:magoo_asymp}
    For $j=1,2$, let $\mathfrak{W}_j\{\mathcal S^7,\widehat\Pi_{\mathfrak g}\}$ be Magoo spheres constructed from the same $\xi$-rays of universal correspondences $(w_\xi^s)_{s\in\mathbb N,\xi\in\overline{\mathcal Q}}\, ,$  but from two distinct chains $\mathcal C_1$ and $\mathcal C_2$ of finite subsets of $\overline{\mathcal Q}$ satisfying \eqref{chainCn}. Then, 
    \begin{equation}
    \mathfrak{W}_1\{\mathcal S^7,\widehat\Pi_{\mathfrak g}\} \xrightarrow{n\prec s} Poly(\mathcal S^7,\widehat\Pi_{\mathfrak g})\iff \mathfrak{W}_2\{\mathcal S^7,\widehat\Pi_{\mathfrak g}\} \xrightarrow{n\prec s} Poly(\mathcal S^7,\widehat\Pi_{\mathfrak g})\,.
\end{equation}
\end{corollary}

Hence, although there are infinitely many different  chains $\mathcal C$ of $\overline{\mathcal Q}$ satisfying \eqref{chainCn}, the Poisson condition for a Magoo sphere is independent of their choice, thus we can restrict ourselves to a canonical choice, as follows. Recall there exits a well defined function $r:\overline{\mathcal Q}\to\mathbb R^+$, $\xi\mapsto r(\xi)$, the integral radius of $\xi$, cf. \eqref{intrad} in Definition \ref{intrad-defn}, so we take the index $n$ in \eqref{chainCn} to be an increasing function of $r$,
\begin{equation}\label{nxi}
    n_\xi=n(r(\xi)) \ , \ \ r(\xi)< r(\xi')\iff n_\xi< n_{\xi'} \ .  
\end{equation}

\begin{definition}\label{Prn-defn}
The \emph{radial chain}  $\mathcal C^r=(\mathcal R_n)_{n\in\mathbb N}$ is the chain as in \eqref{chainCn} such that, $\forall \xi,\xi'\in \overline{\mathcal Q}$, \eqref{nxi} holds.
\end{definition}

In other words, for any given $n$, $\mathcal R_n$  is  the union of all rational orbits whose integral radius $r(\xi)$ is such that $n_\xi=n(r(\xi))\leq n$. Thus, $\mathcal R_1=\{\xi_{(1,0)},\xi_{(0,1)}\}$, $\mathcal R_2=\mathcal R_1\cup\{\xi_{(1,1)}\}$, $\mathcal R_3=\mathcal R_2\cup\{\xi_{(2,1)},\xi_{(1,2)}\}$ and so on, so that as $n$ increases we add up  orbits  $\xi\in\mathcal Q$ to $\mathcal R_n$ in increasing order of integral radius.

\begin{remark}\label{eisenstein}
A systematic way for determining all integral orbits of a given radius is as follows.\footnote{Note that the radius of an integral orbit $\mathcal O$ is always a natural multiple of the integral radius of $\mathcal O_\xi \subset\mathcal S^7$ for which $\mathcal O \sim \mathcal O_\xi$, cf. \eqref{eqOrb} and Definition \ref{intrad-defn}.} Let $\mathcal S^7(\rho)\subset \mathfrak{su}(3)$ be the $7$-sphere of radius $\sqrt{2}\,\rho/\sqrt{3}$ centered at the origin, so that the intersection of $\mathcal S^7(\rho)$ with the closed principal Weyl chamber is given by the points
\begin{equation}
	\Xi_{(X,Y)}^\rho = X\,\varpi_1+Y\,\varpi_2 \ , \ \ \begin{cases}
		X,Y\ge 0 \\
		X^2+XY+Y^2 = \rho^2
	\end{cases}.
\end{equation}
The integral orbits of $\mathcal S^7(\rho)$ are given by the integer solutions of
\begin{equation}
	X^2+XY+Y^2-\rho^2 = 0 \ , \ \ X,Y\in\mathbb N_0 \ . 
\end{equation}
For $X,Y \in \mathbb Z$, the quantity $X^2+XY+Y^2$ is the norm of the Eisenstein integer $X-Y\phi$, where $\phi = e^{i2\pi/3}$, so the problem becomes how to factorize $\rho^2$ in prime factors on $\mathbb Z[\phi]$, which is an UFD with units $\{\pm 1, \pm \phi, \pm \phi^2\}$.\footnote{We refer to \cite{cox} for a very nice description of the ring $\mathbb Z[\phi]$.} 
\end{remark}

\begin{example}
As an example for Remark \ref{eisenstein}, take $\rho^2 = 13^2\cdot 43$. The integral orbits of $\mathcal S^7(\rho)$ are solutions of
	\begin{equation}\label{int_orbit_ex}
		(X-Y\phi)\overline{(X-Y\phi)} = 13^2\cdot 43
	\end{equation}
    in $\mathbb Z[\phi]$. The prime factorization of $13$ and $43$ in the ring of Eisenstein integers are, up to the units $\{\pm 1,\pm \phi, \pm \phi^2\}$,
    \begin{equation}
        13 = (3-\phi)(4+\phi) \ \ , \ \ \ 43 = (6-\phi)(7+\phi) \  ,
    \end{equation}
    thus the set of solutions of \eqref{int_orbit_ex} is
    \begin{equation}
    \begin{aligned}
        & \Big\{\alpha\,13(6-\phi)\ , \ \ \alpha\,13(7+\phi) \ , \ \ \alpha(3-\phi)^2(6-\phi)\ ,\ \  \alpha(3-\phi)^2(7+\phi) \, ,\\
        &  \hspace{4 em} \alpha (4+\phi)^2(6-\phi)\ ,\ \  \alpha (4+\phi)^2(7+\phi)\ :\ \alpha = \pm 1,\pm \phi, \pm \phi^2\Big\}\, .
    \end{aligned}
    \end{equation}
    Some solutions represent the same orbit in different Weyl chambers. Restricting to the principal Weyl chamber, given by $X-Y\phi$ with $X,Y\ge 0$, we get the solutions
	\begin{equation}\label{sol_int_orb_ex}
		\Big\{13-78\phi\ ,\ \  78-13\phi \ , \ \ 41-57\phi \ , \ \ 57-41\phi\Big\}
	\end{equation}
    corresponding to the set of rational orbits 
    \begin{equation*}
        \small{\bigg\{\left(\dfrac{1}{\sqrt{43}},\dfrac{6}{\sqrt{43}}\right) \, , \ \left(\dfrac{6}{\sqrt{43}},\dfrac{1}{\sqrt{43}}\right)\, , 
        \left(\dfrac{41}{13\sqrt{43}},\dfrac{57}{13\sqrt{43}}\right)\, ,\ \left(\dfrac{57}{13\sqrt{43}},\dfrac{41}{13\sqrt{43}}\right)\bigg\}}\!\subset\!{\mathcal Q}\, .
    \end{equation*}
\end{example}

\ 

\begin{definition}\label{RadMagooSphere-defn}
    A \emph{radial Magoo sphere} is a Magoo sphere constructed using the radial chain $\mathcal C^r$, cf. Definitions \ref{def:penc_magoo} and \ref{Prn-defn}.
\end{definition}
\begin{remark}\label{canCchoice}
    However, because the radial chain $\mathcal C^r=(\mathcal R_n)_{n\in\mathbb N}$ is a canonical choice and in light of Corollary \ref{cor:magoo_asymp}, from now on we shall always assume this choice \ $\mathcal C=\mathcal C^r$, by default, when we refer to a Magoo sphere in general.    
\end{remark}

Now, we proceed to reverse the order of the iterated limits in Definition \ref{MagooAsymp-defn}.

\begin{definition}\label{unifPoiss}
We say that a Magoo sphere $\mathfrak W\{\mathcal S^7,\widehat\Pi_{\mathfrak g}\}$ of Poisson type  is of \emph{uniform Poisson type} if its Magoo product and Magoo bracket satisfy
\begin{equation}\label{unif_magoo_prod_lim}
    \begin{aligned}
    &\lim_{s\to\infty}\lim_{n\to\infty} f_1\star_{\mathcal C}f_2 = f_1f_2\in \mathfrak P \ \iff \\
      &\big((f_1|_{\mathcal O_\xi}\star^s_\xi f_2|_{\mathcal O_\xi})_{\xi\in \overline{\mathcal Q}}\big)_{s\ge s_0}\xrightarrow{s\to\infty} (f_1f_2|_{\mathcal O_\xi})_{\xi\in \overline{\mathcal Q}}\, , 
    \end{aligned}
\end{equation}
\begin{equation}\label{unif_magoo_bracket_lim}
    \begin{aligned}
      &\lim_{s\to\infty}\lim_{n\to\infty} s[f_1,f_2]^r_{\star_{\mathcal C}} = i\{f_1,f_2\}\in \mathfrak P \ \iff \\  
      &\big((sr(\xi)[f_1|_{\mathcal O_\xi},f_2|_{\mathcal O_\xi}]_{\star^s_\xi})_{\xi\in \overline{\mathcal Q}}\big)_{s\ge s_0}\xrightarrow{s\to\infty}  (i\{f_1,f_2\}|_{\mathcal O_\xi})_{\xi\in \overline{\mathcal Q}}\, , 
\end{aligned}
\end{equation}
for any $f_1,f_2 \in Poly(\mathcal S^7)$, cf. \eqref{lim_s}. In this case, we write 
\begin{equation}\label{WsimP}
    \mathfrak {W}\{\mathcal S^7,\widehat\Pi_{\mathfrak g}\} \xrightarrow{ \ \sim \ } Poly(\mathcal S^7,\widehat\Pi_{\mathfrak g}) \, . 
\end{equation}
\end{definition}

The term \emph{uniform} and notation \eqref{WsimP} are justified by the following: 

\begin{proposition}\label{charUnifPoisson}
$\mathfrak W\{\mathcal S^7,\widehat\Pi_{\mathfrak g}\}$ is of uniform Poisson type if and only if, for every $f_1,f_2 \in Poly(\mathcal S^7)$, we have both
\begin{equation}
    \norm{f_1|_{\mathcal O_\xi}\star^s_\xi f_2|_{\mathcal O_\xi}-f_1f_2|_{\mathcal O_\xi}}_\xi \ \ \mbox{and} \ \ \norm{sr(\xi)[f_1|_{\mathcal O_\xi},f_2|_{\mathcal O_\xi}]_{\star^s_\xi}-i\{f_1,f_2\}|_{\mathcal O_\xi}}_\xi
\end{equation}
converging to $0$ uniformly over $\overline{\mathcal Q}$, as $s\to\infty$.
\end{proposition}
\begin{proof}
This is immediate from the definitions, cf.~\eqref{lim_s}-\eqref{lim_n} and \eqref{unif_magoo_prod_lim}-\eqref{unif_magoo_bracket_lim}, once we note that, in this case, taking the limit $n\to\infty$ first is equivalent to replacing  \eqref{lim_s}-\eqref{lim_n} by just 
\begin{equation}\label{lim_s-Q}
   \lim_{s\to\infty}f^s = f \iff \forall \epsilon>0 , \exists s_\epsilon\in\mathbb N:\forall s\ge s_\epsilon \implies \norm{f^s_\xi-f_\xi}_\xi <\epsilon \ \ \forall \xi \in  \overline{\mathcal Q}\, ,
\end{equation} 
since $\displaystyle{\lim_{n\to\infty}}\mathcal P_n=\overline{\mathcal Q}$ and we  start with common extensions $f_1,f_2\in Poly(S^7)\hookrightarrow \mathfrak P$ to ${f_1}_{|n},{f_2}_{|n}\in \mathfrak P_n$, $\forall n\in\mathbb N$.
\end{proof}

\begin{remark}
From Proposition \ref{charUnifPoisson}, the uniform Poisson property only depends on the pencil of rays of universal correspondences, not on any chain $\mathcal C$ satisfying \eqref{chainCn} which is used to construct the Magoo sphere, just as in  Corollary \ref{cor:magoo_asymp}. So, again, because $\mathcal C^r$ is a canonical choice  we assume  \ $\mathcal C=\mathcal C^r$, by default. 
\end{remark}

However, the relevant question is whether there exists any Magoo sphere of uniform Poisson type. In the next subsection, we start investigating this question for the paradigmatic Magoo sphere of Poisson type:

\begin{definition}\label{BerMagSph}
The \emph{Berezin Magoo sphere} is the Magoo sphere such that,  $\forall \xi\in\overline{\mathcal Q}$, $(w^s_\xi)_{s\in\mathbb N}\equiv(b^s_\xi)_{s\in\mathbb N}$ is the $\xi$-ray of Berezin universal correspondences, cf. Definition \ref{defn:univBerezin} and Proposition \ref{prop:B=b}, and Definition \ref{def:penc_magoo}. We denote it by
\begin{equation}
    \mathfrak{B}\{\mathcal S^7, \widehat\Pi_{\mathfrak g}\}=\big(b_{\mathcal C},\star_{\mathcal C},[\cdot ,\cdot]_{\star_\mathcal C}^r\big) \ . 
\end{equation}
\end{definition}

\subsection{On the asymptotics of the Berezin Magoo sphere}

In this subsection, we prove the following result:

\begin{theorem}\label{theo:unif_conv_compact_ber}
Let $\mathcal K\subset \mathcal F$ be any compact, and denote $\mathcal Q_{\mathcal K}=\mathcal Q\cap\mathcal K$. Then, for the Berezin Magoo sphere, cf. Definition \ref{BerMagSph}, and for any $f_1,f_2 \in Poly(\mathcal S^7)$, 
\begin{equation}\label{ber_magoo_prod_brack_norm}
    \norm{f_1|_{\mathcal O_\xi}\star^s_\xi f_2|_{\mathcal O_\xi}-f_1f_2|_{\mathcal O_\xi}}_\xi \ \ \mbox{and} \ \ \norm{sr(\xi)[f_1|_{\mathcal O_\xi},f_2|_{\mathcal O_\xi}]_{\star^s_\xi}-i\{f_1,f_2\}|_{\mathcal O_\xi}}_\xi
\end{equation}
converge to $0$  uniformly over $\mathcal Q_{\mathcal K}$, as $s\to \infty$. 
\end{theorem}
\begin{proof}
    The proof uses a series of lemmas. The first one is immediate from \eqref{unif-error} in Proposition \ref {prop:error_conv}.

\begin{lemma}\label{unif-Q-conv-beta}
    For any $u \in U(\mathfrak{sl}(3))$, the limit 
    \begin{equation}
	\lim_{s\to \infty}(sr(\xi))^{-\deg(u)}b^s_\xi[u] = (-i)^{\deg(u)}\beta_{\deg(u)}[u]|_{\mathcal O_\xi}
\end{equation}
holds uniformly over $\overline{\mathcal Q}$.
\end{lemma}

The following two lemmas will be used to show that we can ensure the validity of a decomposition similar to Proposition \ref{prop:decomp_b^s_basis} on all orbits in $\mathcal Q_{\mathcal K}$ simultaneously. 

\begin{lemma}\label{lemma:b^s_xi_basis}
    Let $m = \dim(Poly(\mathcal O_{\xi_0})^{\vb*a})$, $\xi_0 \in \mathcal Q$. There are $u_1,...,u_m\in U(\mathfrak{sl}(3))^{\vb*a}$ and $s(\xi_0) \in \mathbb N$, as well as  an open neighborhood $\mathcal U(\xi_0)$ of $\xi_0$ in $\mathcal F$,
    such that 
    \begin{equation}\label{m-basisPoly}
 		\left\{(-i)^{d(1)}\beta_{d(1)}[u_1],...,(-i)^{d(m)}\beta_{d(m)}[u_m]\right\} \ , \ \ d(j)\equiv\deg(u_j)\, , 
 	\end{equation}
    is a basis of $Poly(\mathcal O_{\xi})^{\vb*a}$, for every $\xi\in \mathcal U(\xi_0)\subset \mathcal F$, and in addition, 
	\begin{equation}\label{b^s_xi-basis}
	s\geq s(\xi_0)\implies	\left\{(sr(\xi))^{-(d(1))}b^s_\xi[u_1],...,(sr(\xi))^{-(d(m))}b^s_\xi[u_m]\right\}
	\end{equation}
    is also a basis of $Poly(\mathcal O_{\xi})^{\vb*a}$, for every $\xi\in\mathcal U(\xi_0)\cap \mathcal Q=:\mathcal V(\xi_0)$.
\end{lemma}
\begin{proof}
	Take $\{h_1,...,h_m\}\subset Poly(\mathfrak{su}(3))^{\vb*a}$ such that each $h_j$ is a homogeneous polynomial of degree $d(j)$ and $\{h_1|_{\mathcal O_{\xi_0}},...,h_m|_{\mathcal O_{\xi_0}}\}$ is a basis of $Poly(\mathcal O_{\xi_0})^{\vb*a}$. There are $g_1,...,g_m\in SU(3)$ for which the matrix $H(\xi_0)$ with entries 
    \begin{equation}\label{Hmatrix}  
      (H(\xi_0))_{j,k} = h_k^{g_j}(\xi_0)
    \end{equation}
    is non singular. Consider its extension to the matrix valued function 
    \begin{equation}\label{HmatrixX}
        H:\mathfrak{t}\to M_\mathbb C(m)\, , \ X\mapsto (H(X))_{j,k} = h_k^{g_j}(X) \ , 
    \end{equation}
    from which we have the polynomial $\varphi \in Poly_{d(1)+...+d(m)}(\mathfrak t)$, given by
	\begin{equation}
		\varphi(X) = \det(H(X)) \ . 
	\end{equation}

    By construction, $\varphi(\xi_0)\ne 0$, thus $\mathcal Z = \varphi^{-1}(0)\cap \mathcal F$ is finite, and $\{h_1|_{\mathcal O_{\xi}},...,h_m|_{\mathcal O_{\xi}}\}$ is l.i.~for every $\xi \in \mathcal F\setminus \mathcal Z$. Since $\dim(Poly(\mathcal O_\xi)^{\vb*a})$ is constant on $\mathcal F$, we conclude that $\{h_1|_{\mathcal O_{\xi}},...,h_m|_{\mathcal O_{\xi}}\}$ is a basis of $Poly(\mathcal O_\xi)^{\vb*a}$ for every $\xi \in \mathcal F\setminus \mathcal Z$. 
    
    Therefore, there exists $\epsilon(\xi_0)>0$ such that the closed ball $\overline{B}_{H(\xi_0)}(\epsilon(\xi_0))\subset M_\mathbb C(m)$ of radius $\epsilon(\xi_0)$ centered at $H(\xi_0)$ contains only non singular matrices, that is,
    \begin{equation}\label{Ballepsilonxi0}
        \overline{B}_{H(\xi_0)}(\epsilon(\xi_0))\subset GL_m(\mathbb C) \ , \ \ \exists\,  \epsilon(\xi_0)> 0\, .
    \end{equation} 
    Then, taking
	\begin{equation}\label{Uxi0}
		\mathcal U(\xi_0) = \{\xi \in \mathcal F: \norm{H(\xi)-H(\xi_0)}<\epsilon(\xi_0)/2\}\, ,
	\end{equation}
	we have that
    \begin{equation}\label{HinB}
        H(\xi)\in \overline{B}_{H(\xi_0)}(\epsilon(\xi_0))\subset GL_m(\mathbb C) \ , \ \ \forall \xi\in \mathcal U(\xi_0) \ . 
    \end{equation}

	Now, take $u_j = (-i)^{-d(j)}S(h_j)\in U_{\le d(j)}(\mathfrak{sl}(3))^{\vb*a}$, so that $h_j = (-i)^{d(j)}\beta_{d(j)}[u_j]$. By Lemma \ref{unif-Q-conv-beta}, each $(sr(\xi))^{-d(j)}b^s_\xi[u_j]$ converges to $h_j$ uniformly on $\overline{\mathcal Q}$ as $s\to\infty$. Thus, for the sequence of matrix-valued functions $B(s,\cdot)_{s\in\mathbb N}$, where 
    \begin{equation}
        B(s,\cdot):\overline{\mathcal Q}\to M_{\mathbb C}(m)\, , \ \xi\mapsto  (B(s,\xi))_{j,k} = (sr(\xi))^{-d(k)}b^s_\xi[u_k]^{g_j} \ , 
    \end{equation}
    we have that $B(s,\cdot)_{s\in\mathbb N}$ also converges uniformly to $H$ on $\overline{\mathcal Q}$ as $s\to\infty$, that is,
    \begin{equation}\label{Bs0}
    \forall\epsilon> 0, \exists s_\epsilon\in\mathbb N: \   \forall s\geq s_\epsilon\implies\norm{B(s,\xi)-H(\xi)}< \epsilon/2 \ , \ \ \forall \xi\in \overline{\mathcal Q} \ . 
    \end{equation}
    
    Combining \eqref{Uxi0} and \eqref{Bs0},  \  $\exists\,  s(\xi_0)\in\mathbb N$  \ such that, $\forall \xi \in \mathcal V(\xi_0)$\,,
    \begin{equation}\label{unif-norms-to-0}
    \begin{aligned}
	s\ge s(\xi_0) \implies	\norm{B(s,\xi)-H(\xi_0)}  &\le \norm{B(s,\xi)-H(\xi)}+\norm{H(\xi)-H(\xi_0)}\\
    & < \dfrac{\epsilon(\xi_0)}{2}+\dfrac{\epsilon(\xi_0)}{2}\, ,
    \end{aligned}
	\end{equation}
	which implies $B(s,\xi)$ is non singular too, that is, 
    \begin{equation}\label{BinB}
    B(s,\xi)\in \overline{B}_{H(\xi_0)}(\epsilon(\xi_0))\subset GL_m(\mathbb C) \ , \ \ \forall \xi\in  \mathcal V(\xi_0), \ \forall s\ge s(\xi_0) \ .
    \end{equation}
    
    Since $Poly(\mathcal O_\xi)^{\vb*a}\simeq Poly(\mathcal O_{\xi_0})^{\vb*a}$, $\forall \xi \in \mathcal F$, we conclude from \eqref{HinB} that \eqref{m-basisPoly} is a basis for $Poly(\mathcal O_\xi)^{\vb*a}$, $\forall\xi\in\mathcal U(\xi_0)$, and from \eqref{BinB} that the set in \eqref{b^s_xi-basis} is also a basis of $Poly(\mathcal O_\xi)^{\vb*a}$, $\forall\xi\in\mathcal V(\xi_0)$, $\forall s\geq s(\xi_0)$.
\end{proof}

\begin{lemma}\label{lemma-alpha-infty}
	Let $u_1,...,u_m \in U(\mathfrak{sl}(3))^{\vb*a}$, $\mathcal V(\xi_0)$ and $s(\xi_0)$ be as in the previous lemma. For $f\in Poly(S^7)^{\vb*a}$, $\xi \in \mathcal V(\xi_0)$ and $s\ge s(\xi_0)$, there are $\alpha_j(s,\xi)\in \mathbb C$ for $j\in\{1,...,m\}$, such that
	\begin{equation}
		f = \sum_{j=1}^m\alpha_j(s,\xi)(sr(\xi))^{-d(j)}b^s_\xi[u_j] = \sum_{j=1}^m\alpha_j^\infty(\xi)(-i)^{d(j)}\beta_{d(j)}[u_j]|_{\mathcal O_\xi}\, ,
	\end{equation}
	where
	\begin{equation}
		\lim_{s\to \infty}\alpha_j(s,\xi) = \alpha_j^\infty(\xi)\in\mathbb C
	\end{equation}
	holds uniformly on $\mathcal V(\xi_0)$.
\end{lemma}

\begin{proof}
For $H(\xi)$ and $B(s,\xi)$ as in the previous lemma, let $F(\xi)=H(\xi)^{-1}$ and $A(s,\xi) = B(s,\xi)^{-1}$, for any $\xi\in\mathcal V(\xi_0)$ and $s\geq s(\xi_0)$, cf. \eqref{HinB} and \eqref{BinB}. Then, from \eqref{Bs0}, $B(s,\xi)\xrightarrow[]{s\to\infty}H(\xi)$ implies
	\begin{equation}\label{alphatoalpha}
		\alpha_j(s,\xi) = \sum_{k=1}^m(A(s,\xi))_{j,k}f^{g_k}(\xi)\ \xrightarrow[]{s\to\infty} \ \alpha_j^\infty(\xi) = \sum_{k=1}^m(F(\xi))_{j,k}f^{g_k}(\xi)\, , 
	\end{equation}
so it only remains to show that the convergence in \eqref{alphatoalpha} is uniform over $\mathcal V(\xi_0)$. 
	
    Recall from the proof of the previous lemma that there is $\epsilon(\xi_0)> 0$ for which $\overline{B}_{H(\xi_0)}(\epsilon(\xi_0))\subset GL_m(\mathbb C)$ as in \eqref{Ballepsilonxi0} is compact and both $H(\xi)$ and $B(s,\xi)$ lie in its interior, $\forall\xi\in\mathcal V(\xi_0)$, $\forall s\geq s(\xi_0)$, cf. \eqref{HinB} and \eqref{BinB}. Hence, by continuity of the inversion map on $GL_m(\mathbb C)$, there is $C>0$ such that, if $s\ge s(\xi_0)$ and $\xi \in \mathcal V(\xi_0)$, then $\norm{A(s,\xi)}$ and $\norm{F(\xi)}$ are both bounded by $C$, giving
    \begin{equation}\label{normAF}
    \begin{aligned}
        \norm{A(s,\xi)-F(\xi)} & = \norm{A(s,\xi)(H(\xi) - B(s,\xi))F(\xi)}\\
        & \le C^2\norm{H(\xi) -B(s,\xi)}\, ,
    \end{aligned}
    \end{equation}
    with this last line converging to $0$ uniformly on $\mathcal V(\xi_0)$, cf. \eqref{Bs0}. Therefore, from 
	\begin{equation}
		|\alpha_j(s,\xi)-\alpha_j^\infty(\xi)| \le \sum_{k=1}^m\big|(A(s,\xi))_{j,k}-(F(\xi))_{j,k}\big||f^{g_k}(\xi)| \ , 
	\end{equation}
    \eqref{Bs0} and \eqref{normAF} imply that the convergence in \eqref{alphatoalpha} is uniform over $\mathcal V(\xi_0)$. 
\end{proof}

We now proceed to finish the proof of the theorem. 

Again, by bilinearity of the operations, it is sufficient to show the result for $f_j \in Poly(S^7)^{\vb*a_j}$. Now, for any $\xi_0 \in \mathcal Q_{\mathcal K}$, and $\mathcal U(\xi_0)\subset \mathcal F$ as in Lemma \ref{lemma:b^s_xi_basis}, we have $u_1^j,...,u^j_{m_j}\in U(\mathfrak{sl}(3))^{\vb*a_j}$, with $\deg(u_k^j) = d_j(k)$, such that
	\begin{equation}\label{h^j_k-ber_unif}
		\{h_1^j|_{\mathcal O_\xi},...,h_{m_j}^j|_{\mathcal O_\xi}\} \ , \ \ h_k^j = (-i)^{d_j(k)}\beta_{d_j(k)}[u_k^j] \ , 
	\end{equation}
	is a basis of $Poly(\mathcal O_{\xi})^{\vb*a_j}$ for every $\xi \in \mathcal U(\xi_0)$,  and also exists $s(\xi_0) \in \mathbb N$ such that
	\begin{equation}\label{b[u^j_k]-ber_unif}
		\left\{(sr(\xi))^{-(d_j(1))}b^s_\xi[u_1^j],...,(sr(\xi))^{-(d_j(m_j))}b^s_\xi[u^j_{m_j}]\right\}
	\end{equation}
	is a basis of $Poly(\mathcal O_\xi)^{\vb*a_j}$ for every $\xi \in \mathcal V(\xi_0)$ and $s\ge s(\xi_0)$. 
    
    Hence, from Lemma \ref{lemma-alpha-infty}, there are $\alpha_k^j(s,\xi)\in \mathbb C$ for $k\in\{1,...,m_j\}$ such that
	\begin{equation}
		\lim_{s\to \infty}\alpha_k^j(s,\xi)=:(\alpha_k^j)^\infty_\xi  \in \mathbb C
	\end{equation}
	holds uniformly over $\mathcal V(\xi_0)$, and
	\begin{equation}
		f_j|_{\mathcal O_{\xi}} = \sum_{k=1}^{m_j}\alpha_k^j(s,\xi)(sr(\xi))^{-d_j(k)}b^s_\xi[u_k^j] = \sum_{k=1}^{m_j}(\alpha_k^j)^\infty_\xi h_k^j|_{\mathcal O_\xi} \, , 
	\end{equation}
    for every $\xi\in\mathcal V(\xi_0)$. Therefore, 
	\begin{equation}
		f_1|_{\mathcal O_\xi}\star^s_\xi f_2|_{\mathcal O_\xi} = \sum_{j,k}\alpha_j^1(s,\xi)\alpha_k^2(s,\xi)(sr(\xi))^{-(d_1(j)+ d_2(k))}b^s_\xi[u_j^1u_k^2]
	\end{equation}
	converges to (cf. \eqref{convprod} in Theorem \ref{theo:prod_ber_conv})
	\begin{equation}
		\sum_{j,k}(\alpha_j^1)^\infty_\xi(\alpha_k^2)^\infty_\xi h_j^1|_{\mathcal O_\xi} h_k^2|_{\mathcal O_\xi} = f_1f_2|_{\mathcal O_\xi}
	\end{equation}
	uniformly on $\mathcal V(\xi_0)$. Analogously,
	\begin{equation}
    \begin{aligned}
		& sr(\xi)[f_1|_{\mathcal O_\xi},f_2|_{\mathcal O_\xi}]_{\star^s_\xi}\\
        & \hspace{2 em} = \sum_{j,k} \alpha_j^1(s,\xi) \alpha_k^2(s,\xi)(sr(\xi))^{-(d_1(j)+d_2(k)-1)}b^s_\xi[u_j^1u_k^2-u_k^2u_j^1]
    \end{aligned}
	\end{equation}
	converges uniformly on $\mathcal V(\xi_0)$ to (cf. \eqref{convpb} in Theorem \ref{theo:prod_ber_conv})
	\begin{equation}
		\sum_{j,k}(\alpha_j^1)^\infty_\xi(\alpha_k^2)^\infty_\xi i\{h_j^1,h_k^2\}|_{\mathcal O_\xi} = i\{f_1,f_2\}|_{\mathcal O_\xi} \ . 
	\end{equation}
	
	To finish, by compactness, there exists a finite set $\{\xi_1,...,\xi_k\}\subset \mathcal Q_{\mathcal K}$ such that the open sets $\mathcal U(\xi_1),...,\mathcal U(\xi_k)\subset\mathcal F$ (from which we write the basis \eqref{h^j_k-ber_unif} and \eqref{b[u^j_k]-ber_unif}) cover $\mathcal K\subset\mathcal F$, and therefore $\mathcal V(\xi_1),...,\mathcal V(\xi_k) \subset \mathcal Q$ cover $\mathcal Q_{\mathcal K}$. In the previous paragraph, we have proved that, for any $\epsilon> 0$, there is $s_\epsilon(\xi_j)\in \mathbb N$ such that
    \begin{equation}
    s\geq s_\epsilon(\xi_j)\implies\begin{cases}\norm{f_1|_{\mathcal O_\xi}\star^s_\xi f_2|_{\mathcal O_\xi}-f_1f_2|_{\mathcal O_\xi}}_\xi<\epsilon \\
    \norm{sr(\xi)[f_1|_{\mathcal O_\xi},f_2|_{\mathcal O_\xi}]_{\star^s_\xi}-i\{f_1,f_2\}|_{\mathcal O_\xi}}_\xi<\epsilon
    \end{cases} , \ \ \forall \xi \in \mathcal V(\xi_j)\, .
    \end{equation} 
    Then, taking
    \begin{equation}
        s_\epsilon = \max\{s_\epsilon(\xi_1),...,s_\epsilon(\xi_k)\}\in \mathbb N\, ,
    \end{equation}
    we get
    \begin{equation}
    s\geq s_\epsilon\implies\begin{cases}\norm{f_1|_{\mathcal O_\xi}\star^s_\xi f_2|_{\mathcal O_\xi}-f_1f_2|_{\mathcal O_\xi}}_\xi<\epsilon \\
    \norm{sr(\xi)[f_1|_{\mathcal O_\xi},f_2|_{\mathcal O_\xi}]_{\star^s_\xi}-i\{f_1,f_2\}|_{\mathcal O_\xi}}_\xi<\epsilon
        \end{cases} , \ \ \forall \xi \in \mathcal Q_{\mathcal K}\, .
    \end{equation} 
\end{proof}

\begin{remark}\label{Berspecial}
    We emphasize that the uniform convergence established in Lemma \ref{unif-Q-conv-beta} is a special property of the Berezin Magoo sphere which does not hold for general Magoo spheres of Poisson type. 
    
    For example, for any $\xi \in \overline{\mathcal Q}$, consider the $\xi$-ray $(w^s_\xi)_s$ of universal correspondences given by the following rule: for every $u \in U(\mathfrak{sl}(3))^{(a,b)}$,
    \begin{equation}
        w^s_\xi[u] = \begin{cases}
            
        \left(1+\dfrac{r(\xi)}{s}\right)b^s_\xi[u] \ \ \ \  \mbox{if} \ \ \ \ (a,b)\ne (0,0)\\
        \ \ \ \ \ \ \ \ b^s_\xi[u] \ \ \ \ \ \ \ \ \ \ \  \mbox{otherwise}
        \end{cases}.
    \end{equation}
    Then, for any $u \in U(\mathfrak{sl}(3))$, we have
    \begin{equation}\label{w-non_unif}
        \lim_{s\to\infty}(sr(\xi))^{-\deg(u)}w^s_\xi[u] = (-i)^{\deg(u)} \beta_{\deg(u)}[u] \ , \ \ \forall \xi \in \overline{\mathcal Q}\, ,
    \end{equation}
    which means each $\xi$-ray $(w^s_\xi)$ is of Poisson type. 
    
    However, if $u\in U(\mathfrak{sl}(3))$ lies in any non trivial irrep, then
    \begin{equation}
        (sr(\xi))^{-\deg(u)}w^s_\xi[u] - (-i)^{\deg(u)} \beta_{\deg(u)}[u] = \varepsilon^s_\xi[u]+\dfrac{r(\xi)}{s}b^s_\xi[u]\, ,
    \end{equation}
    where $\varepsilon^s_\xi$ is the error function of $b^s_\xi$. Since the integral radius function $r$ is unbounded on any neighborhood of any $\xi \in \overline{\mathcal Q}$, cf. Proposition \ref{unboundedr}, the convergence \eqref{w-non_unif} is not uniform anywhere.
    \end{remark}

    In view of the previous remark, we have the following: 

    \begin{proposition}\label{Berspecial-prop}
        For a general Magoo sphere of Poisson type, the uniform Poisson property may not be satisfied for any neighborhood of any $\xi\in\overline{\mathcal Q}$. 
    \end{proposition}
\begin{proof}
    Since it is enough to show this non-uniformity in a single example, we show it explicitly for a single polynomial $\star$-product  in the example of Remark \ref{Berspecial}. 

    Thus, let $u \in U_1(\mathfrak {sl}(3))\equiv \mathfrak{sl}(3)$ be a highest weight vector, so that $u^2 \in U_2(\mathfrak{sl}(3))$ is a highest weight vector for a representation $(2,2)$. Then,
    \begin{equation}
    \begin{aligned}
    f &= -i\beta_1[u] \implies \\
        f|_{\mathcal O_\xi} = (sr(\xi))^{-1}b^s_\xi[u] &=\left(sr(\xi)\left(1+\dfrac{r(\xi)}{s}\right)\right)^{\!-1}\!\!w^s_\xi[u]\, \ , \ \  \forall \, \xi \in \overline{\mathcal Q}\,,
    \end{aligned}
    \end{equation}
    and, for the twisted product $\star^s_\xi$ induced by $w^s_\xi$, we have
    \begin{equation}
    \begin{aligned}
        f|_{\mathcal O_\xi}\star^s_\xi f|_{\mathcal O_\xi} & = (sr(\xi))^{-2}\left(1+\dfrac{r(\xi)}{s}\right)^{-1}b^s_\xi[u^2]\\&  = \left(1+\dfrac{r(\xi)}{s}\right)^{-1}\left(f^2|_{\mathcal O_\xi}+\varepsilon^s_\xi[u^2]\right)\,.
    \end{aligned}
    \end{equation}
    By the triangular inequality,
    \begin{equation}\label{triang-ineq}
    \begin{aligned}
        \norm{f|_{\mathcal O_\xi}\star^s_\xi f|_{\mathcal O_\xi}- f^2|_{\mathcal O_\xi}}_\xi & \ge\Bigg|\,\norm{\left(1+\dfrac{r(\xi)}{s}\right)^{-1}f^2|_{\mathcal O_\xi}-f^2|_{\mathcal O_\xi}}_\xi\\
        & \hspace{3 em} - \norm{\left(1+\dfrac{r(\xi)}{s}\right)^{-1}\varepsilon^s_\xi[u^2]}_\xi\Bigg| \ .
    \end{aligned}
    \end{equation}
    For the last term in the r.h.s.~of \eqref{triang-ineq}, from Proposition \ref{prop:error_conv}, we have 
    \begin{equation}\label{rhsto0}
        \norm{\left(1+\dfrac{r(\xi)}{s}\right)^{-1}\varepsilon^s_\xi[u^2]}_\xi \le \dfrac{M(u^2)}{r(\xi)(s+r(\xi))}\, ,
    \end{equation}
    and hence this vanishes uniformly over $\overline{\mathcal Q}$. But on the other hand, 
    \begin{equation}
        \norm{\left(1+\dfrac{r(\xi)}{s}\right)^{-1}f^2|_{\mathcal O_\xi}-f^2|_{\mathcal O_\xi}}_\xi = \dfrac{r(\xi)}{s+r(\xi)}\norm{ f^2|_{\mathcal O_\xi}}_\xi \ , 
    \end{equation}
    and this does not vanish uniformly anywhere, since $r$ is unbounded on any  neighborhood of $\overline{\mathcal Q}$, cf. Proposition \ref{unboundedr}. 
    Hence, although the l.h.s.~of \eqref{triang-ineq} vanishes as $s\to\infty$,  $\forall\xi\in\overline{\mathcal Q}$, it does not vanish uniformly in any neighborhood of any $\xi\in\overline{\mathcal Q}$.
\end{proof}

Thus, from the bijection $\overline{\mathcal F}\ni\xi\leftrightarrow \mathcal O_\xi\subset \mathcal S^7$, Theorem \ref{theo:unif_conv_compact_ber} states that we have Poisson uniformity for any compact  Berezin Magoo ``cylinder'', that is, we have 
\begin{equation}
    \mathfrak{B}\{\mathcal S^7|_{\mathcal K}, \widehat\Pi\}\xrightarrow[]{\sim}Poly(\mathcal S^7|_{\mathcal K}, \widehat\Pi) \ , 
\end{equation}
cf. \eqref{WsimP} in Definition \ref{unifPoiss}, where $\mathcal S^7|_{\mathcal K}$ is the compact ``cylinder''
\begin{equation}
  \mathcal S^7|_{\mathcal K}=  \bigcup_{\xi\in\mathcal K}\mathcal O_\xi \ \subset \ \mathcal S^7\, . 
\end{equation}

\begin{remark}\label{open-unifPoisson}
However, we haven't yet been able to prove or disprove Poisson uniformity of the whole Berezin Magoo sphere, that is, for the whole $\overline{\mathcal Q}$. Thus, the question of whether there is a Magoo sphere of uniform Poisson type remains open. 
\end{remark}

\section{Concluding remarks}\label{sec:conc}

In this series of two papers on quark systems, we explored the properties and results for $SU(3)$ in detail, which allowed us to paint a clear and detailed picture of quantum and classical quark systems and their relationship via symbol correspondences and semiclassical asymptotics. However, a lot of what has been done for $SU(3)$ generalizes to other compact symmetry groups. So, here we conclude this series by highlighting what can be generalized to other groups and commenting on some peculiarities of $SU(3)$. We shall proceed by decreasing order of generality,  summarizing the main arguments, and refer to \cite{Thesis} for a more complete analysis. 

In Remark I.3.4, we indicated that the material of section I.3 holds for any compact Lie group. Indeed, let $G$ be a connected compact Lie group with Lie algebra $\mathfrak g$. If $\rho$ is a unitary $G$-irrep on $\mathcal H$, then it is finite dimensional \cite{nach2}, hence the space $\mathcal B(\mathcal H)$ of all operators on $\mathcal H$ is also finite dimensional and carries a unitary (with respect to the trace inner product) $G$-representation. Also, given a Hamiltonian $G$-space $P$, we can use the isomorphism $P \simeq G/G_0$, where $G_0$ is the isotropy subgroup of some point $\vb*\varsigma_0\in P$, to descend the Haar measure of $G$ to $P$ so that $C^\infty_{\mathbb C}(P)\subset L^2(P)$. Thus, defining symbol correspondences from $\mathcal B(\mathcal H)$ to $C^\infty_{\mathbb C}(P)$ analogously to Definition I.3.1, everything done in section I.3 follows.

Besides that, the representation on $\mathcal B(\mathcal H)$ is completely reducible because it is a unitary representation on a finite dimensional space.\footnote{Note that the natural isomorphism $\mathcal B(\mathcal H) \simeq \mathcal H\otimes \mathcal H^\ast$ allows us to write this representation as the tensor product of $\rho$ with its dual representation, so the decomposition of $\mathcal B(\mathcal H)$ into irreps is an instance of Clebsch-Gordan series.} Also, by the Peter-Weyl Theorem and the already stated isomorphism $P\simeq G/G_0$, the space $L^2(P)$ inherits a decomposition into irreps from $L^2(G)$, with orthonormal basis comprised by smooth harmonic functions \cite{foll}. These decompositions of operators and functions lead to the characterization of symbol correspondences by characteristic matrices (characteristic numbers for highest symmetry) in the sense of sections I.4 and I.5.

Moreover, $P$ covers a coadjoint orbit $\mathcal O\subset \mathfrak g^\ast$ via the momentum map, so the coadjoint orbits are of particular interest as models of Hamiltonian $G$-spaces and there are only finitely many types of them \cite{kir}. For the methods of Paper II, the argument used to identify the space of polynomials on an orbit with the linear span of harmonic functions works for general compact Lie groups, so one may reason it's fairer to restrict the codomain of symbol correspondences to space of polynomials $Poly(\mathcal O)$ defined as in \eqref{Poly(O)}, but now replacing $\mathfrak{su}(3)$ by $\mathfrak g^\ast$. 

Henceforth we make the further assumptions that compact $G$ is semisimple (so the Killing form provides an identification $\mathfrak g \leftrightarrow \mathfrak g^\ast$ and it doesn't matter whether we work with coadjoint or adjoint action \cite{hump}) and simply connected (which implies that the irreps of $G$ are all determined by the Theorem of Highest Weight \cite{hump} and that the (co)adjoint $G$-orbits are simply connected, so they are the unique Hamiltonian $G$-spaces \cite{arms}). Therefore, the irreps obtained from dominant weights and the (co)adjoint orbits exhausts all the possibilities of quantum and classical systems, respectively, for which there are symbol correspondences. 

A general result due to Wildberger \cite{figueroa, wild} (that we specialized for quark systems in Theorem I.5.24) says even more: let $\omega$ be a dominant weight of $\mathfrak g$ and $\xi = \omega/\norm{\omega}$, so that we write $\mathcal H_\omega$ for an irrep with highest weight $\omega$ and $\mathcal O_\xi$ for the orbit of $\xi$, then the set of symbol correspondence from $\mathcal B(\mathcal H_\omega)$ to $Poly(\mathcal O_\xi)$ is not empty, it contains a Berezin correspondence (defined via highest weight $\omega$).

Furthermore, the arguments in section II.\ref{sec:basic_frame}, about the inadequacy of formally deforming the algebra of $C^\infty_{\mathbb C}(\mathcal O)$ and proceeding instead by looking at sequences of twisted algebras of increasing finite dimensions, apply in this more general context, since Proposition \ref{CstarE} generalizes to any pair $(G,\mathcal O)$, where $G$ is a compact simply connected semisimple Lie group and $\mathcal O$ any of its (co)adjoint orbits. 

Then, similarly to section II.\ref{sec:basic_frame}, for the complexification $\mathfrak g_{\mathbb C}$ of $\mathfrak g$, we get an isomorphism $\beta_{\mathfrak g_{\mathbb C}}:U(\mathfrak{g}_{\mathbb C})\to Poly(\mathfrak{g}_{\mathbb C})$ from the PBW Theorem in the same vein of \eqref{beta_basis_B_infty}, and $Poly(\mathfrak g_{\mathbb C})$ can be properly identified with $Poly(\mathfrak g)$ so that the pointwise product and the Poisson bracket on $Poly(\mathfrak g)$ are given by expressions analogous to \eqref{point_prod_pi} and \eqref{pois_b_pi}, respectively. Pullbacks of symbol correspondences to the universal enveloping algebra are available as well, so universal Berezin correspondences (recall  Wildberger's argument) are given as in Proposition \ref{prop:B=b}, now using $\beta_{\mathfrak g_{\mathbb C}}$.

Thus, everything points to generalizing the definitions of rays of universal correspondences, cf. Definition \ref{def:pois_ray}, and the ones of Poisson type, cf. Definition \ref{PoissonOxi}, in this larger context, wherein the proof of Theorem \ref{theo:prod_ber_conv} suits well -- we refer again to \cite{karab}. Hence, it should be clear that the criteria in Theorems \ref{cor:pois_crit}, \ref{theo:pois_char_m_iso} and \ref{theo:pois_char_m_ber}
hold in the context of any semisimple simply connected compact Lie group. 

For the unit sphere $\mathcal S\subset \mathfrak g$, we still have a countable dense subset of the orbit space $\mathcal S/G$ comprised by orbits that are equivalent to highest weight orbits in $\mathfrak g$, in the sense of Definition \ref{ratorbdef}, leading to generalizations of the integral radius and the coarse Poisson sphere, cf. Definitions \ref{intrad-defn} and \ref{def:rat_coars}. To properly extend the notion of Magoo sphere, we need invariant polynomials satisfying \eqref{deltaP2}, and they can be constructed using the Harish-Chandra Theorem and the Chevalley Theorem. Then, results analogous to Theorem \ref{theo:Magoo_pois} and Corollary \ref{cor:magoo_asymp} are available. 

Besides, a similar version of Theorem \ref{theo:unif_conv_compact_ber} holds for any compact simply connected semisimple group $G$, that is, the (highest weight) Berezin correspondences for $G$ satisfy the Poisson property uniformly on compact sets of the regular stratum of the symplectic foliation of $\mathcal S$,  because the fundamental premise of such result is the fact that the error maps of Berezin correspondences vanish uniformly, as asserted in Proposition \ref{prop:error_conv}, whose statement holds in this greater generality.

Now, for some peculiarities from $SU(3)$. Although not necessary for the main argument in subsection II.\ref{main-arguments}, we suspect that Theorem \ref{theo:const_quant_cp2} for $(SU(n), \mathbb C P^{n-1})$ can be generalized from  $n=2,3$ to $n>3$, but we still don't know if this is true.

Also, for spin systems the relation between Berezin and Stratonovich-Weyl symbol correspondences is rather direct, something we lost for mixed quark systems, cf. Remark I.5.27. But since Stratonovich-Weyl correspondences  and, more generally, semi-conformal correspondences are also special, it would be interesting to investigate their relation to Berezin correspondences in more detail, still in the case of $SU(3)$, and then see how much more complex this relation can get as we move to $SU(4)$ and beyond to other compact Lie groups. 

In particular, a pertinent question to be answered, still in the context of $SU(3)$, is whether there exists a Magoo sphere constructed from Stratonovich-Weyl correspondences, or semi-conformal correspondences, which is of uniform Poisson type, cf. Definition \ref{unifPoiss}, or at least satisfies the uniform Poisson property on compacts of the regular part of the foliation of the unit sphere, as  proved for the Berezin Magoo sphere in Theorem \ref{theo:unif_conv_compact_ber}. Because, although we have not yet answered the question of Poisson uniformity for the whole Berezin Magoo sphere, cf. Remark \ref{open-unifPoisson}, the missing part is the one containing the singularities of the symplectic foliation of $\mathcal S^7$. But moving forward to $SU(4)$, and beyond to $SU(n)$, this question could get harder, since the singular foliation of the unit sphere by (co)adjoint orbits is stratified and has deeper singularities.\footnote{We refer to \cite{guillemin} for a description of the (co)adjoint  orbits and their foliation of $\mathfrak{su}(n)$.} So, while for $SU(3)$ the singular foliation of the Poisson unit sphere has only two isolated singular orbits and the 
singularities are of the simplest possible type, Morse-Bott type, already in the case of $SU(4)$ the intersection of the principal Weyl chamber with the unit sphere is a closed spherical triangle, with its interior mapping to the regular stratum of the symplectic foliation and its edges to the singular strata, wherein the vertices map to the deeper singular orbits which are isomorphic to $\mathbb C P^3$. Thus, it is conceivable that this more elaborate singular structure, with qualitatively different ways of reaching the deeper singularities starting from the regular stratum,  could play a role in the question of Poisson uniformity of Magoo spheres  for $\mathfrak{su}(4)$. And so on for $\mathfrak{su}(n)$.

On the other hand, for any compact semisimple Lie group $G$ of rank $2$ the symplectic foliation of the unit sphere in $\mathfrak{g}$ is parameterized by a closed arc of circumference and the stratification of singular orbits is trivial. Besides $SU(3)$, there are two other 
such groups that are simply connected, namely: $SU(2)\times SU(2)\simeq Spin(4)$ and $Sp(2)\simeq Spin(5)$.\footnote{For $n\geq 3$, the group $Spin(n)$ is the (universal) double cover of the special orthogonal group $SO(n)$, but for $n=3,4,5$, we have the isomorphisms $Spin(3)\simeq SU(2)\simeq Sp(1)$, $Spin(4)\simeq SU(2)\times SU(2)$, $Spin(5)\simeq Sp(2)$, where $Sp(n)$ is the group of $n\times n$ unitary matrices over the quaternions, also called the compact symplectic group.} In the former case, the generic (co)adjoint orbits are isomorphic to $\mathcal S^2 \times \mathcal S^2$, whereas the degenerate ones are isomorphic to $\mathcal S^2$, with Morse-Bott singularities for the symplectic foliation of $\mathcal S^5 \subset \mathfrak{su}(2)\oplus \mathfrak{su}(2)\simeq \mathfrak{so}(4)$. However, we lack a similar understanding of the orbit foliation in the latter case. Thus, it could be interesting to work both cases in full details.

Furthermore, in the case of $SU(2)$ there is more freedom for the signs of the characteristic numbers of Berezin correspondences, than in the case of $SU(3)$ (compare Section I.4.3 to \cite[Section 6.2.3]{RS}). Now, sign changing is an involution, but two standard involutions present in $SU(2)$ do not generalize in form for $SU(3)$. First, $-\mathds 1$ is a central involution of $SU(2)$, but $-\mathds 1\notin SU(3)$. 
Also, the longest element of the Weyl group $W$ of a semisimple Lie group $G$ is always an involution of $W$, but it is not always a central element of $W$. This is so for $SU(2)$, but not for $SU(3)$. 
Hence, it could also be interesting to see if we get more freedom for signs of the characteristic numbers and matrices of Berezin correspondences, for other groups for which one or both of these central involutions are present.\footnote{Both of these central involutions are present for  $Spin(2n+1)$ and $Spin(4n)$, for instance.}

Finally, it could be interesting to expand on the investigations of asymptotic localization, in a general and systematic way as was done in \cite{AR} for spin systems, now in the context of quark systems. In the same vein, one could try working out the formalism of sequential quantizations, in a complete and detailed way as was done for $\mathcal S^2$ in \cite{AR}, now for the (co)adjoint orbits of $\mathfrak{su}(3)$, and eventually, perhaps, joining them together along the coarse Poisson sphere, if possible.

\bibliographystyle{plain}
\bibliography{main}

\

\appendix

\section{A proof of Proposition \ref{ThmBCN}}\label{proofThmBerCharNumbers}
From Proposition I.4.17,
	\begin{equation}\label{berezin-cn}
		b_n^p = (-1)^p\sqrt{\dfrac{(p+1)(p+2)}{2(n+1)^3}}\cg{(p,0)}{(0,p)}{n}{(p,0,0)}{(0,p,p)}{(n,n,n),0} \, ,
	\end{equation}
	so we just need to compute these CG coefficients. Let $a_0,...,a_n \in \mathbb R$ be such that
	\begin{equation}
		T_-^n(\vb*e(n;(2n,0,n),n/2)) = \sum_{J = 0}^n a_J \vb*e(n;\vb*0_n, J)\, .
	\end{equation}
	We know that
	\begin{equation}
		\ip{\vb*e(\vb*p;(p,0,0))\otimes \widecheck{\vb*e}(\widecheck{\vb*p};(0,p,p))}{\vb*e(n;\vb*0_n, J)} \ne 0 \iff J = 0\, .
	\end{equation}
	From (I.2.26), we have
	\begin{equation}
		a_0 = \dfrac{\sqrt{(2n+1)!}}{n+1}\, .
	\end{equation}
	Applying $U_-^n$ to \eqref{pq_CG_choice}, we obtain
	\begin{equation}
		\begin{aligned}
			& \vb*e(n;(2n,0,n),n/2) = \dfrac{(-1)^n}{\mu_n(p)}T_+^n\\
            & \hspace{3 em} = \dfrac{(-1)^p n!}{\mu_n(p)}\sqrt{p\choose n}\vb*e(\vb*p;(p,0,0))\otimes \widecheck{\vb*e}(\widecheck{\vb*p};(n,p-n,p))\\
			& \hspace{4 em} + \sum_{\substack{j+k+l = p\\j\ne p}}c_{j,k,l}\,\vb*e(\vb*p;(j,k,l))\otimes \widecheck{\vb*e}(\widecheck{\vb*p};(p-j+n,p-k-n,p-l))\, .
		\end{aligned}
	\end{equation}
	Again from (I.2.26), we have\footnote{Note that $\widecheck{\vb*e}(\widecheck{\vb*p};(j,k,l))$ has weight $(j-k)/2$ for the subrepresentation $(2p-l)/2$ of $t$-standard $SU(2)$.}
	\begin{equation}
		T_-^n(\widecheck{\vb*e}(\widecheck{\vb*p};(n,p-n,p))) = n!\sqrt{p\choose n}\widecheck{\vb*e}(\vb*p;(0,p,p))\, ,
	\end{equation}
	then
	\begin{equation}
		\begin{aligned}
			& \dfrac{\sqrt{(2n+1)!}}{n+1}\cg{(p,0)}{(0,p)}{n}{(p,0,0)}{(0,p,p)}{(n,n,n),0} \\
            & \hspace{5 em} = \ip{\vb*e(\vb*p;(p,0,0))\otimes \widecheck{\vb*e}(\widecheck{\vb*p};(0,p,p))}{T_-^n(\vb*e(n;(2n,0,n),n/2))}\\
			& \hspace{10 em} = \dfrac{(-1)^p}{\mu_n(p)}\left(n!\sqrt{p\choose n}\right)^2\, .
		\end{aligned}
	\end{equation}
	Using the expression for $\mu_n(p)$ in \eqref{pq_CG_choice}, we get
	\begin{equation}
		\cg{(p,0)}{(0,p)}{n}{(p,0,0)}{(0,p,p)}{(n,n,n),0} = (-1)^p\sqrt{\dfrac{2(n+1)^3}{(p+1)(p+2)}}\sqrt{\dfrac{{p\choose n}}{{p+n+2 \choose n}}}\, .
	\end{equation}
    
	Therefore,
	\begin{equation}\label{b_n}
			b_n^p = \sqrt{\dfrac{{p\choose n}}{{p+n+2 \choose n}}} = \prod_{m=1}^n\sqrt{\dfrac{1-(m-1)/p}{1+(m+2)/p}} > 0\, .
	\end{equation}
    Since the function
    \begin{equation}
        f(x)= \prod_{m=1}^n\sqrt{\dfrac{1-(m-1)x}{1+(m+2)x}}
    \end{equation}
    is analytic around $0$, we have that
    \begin{equation}
        \lim_{p\to\infty}p(b_n^p-1) = f'(0) = -\dfrac{n(n+2)}{2}\, ,
    \end{equation}
    that is, $|b_n^p-1| \in O(1/p)$, $\forall n\in\mathbb N$.

\section{Alternative proof of Corollary \ref{cor:pois_cp2}}\label{app:pq-asymp}

In this appendix, our main goal is to indicate an alternative approach to prove Corollary \ref{cor:pois_cp2} using the symmetries of Clebsch-Gordan coefficients established by Theorem I.2.16. We won't present full proofs for the statements in this appendix, but we outline all the arguments and refer to \cite{Thesis} for a complete treatment.

For $(x_1,...,x_8)$ the coordinates w.r.t.~the orthonormal basis $\{E_j:j=1,...,8\}$, we resort to the following helpful coordinates:
\begin{equation}\label{(1,1)_poly}
\begin{aligned}
    t_+=x_1+ix_2 \ ,& \ \ 
        t_-=x_1-ix_2 \ , \ \ 
        v_+=x_4+ix_5 \ , \\
        v_- =x_4-ix_5 \ , & \ \
        u_+=x_6+ix_7 \ , \ \
        u_- = x_6-ix_7 \ , \\  
         t=x_3 & \ , \ \ 
        u = (\sqrt{3}x_8-x_3)/2\, .
\end{aligned}
\end{equation}
Indeed, using these coordinates, we have
\begin{equation}\label{pois_bi_vec_op}
    \begin{aligned}
        i\,\Pi_{\mathfrak g} & = \sqrt{2}\Big(\partial_{t_+}\otimes T_++\partial_{t_-}\otimes T_-+\partial_{v_+}\otimes V_++\partial_{v_-}\otimes V_-\\
        &\ \ \ \ \ +\partial_{u_+}\otimes U_++\partial_{u_-}\otimes U_-+\partial_t\otimes T_3+\partial_u\otimes U_3\Big)\, ,
        \end{aligned}
    \end{equation}
and, for the harmonic functions,
\begin{equation}\label{X^1_poly}
	\begin{aligned}
		X^1_{(2,1,0),1/2} \equiv 2 v_+ &\ , \ \ X^1_{(2,0,1),1/2} \equiv - 2t_+ \ , \ \ X^1_{(1,2,0),1} \equiv 2u_+ \ ,\\
		X^1_{(1,0,2),1} \equiv 2u_- &\ , \ \ X^1_{(0,2,1),1/2} \equiv 2t_- \ , \ \ X^1_{(0,1,2),1/2} \equiv 2v_-\ ,\\
		X^1_{\vb*0_1,1} &\equiv -2\sqrt{2} u\ , \ \ X^1_{\vb*0_1,0} \equiv 2\sqrt{\dfrac{2}{3}}(2t+u) \, ,
	\end{aligned}
\end{equation}
so $X^n_{\vb*\nu, J}\in Poly_n(\mathcal O_{(1,0)})$ for every $n$. Thus, $\mathcal X_p = Poly_{\le p}(\mathcal O_{(1,0)})$ is the image of $W^p$, cf. Corollary I.4.10. Furthermore, we set $\mathcal X = Poly(\mathcal O_{(1,0)})$.

Now, let $(W^p)$ be a sequence of symbol correspondences as in \eqref{seq_symb_corresp_pq}, with characteristic numbers $c_n^p$. Then each $W^p$ induces a twisted product $\star^p$ on $\mathcal X_p$. The route for the alternative semiclassical analysis is summarized in the following steps:
\begin{enumerate}[1.]
    \item Verify that
    \begin{equation}\label{prod_conv}
        f_1\star^pf_2\to f_1f_2
    \end{equation}
    for every $f_1 \in \mathcal X_1$ and $f_2 \in \mathcal X$ if $c_n^p\to 1$ as $p\to \infty$ for every $n\ge 1$. In addition, Poisson condition and $c_1^p\to 1$ together give that $c_n^p\to 1$, for every $n\ge 1$. 
    
    \item Apply induction to conclude that \eqref{prod_conv} holds
    for every $f_1,f_2 \in \mathcal X$ if $c_n^p\to 1$ as $p\to \infty$, for every $n\ge 1$.

    \item Show that, if $c_n^p\to 1$ as $p\to \infty$, for every $n\ge 1$, then $\norm{[f_1,f_2]_{\star^p}}\in O(1/p)$ for every $f_1,f_2 \in \mathcal X$.

    \item Prove that the convergence $c_1^p\to 1$ as $p\to \infty$ is equivalent to
    \begin{equation}\label{comm_conv}
        p[f_1,f_2]_{\star^p}\to i\sqrt{\dfrac{3}{2}}\{f_1,f_2\}
    \end{equation}
    for every $f_1 \in \mathcal X_1$ and every $f_2\in \mathcal X$.

    \item By induction again, based on the previous two steps, show that $c_n^p\to 1$ as $p\to \infty$, for every $n\ge 1$, also gives \eqref{comm_conv} for every $f_1,f_2 \in \mathcal X$.
\end{enumerate}

Therefore, if $(W^p)$ is of Poisson type, then Steps 1 and 4 together imply that the characteristic numbers satisfy $c_n^p\to 1$ as $p\to \infty$ for all $n\ge 1$; on the other hand if all the characteristic numbers converge to $1$, then Steps 2 and 5 show that $(W^p)$ is of Poisson type. This proves Corollary \ref{cor:pois_cp2}.

We now analyze each of the Steps $1$ through $5$, as stated above. 

\ 

\noindent\underline{\bf Step 1}. \ 
From Lemma I.2.18 and Theorem I.2.20, the star product of harmonic functions on $\mathcal O_{(1,0)}\simeq \mathbb CP^2$ can be straightforwardly seen to satisfy
\begin{equation}\label{X1_starXn}
\begin{aligned}
    X^1_{\vb*\nu_1,J_1}\star^p X^n_{\vb*\nu_2,J_2}& = \sqrt{\delta(p)}\sum_{m=n-1}^{n+1}\sum_{\sigma}\sum_{\substack{\vb*\mu, I\\ \vb*\nu, J}}(-1)^p\dfrac{c_m^p}{c_1^pc_n^p}\cg{1}{n}{(m;\sigma)}{\vb*\nu_1 J_1}{\vb*\nu_2 J_2}{\vb*\nu J}\\
    & \hspace{5 em} \times\cg{1}{n}{(m; \sigma)}{\vb*\mu I}{\widecheck{\overline{\vb*\mu}} I}{\vb*0_m\, 0}\begin{bmatrix}
    1 & n & m\\
    \vb*\mu, I & \widecheck{\overline{\vb*\mu}}, I & \vb*0_m, 0 
    \end{bmatrix}\!\![\vb* p]X^m_{\vb*\nu J}\, ,
\end{aligned}
\end{equation}
where
\begin{equation}
    \begin{aligned}
    \overline{(2,1,0)} = (n+1,n,n-1) \ , \ \ \overline{(2,0,1)} = (n+1,n-1,n) \, , \\ \overline{(1,2,0)} = (n,n+1,n-1)\ , \ \
    \overline{(0,1,2)} = (n-1,n,n+1)  \, , \\ \overline{(0,2,1)} = (n-1,n+1,n) \ , \ \ \overline{(1,0,2)} = (n,n-1,n+1)\, .
    \end{aligned}
\end{equation}
By determining a proportionality
\begin{equation}
\begin{aligned}
    &\sum_{\substack{\vb*\mu, I}}\cg{1}{n}{(m;\sigma)}{\vb*\mu I}{\widecheck{\overline{\vb*\mu}} I}{\vb*0_m\, 0}\begin{bmatrix}
    1 & n & m\\
    \vb*\mu, I & \widecheck{\overline{\vb*\mu}}, I & \vb*0_m, 0 
    \end{bmatrix}\!\![\vb* p]\\
    & \hspace{7 em}\propto \cg{1}{n}{(m;\sigma)}{\vb*0_1 0}{\vb*0_n 0}{\vb*0_m 0}\begin{bmatrix}
    1 & n & m\\
    \vb*0_1,0 & \vb*0_n, 0 & \vb*0_m, 0 
    \end{bmatrix}\!\![\vb* p]\, ,
\end{aligned}
\end{equation}
up to order $O(1/p^2)$, we get the following key result.

\begin{proposition}\label{prop:step1}
For $n\ge 1$, we have
\begin{equation}
\begin{aligned}
    X^1_{\vb*\nu_1,J_1}\star^p X^n_{\vb*\nu_2, J_2} & = \sum_{m = n-1}^{n+1}\sum_{\sigma}\sum_{\vb*\nu, J}\dfrac{c_m^p}{c_1^pc_n^p}f_{n,m}(p)\cg{1}{n}{(m;\sigma)}{\vb*\nu_1 J_1}{\vb*\nu_2 J_2}{\vb*\nu J}\\
    & \hspace{8 em}\times \cg{1}{n}{(m; \sigma)}{\vb*0_1 0}{\vb*0_n 0}{\vb*0_m 0}X^m_{\vb*\nu, J}\,+O((c_1p)^{-1})\, ,
\end{aligned}
\end{equation}
where
\begin{equation}
    f_{n,n}(p) = (-1)^p\sqrt{\delta(p)}\dfrac{(2n+1)(2n+3)}{n(n+2)}\begin{bmatrix}
    1 & n & n\\
    \vb*0_1,0 & \vb*0_n, 0 & \vb*0_n, 0 
    \end{bmatrix}\!\![\vb* p]
\end{equation}
and, for $m\in\{n-1,n+1\}$,
\begin{equation}
    f_{n,m}(p) =
        (-1)^p\sqrt{\delta(p)}\dfrac{4(m+n+2)(n+1)}{3(m+1)^2}\begin{bmatrix}
    1 & n & m\\
    \vb*0_1,0 & \vb*0_n, 0 & \vb*0_m, 0 
    \end{bmatrix}\!\![\vb* p]\, .
\end{equation}
Also, the contribution $O((c_1^pp)^{-1})$ comes from $m=n$.
\end{proposition}
\begin{skproof}
The statement follows from exhaustive application of ladder operators $U_-$ and $T_-$ on
\begin{equation}\label{uncoupled-m0}
    \vb*e((m;\sigma);\vb*0_m, 0) = \sum_{\vb*\mu, I}\cg{1}{n}{(m;\sigma)}{\vb*\mu I}{\widecheck{\overline{\vb*\mu}} I}{\vb*0_m 0} \vb*e(1;\vb*\mu, I)\otimes \vb*e(n;\widecheck{\overline{\vb*\mu}}, I)\, .
\end{equation}
This is, however, more subtle when $m = n$, where we need
\begin{equation}
    T_+(\vb*e(n;\overline{(021)}, 1/2)) = - \mu_1(p)\Big[\vb*e(1;(201),1/2),(\vb*e(n;\overline{(021)}, 1/2))\Big]\, ,
\end{equation}
cf. \eqref{pq_CG_choice}, to obtain the contribution of order $O(1/p)$.
\end{skproof}

Thereby we conclude Step 1 if we evaluate $\displaystyle{\lim_{p\to\infty}}f_{n,m}(p)$. For the sake of readability, we set
\begin{eqnarray}
    x_n\equiv x_n[p] &:=& \begin{bmatrix}
    1 & n & n\\
    \vb*0_1,0 & \vb*0_n, 0 & \vb*0_n, 0 
    \end{bmatrix}\!\![\vb* p] \ , \\
    y_n \equiv y_n[p] &:=& \begin{bmatrix}
    1 & n & n+1\\
    \vb*0_1,0 & \vb*0_n, 0 & \vb*0_{n+1}, 0 
    \end{bmatrix}\!\![\vb* p] \ .
\end{eqnarray}
By taking Hermitian cojugate, we get
\begin{equation}
    y_{n-1}[p] = \begin{bmatrix}
    1 & n & n-1\\
    \vb*0_1,0 & \vb*0_n, 0 & \vb*0_{n-1}, 0 
    \end{bmatrix}\!\![\vb* p]\, ,
\end{equation}
so we only need to determine the values of $(x_n[p])_{n< p}$ and $(y_n[p])_{n<p}$.

\begin{proposition}\label{prop:y_n-expr}
For any $p,n\in \mathbb N$ with $n< p$, the following holds
\begin{equation}
    y_n[p] = (-1)^p\,3\,\dfrac{\sqrt{(n+1)(n+2)}}{2n+3}\sqrt{\dfrac{(p+n+3)(p-n)}{p(p+1)(p+2)(p+3)}}\, .
\end{equation}
\end{proposition}
\begin{skproof}
By definition,
\begin{equation}
    \vb*e(1)\vb*e(n) = (-1)^py_n\vb*e(n+1)+ (...)\, ,
\end{equation}
where we are using the shorthand notation
\begin{equation}
    \vb*e(m)\equiv \vb*e(m;\vb*0_m, 0) \, ,
\end{equation}
and where $(...)$ includes components on $\vb*e(m)$ for $m\in \{n-1,n\}$. Then, applying $T_-^{n+1}$ on both sides of the above expressions, the expression for $y_n$ is obtained using (I.2.26) and \eqref{pq_CG_choice}.
\end{skproof}

\begin{proposition}\label{prop:x_n-expr}
For any $p,n\in \mathbb N$ with $n< p$, the following holds
\begin{equation}
    x_n[p] = (-1)^p\dfrac{2n(n+2)}{(2n+1)(2n+3)}\dfrac{2p+3}{\sqrt{p(p+1)(p+2)(p+3)}}\, .
\end{equation}
\end{proposition}
\begin{skproof}
It goes by induction\footnote{It is possible to calculate $y_n[p]$ in a similar manner.}. We have
\begin{equation}\label{e(1)^{n+1}}
    \vb*e(1)^{n+1} = (-1)^{np}G(n)S(n)\vb*e(n)+(...)\, ,
\end{equation}
where $(...)$ includes only components on $\vb*e(m)$ for $m\ne n$ and
\begin{equation}
    G(n) = \prod\limits_{m=1}^{n-1}y_m \ , \ \ S(n) = \sum\limits_{m=1}^nx_m\, .
\end{equation}
From Proposition \ref{prop:y_n-expr}, we get an explicit expression for $G(n)$. By applying $T_-^n$ to \eqref{e(1)^{n+1}} then taking an inner product with $\vb*e(n;(0,2n,n),n/2)$, we obtain that $S(n)$ is proportional to
\begin{equation}\label{S(n)-expr-1}
\begin{aligned}
    \sum_{m=0}^n\ip{\vb*e(n;(0,2n,n),n/2)}{\vb*e(1;(0,2,1),1/2)^m\vb*e(1)\vb*e(1;(0,2,1),1/2)^{n-m}}\, ,
\end{aligned}
\end{equation}
with coefficient of proportionality determined by $G(n)$ and (I.2.26). Each term in the sum is
\begin{equation}\label{ip_trace}
    \begin{aligned}
        & \ip{\vb*e(n;(0,2n,n),n/2)}{\vb*e(1;(0,2,1),1/2)^m\vb*e(1)\vb*e(1;(0,2,1),1/2)^{n-m}}\\
        & \hspace{5 em} = \sqrt{\dfrac{2}{3}}\dfrac{\tr((2T_3+U_3)T_-^{n-m}T_+^{n-m}T_+^mT_-^m)}{\mu_n(p)(\mu_1(p))^{n+1}}\, .
    \end{aligned}
\end{equation}
The basis given in (I.D.1) diagonalizes the operators $2T_3+U_3$, $T_-^{n-m}T_+^{n-m}$ and $T_+^mT_-^m$, so it can be used to calculate the trace above more easily. Explicitly, we obtain
\begin{equation}
\begin{aligned}
& \tr((2T_3+U_3)T_-^{n-m}T_+^{n-m}T_+^mT_-^m)\\
& \hspace{3 em} = \dfrac{1}{2(n+1)}\Bigg(3\dfrac{(n+1)!p!}{(p-n-1)!}{}_2F_1(n+2,n+1-p;-p;1)\\
& \hspace{9 em} +\dfrac{n!(p+1)!}{(p-n)!}(3m-p){}_2F_1(n+1,n-p;-p-1;1)\Bigg)\, ,
\end{aligned}
\end{equation}
where ${}_2F_1$ is the hypergeometric function. By the Vandermonde's formula \cite{slater},
\begin{equation}\label{sum_trace}
\sum_{m=0}^n\tr((2T_3+U_3)T_-^{n-m}T_+^{n-m}T_+^mT_-^m) = \dfrac{n}{4}\dfrac{n!(n+1)!}{(2n+3)!}\dfrac{(p+n+2)!}{(p-n)!}(2p+3)\, .
\end{equation}
Putting \eqref{ip_trace} and \eqref{sum_trace} together, we get the desired expression for the summation in \eqref{S(n)-expr-1}. Therefore
\begin{equation}
     S(n) = (-1)^p\dfrac{n(n+1)}{2n+3}\dfrac{2p+3}{\sqrt{p(p+1)(p+2)(p+3)}}\, .
\end{equation}
To finish, we just need the expression for $x_1$, which can be inferred from \eqref{pq_CG_choice}.
\end{skproof}

Propositions \ref{prop:step1}-\ref{prop:x_n-expr} lead straightforwardly to the following lemma.

\begin{lemma}\label{lemma:f_n,m_conv}
The limit
\begin{equation}
    \lim_{p\to\infty}f_{n,m}(p) = \left(\dfrac{2(n+1)}{(m+1)}\right)^{3/2}
\end{equation}
holds for every $n\in\mathbb N$ and $m\in\{n-1,n,n+1\}$ with order $O(1/p)$.
\end{lemma}

A simple examination of Theorem I.4.5 and \eqref{X1_starXn} in view of Lemma \ref{lemma:f_n,m_conv} gives the following theorems.

\begin{theorem}\label{theo:step1}
    If $c_n^p\to 1$ as $p\to \infty$ for all $n\ge 1$, then the uniform convergence $f_1\star^p f_2\to f_1f_2$ holds for every pair $f_1\in \mathcal X_1$ and $f_2\in \mathcal X$.
\end{theorem}

\begin{theorem}\label{theo:step1_sw}
The twisted products $(\ast^p_{\!S})$ induced by the symmetric Stratonovich-Weyl correspondences are such that $\norm{f_1\ast^p_{\!S} f_2-f_1f_2}\in O(1/p)$ as $p\to\infty$,  for every pair $f_1\in \mathcal X_1$ and $f_2\in \mathcal X$.
\end{theorem}

Now, let $>_m$ denote the highest weight of $(m,m)$, so $X^1_{>_1}X^n_{>_n}$ is a non zero multiple of $X^{n+1}_{>_{n+1}}$

\begin{theorem}\label{theo:tp_c1_conv_cn_conv}
    Suppose the uniform convergence $f_1\star^p f_2\to f_1f_2$ holds for every pair $f_1\in \mathcal X_1$ and $f_2\in \mathcal X$. If $c_1^p\to 1$ as $p\to \infty$, then $c_n^p\to 1$ for all $n\ge 1$.
\end{theorem}
\begin{skproof}
    For $n\in \mathbb N$,
    \begin{equation}
    X^1_{>_1}\star^p X^n_{>_n} = \dfrac{c_{n+1}}{c_1c_n}f_{n,n+1}(p)\cg{1}{n}{n+1}{>_1}{>_n}{>_{n+1}}\cg{1}{n}{n+1}{\vb*0_1 0}{\vb*0_n 0}{\vb*0_{n+1} 0}X^{n+1}_{>_{n+1}}\, .
    \end{equation}
    The statement follows by induction, using Theorem I.4.5 and Lemma \ref{lemma:f_n,m_conv}.
\end{skproof}

\begin{corollary}\label{cor:pois_and_c_1->1}
    If the characteristic numbers $(c_n^p)$ define a sequence of correspondences of Poisson type and $c_1^p\to 1$ as $p\to\infty$, then $c_n^p \to 1$ for every $n\ge 1$.
\end{corollary}

\ 

\noindent\underline{\bf Step 2}. \ 
We'll proceed by induction from Theorem \ref{theo:step1}. Given an harmonic function $X^n_{\vb*\nu, J}$ and $p,m\in \mathbb N$ with $p>\max\{n,m\}$, let
\begin{equation}
\begin{aligned}
    & \hspace{5 em}L^{n,m}_{\vb*\nu, J}[p],R^{n,m}_{\vb*\nu, J}[p]:\mathcal X_{m}\to \mathcal X_{n+m}\, ,\\
    & L^{n,m}_{\vb*\nu, J}[p](f) = X^n_{\vb*\nu, J}\star^p f \ , \ \ R^{n,m}_{\vb*\nu, J}[p](f) = f\star^p X^n_{\vb*\nu, J}\, ,
\end{aligned}
\end{equation}
be the left and right star product operators, respectively.

\begin{lemma}\label{lemma:L-R-bound}
If all characteristic numbers converge to $1$ as $p\to \infty$, then the families of operators $(L^{n,m}_{\vb*\nu, J}[p])_p$ and $(R^{n,m}_{\vb*\nu, J}[p])_p$ are uniformly bounded for every $n,m\ge 1$.
\end{lemma}
\begin{skproof}
It follows from Theorem I.C.3 and equation (I.C.6).
\end{skproof}

\begin{theorem}\label{theo:tp-conv}
If $c_n^p\to 1$ as $p\to \infty$ for all $n\ge 1$, then the uniform convergence
$f_1\star^p f_2\to f_1f_2$ holds for every pair $f_1,f_2 \in \mathcal X$.
\end{theorem}
\begin{skproof}
    Assume that, for $n\in \mathbb N$, $f_1\star^p f_2\to f_1f_2$ whenever $f_1 \in \mathcal X_n$ and $f_2 \in \mathcal X$. Every element of $\mathcal X_{n+1}$ is a linear combination of an element of $\mathcal X_n$ and pointwise products of the form $X^n_{\vb*\nu, J}X^1_{\vb*\mu, I}$, so it is sufficient to prove
\begin{equation}
    (X^n_{\vb*\nu, J}X^1_{\vb*\mu, I})\star^p X^{n'}_{\vb*\nu', J'}\to X^n_{\vb*\nu, J}X^1_{\vb*\mu, I} X^{n'}_{\vb*\nu', J'}\, .
\end{equation}
The idea is to sum and subtract $X^n_{\vb*\nu, J}\star^p (X^1_{\vb*\mu, I}X^{n'}_{\vb*\nu',J'})$ and $X^n_{\vb*\nu, J}\star^p X^1_{\vb*\mu, I}\star^p X^{n'}_{\vb*\nu',J'}$, then use triangular inequality and Lemma \ref{lemma:L-R-bound} to conclude what we want.
\end{skproof}

\ 

\noindent\underline{\bf Step 3}. \ 
To estimate the rate of convergence of $\norm{[f_1,f_2]_{\star^p}}$ when the characteristic numbers all go to $1$, the symmetric Stratonovich-Weyl correspondence is a suitable reference. So let $(\ast^p_{\!S})$ be the twisted products induced by the symmetric Stratonovich-Weyl correspondences.

\begin{theorem}\label{theo:sw_comm_O(1/p)}
For every $f_1,f_2 \in \mathcal X$, we have $\norm{[f_1,f_2]_{\ast^p_S}} \in O(1/p)$.
\end{theorem}
\begin{skproof}
    It follows straightforwardly from Theorem \ref{theo:step1_sw}.
\end{skproof}

\begin{theorem}\label{theo:comm_O(1/p)}
If $c_n^p\to 1$ as $p\to \infty$ for every $n\ge 1$, then $\norm{[f_1,f_2]_{\star^p}}\in O(1/p)$ for every $f_1,f_2\in \mathcal X$.
\end{theorem}
\begin{skproof}
For $n_1,n_2 \in \mathbb N$, the idea is to compare
\begin{equation*}
\norm{[X^{n_1}_{\vb*\nu_1,J_1},X^{n_2}_{\vb*\nu_2, J_2}]_{\ast^p_{\!S}}} \ \ \ \ \mbox{and} \ \ \ \  \norm{[X^{n_1}_{\vb*\nu_1,J_1},X^{n_2}_{\vb*\nu_2, J_2}]_{\star^p}}
\end{equation*}
using the norm given by the maximum of coordinates with respect to the basis of harmonic functions as intermediate. Just note that any two norms on $\mathcal X_{n_1+n_2}$ are equivalent since it is finite dimensional, and the hypothesis on the characteristic numbers implies that there is $C(n_1,n_2)>0$ such that
\begin{equation}
    \left|\dfrac{c_n^p}{c_{n_1}^pc_{n_2}^p}\right|\le C(n_1,n_2)
\end{equation}
for every $n\le n_1+n_2$.
\end{skproof}

\ 

\noindent\underline{\bf Step 4}. \ 
The commutator $[X^1_{\vb*\mu, I},X^n_{\vb*\nu, J}]_{\star^p}$ can be explicitly computed.

\begin{proposition}\label{prop:comm-X1-Xn-asymp}
For any two $\mathbb CP^2$ harmonics $X^1_{\vb*\mu, I},X^n_{\vb*\nu, J}\in \mathcal X$, we have
	\begin{equation}
		\left[X^1_{\vb*\mu, I},X^n_{\vb*\nu, J}\right]_{\star^p} = \dfrac{1}{p\sqrt{1+3/p}}\dfrac{i}{c_1^p}\sqrt{\dfrac{3}{2}}\{X^1_{\vb*\mu, I}, X^n_{\vb*\nu, J}\}\, .
	\end{equation}
In particular, $p[f_1,f_2]_{\star^p}\to i\sqrt{3/2}\{f_1,f_2\}$ uniformly for every $f_1\in \mathcal X_1$ and $f_2 \in \mathcal X$ if and only if $c_1^p \to 1$ as $p\to \infty$.
\end{proposition}
\begin{skproof}
Let $A = \vb*e(1;\vb*\mu, I)$. By definition of twisted product, and with a little abuse of notation,
\begin{equation}
	\begin{aligned}
		\left[X^1_{\vb*\mu, I},X^n_{\vb*\nu, J}\right]_{\star^p} & \overset{W^p}{\longleftrightarrow} \dfrac{\dim(\vb*p)}{c_1^pc_n^p\mu_1(p)}[A,\vb*e(n;\vb*\nu, J)]\\
		& \overset{W^p}{\longleftrightarrow} \dfrac{1}{c_1^p}\dfrac{2\sqrt{3}}{\sqrt{p(p+3)}} A(X^n_{\vb*\nu, J})
	\end{aligned}
\end{equation}
The result follows from \eqref{pois_bi_vec_op} and \eqref{X^1_poly} by straightforward calculation.
\end{skproof}

\ 

\noindent\underline{\bf Step 5}. \ 
Once more, it goes by induction, where now the base step is Proposition \ref{prop:comm-X1-Xn-asymp}. The next proposition contains the inductive step.

\begin{proposition}\label{prop:tc_ind}
Suppose $f\star^p g\to fg$ uniformly for every $f,g \in \mathcal X$. For $n\in \mathbb N$, if the uniform convergence $p[f,g]_{\star^p}\to\,i\sqrt{3/2}\{f,g\}$ holds for every pair $f\in \mathcal X_n$ and $g\in \mathcal X$, then $p[f,g]_{\star^p}\to \,i\sqrt{3/2}\{f,g\}$ for every $f \in \mathcal X_{n+1}$ and $g \in \mathcal X$.
\end{proposition}
\begin{skproof}
Analogously to Theorem \ref{theo:tp-conv}, it is sufficient to prove
\begin{equation}
p\left[X^1_{\vb*\mu, I}X^n_{\vb*\nu, J},X^{n'}_{\vb*\nu', J'}\right]_{\star^p}\to i\sqrt{\dfrac{3}{2}}\left\{X^1_{\vb*\mu, I}X^n_{\vb*\nu, J},X^{n'}_{\vb*\nu', J'}\right\}\, ,
\end{equation}
and it can be done by a serial sum and subtraction of suitable terms, resorting to Theorems \ref{theo:tp_c1_conv_cn_conv} and \ref{theo:comm_O(1/p)}, Lemma \ref{lemma:L-R-bound} and the uniform boundedness principle.
\end{skproof}

Now, putting the above proposition together with Theorem \ref{theo:tp-conv} and Proposition \ref{prop:comm-X1-Xn-asymp}, we finally obtain:

\begin{theorem}\label{theo:c_n->1_pois}
If $c_n^p\to 1$ as $p\to \infty$ for all $n\ge 1$, then $p\left[f,g\right]_{\star^p}\to i \sqrt{3/2}\left\{f,g\right\}$ uniformly for every $f,g \in \mathcal X$.
\end{theorem}
\end{document}